\documentclass[12pt,leqno]{article}

\usepackage{amssymb,latexsym}
\usepackage{mathrsfs}
\usepackage[pdftex]{graphicx}
\usepackage[left=1in,right=1in,top=1in, bottom=1in]{geometry}
\usepackage{chicago}
\usepackage{amsmath}
\usepackage{graphicx}

\usepackage{caption} 
\usepackage{enumerate}
\usepackage{subcaption}
\usepackage{setspace}
 \usepackage{xcolor}
  \usepackage{booktabs,caption}
 \usepackage[flushleft]{threeparttable}

\usepackage[hidelinks]{hyperref}

\long\def\symbolfootnote[#1]#2{\begingroup%
\def\thefootnote{\fnsymbol{footnote}}\footnote[#1]{#2}\endgroup}

\newcommand{\R}{\ensuremath{\mathbb{R}}}

\newcommand{\eps}{\varepsilon}  
\newcommand{\1}{\ensuremath{\mathbf{1}}}  

 \newcommand\Tau{\mathcal{T}}

\usepackage{amsthm}

\theoremstyle{definition} 
\newtheorem{example}{Example}
\newtheorem{definition}{Definition}

\newtheorem{assumption}{Assumption}

\newtheorem{remark}{Remark}

\newtheorem{assumptionp}{Assumption}

\DeclareMathOperator*{\argmax}{arg\,max}
\DeclareMathOperator*{\argmin}{arg\,min}

\theoremstyle{theorem} 
\newtheorem{theorem}{Theorem}
\newtheorem{lemma}{Lemma}

\newtheorem{corollary}{Corollary}

\begin{document}
\bibliographystyle{chicago}

\begin{center}
	
	\quad
	\vspace{-10mm}

	\Large{\bf{Normalizations and misspecification \\in skill formation models}}\symbolfootnote[1]{Thanks to Antonia Antweiler, Gabriella Conti, Philipp Eisenhauer, Janos Gabler, Hans-Martin von Gaudecker, Lena Janys, Julius Kappenberg, Philipp Ketz, Josh Kinsler, Matt Masten, Mauricio Olivares, Ronni Pavan, Jack Porter and seminar participants at TSE, UCL, PSE, LMU, Zurich, Luxembourg, Tsinghua, UCSD, McMaster, Yale, and the German Economic Association for helpful comments and to Matt Wiswall for many useful discussions. I also thank five referees and the editor for constructive comments that helped to improve the paper. The research was supported by the European Research Council (ERC-2020-STG-949319).}

	\vspace{5mm}

	\normalsize

 	Joachim Freyberger\symbolfootnote[3]{University of Bonn. Email: freyberger@uni-bonn.de. } 
 	
	\vspace{5mm}
	
	 August 2025\symbolfootnote[2]{First version: April 10, 2020}

	\vspace{0.3cm}

\end{center}

\normalsize

\noindent \textbf{Abstract} \medskip

\noindent  An important class of structural models studies the determinants of skill formation and the optimal timing of interventions. In this paper, I provide new identification results for these models and investigate the effects of seemingly innocuous  scale and location  restrictions on parameters of interest. To do so, I first characterize the identified set of all parameters without these additional restrictions and show that important policy-relevant parameters are point identified under weaker assumptions than commonly used in the literature. The implications of imposing standard scale and location restrictions depend on how the model is specified, but they generally impact the interpretation of parameters and may affect counterfactuals. Importantly, with the popular CES production function, commonly used scale restrictions fix identified parameters and lead to misspecification. Consequently, simply changing the units of measurements of observed variables might yield ineffective investment strategies and misleading policy recommendations. I show how existing estimators can easily be adapted  to solve these issues. As a byproduct, this paper also presents a general and formal definition of when restrictions are truly normalizations.

\vspace{0.5cm}



\newpage
 
\section{Introduction}

Structural models are key tools of economists to simulate changes in the economic environment and evaluate and design policies. An important class of such models deals with skill and human capital formation. Human capital formation is a main part of many structural models and is an important driver of economic growth and inequality \cite{MT:16}, making policies that target skill formation particularly vital. This growing literature, originating from the seminal papers of \citeN{CH:08} and \citeN{CHS:10}, estimates production functions of various skills of children, studies how past skills, parental skills, and investments affect future skills, and links the skills to adult outcomes. The results provide valuable insights into determinants of skill formation, timing of investments, and the design of optimal interventions for disadvantaged children.  

A major challenge in these models is that skills are not directly observable, lack a natural scale, and can only be approximated through measurements, such as test scores. The early literature has provided sufficient conditions for identification of parametric and nonparametric versions of these complex models, which is often achieved using a two-step approach. First, the distribution of skills is identified from the measurements, and the production function is then identified in the second step (see e.g. \shortciteN{CHS:10}, \citeN{AMNS:17}, and \citeN{AMN:19}). To obtain point identification in the first step, it is necessary to fix the unknown scales and locations of skills. However, it remains unclear whether these restrictions are still required when combined with common parametric assumptions on the production function in the second step.

This paper presents a new identification analysis for skill formation models and investigate the consequences of seemingly innocuous normalizations. Instead of providing sufficient conditions for point identification using multi-step arguments, I start by pooling all parts of the model and characterize the identified set of all parameters without the scale and location restrictions. This approach reveals which parameters are point identified or partially identified and how the additional restrictions affect the identified set. Notably, I show that many critical features of these models are invariant to scale and location restrictions, point-identified without them, and identified under less restrictive assumptions than previously considered. These results apply to both parametric and nonparametric versions of the model.

The exact implications of imposing scale and location restrictions on parameters and counterfactuals depend on the model specification and the object of interest. Specifically, a restriction could be a harmless normalization with one production function, but could impose strong assumptions with a different one. I thus analyze two popular specifications, namely the trans-log and the CES production functions. For the trans-log case, it turns out that standard scale and location restrictions do not impose additional testable restrictions, simply select an element of the identified set, and many objects of interest are invariant to them. However, they still affect production function parameters and certain counterfactuals. Since such restrictions are often arbitrary and tied to the units of measurement of the data, these features can be hard to interpret, difficult to compare across different studies, and potentially mislead policy recommendations. More importantly, for the CES case, I show that commonly restricted scale parameters are in fact identified and setting them to specific values is typically inconsistent with the data. Consequently, with these restrictions, simply changing the units of a skill measure (e.g. from years to months) can impact estimated dynamics, persistence of skills, effects of parental investments, and optimal investment strategies.   

In addition to these identification results, I demonstrate how existing estimators can be modified to estimate all identified parameters. While some parameters remain difficult to interpret, the modified estimator enables the calculation of point identified features and counterfactuals that are invariant to scale and location restrictions and to the units of measurement. I illustrate these results through Monte Carlo simulations and an empirical application based on the framework and estimator of  \shortciteN{AMN:19}.

A broader takeaway is that researchers should carefully verify whether imposed restrictions are in fact normalizations (as formally defined in Definition \ref{d:normalization}). This is particularly crucial in structural models where identification often proceeds in multiple steps, and supposedly innocuous restrictions in one step may unintentionally affect results in subsequent steps. By focusing on features of the model that are invariant to these restrictions, researchers can ensure that their findings remain interpretable and comparable across different  studies.

\textbf{Literature:} Following the influential work of \citeN{CH:08} and \shortciteN{CHS:10}, a growing body of literature has studied the development of latent variables. For example, \citeN{HP:11} study the determinants of children's cognitive and non-cognitive skills using Indian data, \citeN{FK:14} investigate how time allocation affects both cognitive and non-cognitive development using Australian data, \shortciteN{AMN:19} and \shortciteN{AMNS:17} estimate the effects of health and cognition on human capital with data from India and Ethiopia/Peru, respectively, and \shortciteN{ACDMR:19} study which intervention led to gains in cognitive and socio-emotional skills using Colombian data. This literature provides important insights into the nature of persistence, dynamic complementarities, and the optimal targeting of interventions.\footnote{See also \citeN{CH:07}, \citeN{CH:09}, \citeN{Cunha:11}, \citeN{HPS:13}, \citeN{AJ:16}, \citeN{HAP:17}, \citeN{ACM:2022}  and references therein.}

In these studies, the measurement system has a factor structure that requires scale and location restrictions to identify the distribution of latent variables. In particular, \citeN{CH:08} study a model where the production function is log-linear. They utilize the two-step identification arguments described above, except that they impose the scale and location restrictions using an adult outcome to obtain well-defined units of measurement (``anchoring''). They also discuss which parameters depend on the anchor or its units of measurement. In this model, I demonstrate that certain policy-relevant parameters, such as optimal investment sequences, may depend on the specific scales and units of measurements of the observed variables. I further derive a set of features that is identified under weaker restrictions on the measures and the production function and without anchoring the skills or fixing their scales. Importantly, I also provide identification results for alternaive production technologies, such as the widely used CES production function, where standard restrictions are unnecessary for identification. In these case, imposing them through an initial period normalization or through anchoring leads to misspecification.

Common restrictions are to fix the scales and locations of each latent factor in each time period by setting parameters in the measurement system.\footnote{For example, \citeN{CH:08} and \shortciteN{CHS:10} set the scale of the first skill measure to 1 in each period - see the discussion around equations (7) -- (8) of \citeN{CH:08} and the first paragraph on page 891 of \shortciteN{CHS:10}. In addition, as discussed in footnote 17 of \citeN{CH:08}, their identification results rely on either setting one of the location parameters to $0$ in each time period or setting the mean of the skill to $0$ in each time period. These restrictions have also been used in subsequent papers such as \shortciteN{AMN:19}. The nonparametric identification results in \shortciteN{CHS:10}  impose analogous restrictions in their Assumption (v) of Theorem 2.} If the same values are used across all time periods, Agostinelli and Wiswall (2016a, 2016b, 2024) refer to this as an age-invariance assumption (see Assumption \ref{a:ageinvariant_technology_skills}(a)  for a formal definition or Definition 1 of  Agostinelli and Wiswall (2024)). In important contributions,  Agostinelli and Wiswall (2016a, 2016b, 2024) demonstrate that imposing this assumption can yield a misspecified model when the production function has a known scale and location. They also propose relaxations of production functions with age-invariant measures and show that, even without age-invariance, production function restrictions can yield point identification (see Corollary \ref{c:pointident} below and the related discussion). However, in both cases, they impose scale and location restrictions in the first period, arguing that these are necessary for point identification. While this is true for the trans-log production function used in their application, I show that the scale restriction is not required for the CES production function. Furthermore, I demonstrate for different specifications that while production function parameters, certain counterfactuals, and estimated dynamics may depend on specific scales and units of measurements, many key parameters are invariant to these restrictions, including age-invariance, and are in fact point identified without them.  

\nocite{AW:16a}\nocite{AW:16b}\nocite{AW:22}

In independent research, \citeN{DKP:20} show that with a trans-log production function, anchored treatment effects are invariant to scale and location restrictions and are identified without age-invariance.  Through simulations, they also show that standard restrictions lead to inconsistent estimated treatment effects with the CES production function. In the trans-log case, these results aligns with part 4 of Theorem \ref{th:identfunctions} below. While their proof is specific to the trans-log production function with a log-linear measurement system, my results extend to other cases. Additionally, I show that only skill measures in the first period are needed to identify anchored treatment effects, reducing the data requirements and assumptions considerably. I also consider other policy-relevant features. Additionally, for the CES case, I show why standard restrictions can cause misspecification and describe how existing estimators can be adapted to resolve this issue.

Table 1 of \shortciteN{ACM:2022} provides a selective overview of specifications and estimation methods used in well-cited papers in the literature. It shows that with the exception of \shortciteN{CHS:10}, earlier work predominantly employed the Cobb-Douglas specification, likely due to the simplicity of the corresponding estimator.  Since \shortciteN{AMN:19}, there has been increasing adoption of the (nested) CES production function, facilitated by their computationally simple two-step estimator (compared to the MLE of \shortciteN{CHS:10}). Next to the papers in this table, recent contributions using this approach include \shortciteN{ANT:23}, \shortciteN{BFHD:24},  and \citeN{GG:23}. The trans-log production function proposed by \citeN{AW:22}  is nonnested with the CES specification and can also be estimated using an IV/GMM approach.

The consequences of normalizations have been discussed in various contexts. In factor models, while certain restrictions are necessary for point identification (see e.g. \citeN{AR:56} or \citeN{Madansky:64}), \citeN{Williams:20} shows that certain features, such as variance decompositions, can be identified without them. I combine a factor model with a production function, which provides additional restrictions. Many studies argue and show in specific examples that critical features should not depend on normalizations, see e.g.  \citeN{Freyberger:18} and \shortciteN{KSSS:18}. Similar to this paper, but in a very different context, \citeN{AS:14} discuss restrictions that were considered normalizations, but are restrictive assumptions. \shortciteN{KSS:20} show that certain counterfactuals in dynamic discrete choice models are identified, even when the model itself is not.  \shortciteN{RWZ:10} define a normalization in vector autoregressive models to pin down unidentified signs; see end of Section \ref{s:normal} for more details. Matzkin (1994, 2007) discusses several examples of normalizations, some of which are motivated by economic theory. \citeN{Lewbel:19} provides an informal discussion of normalizations, which is conceptually very similar to the formal definition I provide below. When normalizing restrictions are needed for point identification, there are often multiple ways to impose them when estimating the model. Good choices can then yield particularly convenient restrictions on the parameter space (as in \citeN{GL:19}) or even faster rates of convergence (as in \shortciteN{CKK:15}). See also \shortciteN{HWZ:07} for a discussion on estimation with normalizations. 
		
		\nocite{Matzkin:94}\nocite{Matzkin:07}

\textbf{Structure:} In Section \ref{s:normal}, I provide a formal definition of a normalization, which to the best my knowledge currently does not exist in the literature, as well as illustrative examples. Section \ref{s:skillform} contains the identification analysis of different parametric skill formation models. A nonparametric version is discussed in Appendix \ref{s:genident}. Sections \ref{s:montecarlo} and \ref{s:application} contain the Monte Carlo simulations and the empirical application, respectively. All proofs are in the appendix.

\section{Normalizations}
\label{s:normal}

I begin by providing a formal definition of a normalization, which serves as a basis for the subsequent analysis. I illustrate the definition and potential problems using a probit model and a simple version of the skill formation model.

\subsection{General definition and illustration}
\label{s:normal_def}

Suppose we have a model where $\tau_0 \in \Tau$ denotes the true values of the parameters and $\Tau$ is the parameter space. Here $\tau_0$ could be the coefficients in a regression model or the parameters in a skill formation model. If the model is semiparametric, $\tau_0$ could also contain unknown functions. Let $Z$ contain all observed random variables, such as $Y$ and $X$, with distribution $P(Z)$. For any $\tau \in \Tau$, the model generates a joint distribution of the data $Z$, denoted by $P(Z,\tau)$. Since the model is assumed to be  correctly specified, the true distribution of $Z$ is $P(Z,\tau_{0})$.  The model typically contains certain assumptions, such as functional form or independence assumptions, but suppose that so far none of the normalizations are imposed. The identified set for $\tau_{0}$ is $\Tau_0 = \{ \tau \in \Tau: P(Z,\tau) = P(Z,\tau_0) \}.$  If  $\Tau_0$ is a singleton, $\tau_0$ is point identified. We say that $\tau_1, \tau_2 \in \Tau$ are observationally equivalent if they generate the same distribution of the data: $P(Z,\tau_1) = P(Z,\tau_2)$. Let $g(\tau_0)$ be a function of interest, such as a counterfactual.  The identified set for $g(\tau_{0})$ is $\Tau_{g_0} = \{ g(\tau): \tau \in \Tau_0 \}.$
Notice that $g(\tau_0)$ could be point identified (i.e.  $\Tau_{g_0}$ is a singleton) even if $\tau_0$ is not.

In models with normalizations, $\Tau_0$ is typically not a singleton. A normalization is a restriction of the form $\tau \in \Tau_N$, where $\Tau_N \subseteq \Tau$ is a known set. Hence, a normalization restricts the feasible values of $\tau$, such as setting an element to $1$. I define a restriction to be a normalization with respect to a function $g(\tau_0)$ if it does not change the identified set of $g(\tau_0)$. 
\begin{definition}
	\label{d:normalization}
	The restriction $\tau \in \Tau_N$ is a \textit{normalization} with respect to $g(\tau_0)$ if $\{g(\tau): \tau \in \Tau_0 \cap \Tau_N \} = \{g(\tau): \tau \in \Tau_0 \}$ for all $\tau_0 \in \Tau$.
\end{definition}

Typically, $\Tau_0 \cap \Tau_N$ is a singleton. That is, we achieve point identification with the additional restrictions. Definition \ref{d:normalization} then implies $\{g(\tau_0)\} = \{g(\tau): \tau \in \Tau_0 \cap \Tau_N \} = \{g(\tau): \tau \in \Tau_0 \}$ and thus, that $g(\tau_0)$ is point identified, even without the restriction $\tau \in \Tau_N$. Since these restrictions are often arbitrary, $\tau_0$ is usually not in $ \Tau_N$ in which case the restriction $\tau \in \Tau_N$ is not a normalization with respect to $\tau_0$, but it can be a normalization with respect to particular functions of interest. Moreover, the restriction can be a normalization for some function and not for others. Hence, researchers need to argue that normalizations hold with respect to all functions of interest, such as all counterfactuals. Finally, a normalization cannot impose any additional overidentifying restrictions in the sense that  if $\Tau_0 \neq \emptyset$, then $\Tau_0 \cap \Tau_N \neq \emptyset$.

As a simple example, consider the probit model where $Y = \1(\beta_{0,1} + \beta_{0,2}X \geq U)$, $var(X) > 0$, $U \mid X \sim N(\mu_0, \sigma_0^2)$ and $\sigma_0^2 > 0$. The true parameter vector is $\tau_0 = (\beta_{0,1}, \beta_{0,2},\mu_0,\sigma_0)'$ and $Z = (Y,X)$. Now notice that
$$P(Y=1 \mid X = x) = \Phi \left( \frac{\beta_{0,1}-\mu_0}{\sigma_0} +  \frac{\beta_{0,2}}{\sigma_0}x\right),$$
where $\Phi$ denotes the standard normal cdf. Since $var(X) > 0$, $\frac{\beta_{0,1}-\mu_0}{\sigma_0}$ and $\frac{\beta_{0,2}}{\sigma_0}$ are point identified. It is also well known and easy to see that
$$\Tau_0 = \left\{\tau \in \R^3 \times \R_{>0}: \frac{\beta_{1}-\mu}{\sigma} = \frac{\beta_{0,1}-\mu_0}{\sigma_0} \text{ and } \frac{\beta_{2}}{\sigma}= \frac{\beta_{0,2}}{\sigma_0}    \right\}$$ 
because all values in $\Tau_0$ imply the same joint distribution of $(Y,X)$.

Since $\tau_0$ is not point identified, it is common to set $\mu = 0$ and $\sigma = 1$. Using the previous notation, this means that $\Tau_N = \R^2 \times 0 \times 1$ and $\Tau_0 \cap \Tau_N = \left(\frac{\beta_{0,1}-\mu_0}{\sigma_0} , \frac{\beta_{0,2}}{\sigma_0} , 0, 1\right).$ Clearly, this restriction is not a normalization with respect to $\beta_{0,1}$ or $\beta_{0,2}$, which are typically not objects of interest. In fact, in general $\tau_0 \notin \Tau_0 \cap \Tau_N$ unless $\mu_0 = 0$ and $\sigma_0 = 1$. However, this restriction is a normalization with respect to (potentially counterfactual) probabilities
$$ P(Y = 1 \mid X = x) = \Phi \left( \frac{\beta_{0,1}-\mu_0}{\sigma_0} +  \frac{\beta_{0,2}}{\sigma_0}x\right) $$
or, when $X$ is continuous, marginal effects $\frac{\partial  }{\partial x} P(Y = 1 \mid X = x)$.
Clearly, these features are point identified even though $\tau_0$ is not. 

Normalizations may help to provide useful structural interpretations of other parameters. As an example, suppose there are two covariates, 
$Y = \1(\beta_{0,1} + \beta_{0,2}X_1  + \beta_{0,3}X_2  \geq U)$, and $\beta_{0,2} > 0$ . Instead of setting $\mu = 0$ and $\sigma = 1$, we could impose $\beta_{0,1} = 0$ and $\beta_{0,2} = 1$, which are also normalizations with respect to marginal effects. In addition, we then obtain a model of the form $Y = \1( X_1  + \tilde{\beta}_{0,3}X_2  \geq \tilde{U})$ with $ \tilde{\beta}_{0,3} =  \beta_{0,3}/\beta_{0,2}$, which can then be interpreted as a relative effect (which is also identified without scale and location restrictions).  

In the context of vector autoregressive models, \shortciteN{RWZ:10} define a normalization as a restriction on the parameter space that pins down unidentified signs of coefficients. These restrictions are imposed in addition to other assumptions, such as long run restrictions. Just like above, their restrictions do not impose additional testable assumptions and they can help to provide structural interpretations of certain parameters.  Unlike my definition, it is not clear from theirs whether these restrictions are without loss of generality in the sense that they do not affect functions of interest. If they do, one would have to argue why they are reasonable.

\subsection{Simple skill formation model }
\label{s:examples}

As a more involved example, I now discuss a very simple skill formation model, which imposes very restrictive assumptions, but illustrates the previous definition and points out the types of problems that occur in more general models. 

Let $\theta_{t}$ denote skills at time $t$ and let $I_t$ be investment at time $t$, where $t = 0,1,\ldots,T$. We are interested in the roles of investment and past skills in the development of future skills.  Instead of skills, the data only contains measurements of them, denoted by $Z_{\theta,t,m}$. In this section, I first consider the simplest possible model without measurement error and a single measure $Z_{\theta,t,1}$ that takes the form $Z_{\theta,t,1}  = \lambda_{\theta,t,1} \ln \theta_{t}$,
where $\{\lambda_{\theta,t,1}\}^T_{t=0}$ are unknown parameters with $\lambda_{\theta,t,1} \neq 0$ for all $t$. In other words, in each period, we observe a scaled version of  log-skills. I also  assume that there are three periods ($T=2$) and that investment is observed. Finally, for simplicity, I parameterize the marginal distribution of skills in the initial period as $\ln \theta_{0} \sim N(0, s_0^2)$, which implies that $E[Z_{\theta,0,1}] = 0$ and $Var(Z_{\theta,0,1}) = \lambda_{\theta,0,1}^2 s_0^2$.

I consider two simple production functions without unobserved random variables. First suppose skills evolve based on the Cobb-Douglas production function: 
$$\ln \theta_{t+1} = a_t + \gamma_{1t}\ln \theta_{t} + \gamma_{2t} \ln I_{t}.$$
Define the vector containing the true values of all ten parameters as
$$\tau_0 = (\tau_{0,1}, \tau_{0,2}, \ldots, \tau_{0,10}) = (   \lambda_{\theta,0,1},  \lambda_{\theta,1,1},  \lambda_{\theta,2,1}, a_0, a_1, \gamma_{10},  \gamma_{20}, \gamma_{11},  \gamma_{21},  s^2_0  ).$$ 
Without further assumptions, $\tau_0$ is not point identified. To see the restrictions imposed by the obervables, notice that $\ln \theta_{t} = Z_{\theta,t,1}/ \lambda_{\theta,t,1} $ and therefore
$$Z_{\theta,t+1,1} = \lambda_{\theta,t+1,1} a_t  + \frac{\lambda_{\theta,t+1,1}}{\lambda_{\theta,t,1}} \gamma_{1t}Z_{\theta,t,1} + \lambda_{\theta,t+1,1}\gamma_{2t} \ln I_{t}.$$
Hence, the joint distribution of $\{Z_{\theta,t+1,1},Z_{\theta,t,1}, I_{t}\}^1_{t=0}$ identifies (a) the scaled intercepts of the production function $ \lambda_{\theta,t+1,1} a_t$, (b) the scaled slope coefficients  $\frac{\lambda_{\theta,t+1,1}}{\lambda_{\theta,t,1}} \gamma_{1t}$ and $\lambda_{\theta,t+1,1}\gamma_{2t}$, and (c) the scale variance $ \lambda_{\theta,0,1}^2 s_0^2$. The identified set can be shown to consist of all parameters that imply the same values of these identified features as the true parameters. Formally, 
\begin{align*}
	\Tau_0 = \Big\{ \tau \in \Tau : \; &  \tau_2 \tau_{4} = \tau_{0,2} \tau_{0,4}, \tau_3 \tau_{5} = \tau_{0,3} \tau_{0,5}, \\
	&   \frac{\tau_2}{\tau_1} \tau_6 =  \frac{\tau_{0,2}}{\tau_{0,1}} \tau_{0,6}  ,  \frac{\tau_3}{\tau_2} \tau_8 =  \frac{\tau_{0,3}}{\tau_{0,2}} \tau_{0,8},  \tau_2 \tau_{7} = \tau_{0,2} \tau_{0,7}  ,   \tau_3 \tau_{9} = \tau_{0,3} \tau_{0,9},  \\
	&   \tau_{1}^2 \tau_{10} = \tau_{0,1}^2 \tau_{0,10}     \Big\}	
\end{align*}
where the three rows correspond to the three sets of features above. To achieve point identification, we could set $ \lambda_{\theta,t,1} = 1$ for all $t$, in which case the intersection of $\Tau_0$  and the additional restrictions is the singleton $	\big\{ \big(  1,   1,   1,   \tau_{0,2} \tau_{0,4},   \tau_{0,3} \tau_{0,5},   \frac{\tau_{0,2}}{\tau_{0,1}} \tau_{0,6},   \tau_{0,2} \tau_{0,7}, \frac{\tau_{0,3}}{\tau_{0,2}} \tau_{0,8},  \tau_{0,3} \tau_{0,9},  \tau_{0,1}^2 \tau_{0,10} \big) \big\}.	$

As a simple numerical example, suppose $Z_{\theta,t,1} = 12 \ln \theta_{t}$ and
\begin{eqnarray*}
	\ln \theta_{t+1} &=& 0.5 \ln \theta_{t} +  0.5 \ln I_t
\end{eqnarray*}
for all $t$, and $\ln \theta_{0} \sim N(0,1) $. Here $\tau_0 = \left(12,12,12,0,0,0.5,0.5,0.5,0.5,1\right)$. If we set $\lambda_{\theta,t,1} = 1$, even though the true value is $12$, we essentially treat  $ \ln \tilde{\theta}_{t} \equiv Z_{\theta,t,1} = 12  \ln \theta_{t}$ as the skills and the corresponding production function is
$$\ln \tilde{\theta}_{t+1} =  0.5 \ln \tilde{\theta}_{t} +  6 \ln I_t. $$
The intersections of the identified set and set of additional restrictions is then the singleton $\left\{ \left(1,1,1,0,0,0.5,6,0.5,6,144\right) \right\}$
which would be the parameter estimated in practice (instead of $\tau_0$). The coefficients in front of investment are thus hard to interpret. For example, using $\tau_0$ (or setting $\lambda_{\theta,t,1} = 12$) one might conclude that increasing investment by 1\% increases skills by $0.5\%$, but with the  restriction $\lambda_{\theta,t,1} = 1$ that effect changes to $6\%$.  Hence, $\lambda_{\theta,t,1} = 1$ is not a normalization with respect to these parameters. The coefficient in front of log-skills is invariant to the scale restrictions, but only because here $\lambda_{\theta,t,1} = \lambda_{\theta,t+1,1} $ for all $t$. Notice that $Z_{\theta,t,1}$ is simply a scaled version of $\ln \theta_{t} $. If we had an alternative measure with a different scale (say 	$Z_{\theta,t,2}  = \lambda_{\theta,t,2} \ln \theta_{t}$), we would generally obtain different parameters if we replaced $Z_{\theta,t,1}$ with $Z_{\theta,t,2}$. This alternative measure could for example result from changing the units of measurements of  $Z_{\theta,t,1}$, such as using years instead of months of education.  

Even though the production function parameters are not identified, there are potential interpretations that adapt to the units of measurements. For example, the identified parameter $\lambda_{\theta,t+1,1}\gamma_{2t}$  tells us the effect of a one unit increase in $\ln(I_t)$ on skills, measured in the units of $Z_{\theta,t+1,1}$ (see Section \ref{s:application} for a specific example). It can also be shown that
\begin{align*}
  F_{\ln\theta_{t+1}}( a_t + \gamma_{1t}\ln Q_{\alpha}(\theta_{t}) + \gamma_{2t} \ln i_{t}) &= F_{\lambda_{\theta,t+1,1}\ln\theta_{t+1}}( \lambda_{\theta,t+1,1} (a_t +  \gamma_{1t}\ln Q_{\alpha}(\theta_{t}) +  \gamma_{2t} \ln i_{t}))  \\
& \hspace{-8mm} =     F_{Z_{\theta,t+1,1}}\left( \lambda_{\theta,t+1,1}a_t + \frac{\lambda_{\theta,t+1,1}}{\lambda_{\theta,t,1}}\gamma_{1t}  Q_{\alpha}(Z_{\theta,t,1}) + \lambda_{\theta,t+1,1}\gamma_{2t} \ln i_{t}\right)
\end{align*}
where $i_{t}$ is a fixed level of investment, $F_{\ln\theta_{t+1}}$ is the cdf of $\ln\theta_{t+1}$, and $Q_{\alpha}(\theta_{t})$ is the $\alpha$-quantile of $\theta_{t}$. Since the right hand side only depends on identified parameters, similar to marginal effects in the probit model, we can identify how exogenous changes in someones investment affects her rank in the skill distribution at time $t+1$ for a given skill quantile at time $t$.

As another example, first combine the production functions from two periods, write
$$\ln \theta_{2} = a_1 + \gamma_{11}a_0 +  \gamma_{11} \gamma_{10}\ln \theta_{0} +  \gamma_{11}\gamma_{20} \ln I_{0} +  \gamma_{21} \ln I_{1},$$
and  for a given investment sequence $(i_0,i_1)$ define 
$$\ln \theta_{2}(i_0,i_1) = a_1 + \gamma_{11}a_0 +  \gamma_{11} \gamma_{10}\ln \theta_{0} +  \gamma_{11}\gamma_{20} \ln i_{0} +  \gamma_{21} \ln i_{1}$$
We can rewrite this equation in terms of identified features as
\begin{align*}
	\lambda_{\theta,2,1} \ln \theta_{2}(i_0,i_1) &= \lambda_{\theta,2,1} a_1 + \lambda_{\theta,2,1}\gamma_{11}a_0  	+\left( \frac{\lambda_{\theta,2,1}}{\lambda_{\theta,1,1}}  \gamma_{11} \right)\left( \frac{\lambda_{\theta,1,1}}{\lambda_{\theta,0,1}}  \gamma_{10}\right) \lambda_{\theta,0,1}  \ln \theta_{0} \\
	& \quad  +  \left(\frac{\lambda_{\theta,2,1}}{\lambda_{\theta,1,1}} \gamma_{20}\right) \left(\lambda_{\theta,1,1} \gamma_{11}\right) \ln i_{0} +  \left(\gamma_{21}\lambda_{\theta,2,1}\right) \ln i_{1}
\end{align*} 
showing that we can identify the distribution of $\lambda_{\theta,2,1} \ln \theta_{2}(i_0,i_1)$, which we can interpret as a counterfactual measure for a given investment sequence. Using this result, we can also identify the sequence of investment which maximizes expected log-skills. That is, consider $g(\tau_0) \in \R^2$ defined as
\begin{align*}
	g(\tau_0) &= \argmax_{(i_{0},i_{1}) \in \mathcal{I}} E \left[\ln \theta_{2}(i_0,i_1)   \right],  
\end{align*}
where $\mathcal{I}$ is a set of feasible investments. Since the maximizer is invariant to changes in scales and locations of the objective function,  $g(\tau_0) = \argmax_{(i_{0},i_{1}) \in \mathcal{I}} E \left[\lambda_{\theta,2,1} \ln \theta_{2}(i_0,i_1)  \right] $,
showing that $	g(\tau_0)$ is identified. All the identified features above are invariant to the restriction $\lambda_{\theta,t,1} = 1$, implying that $\lambda_{\theta,t,1} = 1$ is a normalization with respect to these features.  

Next to primitive parameters, some counterfactuals reported in applications are not be invariant to $\lambda_{\theta,t,1} = 1$ either. As an example, suppose the production function for $\ln \theta_1$ also includes an additive interaction term, $\gamma_{30}\ln \theta_0 \ln I_0$, in which case it can be shown that
\begin{align*}
	  \lambda_{\theta,2,1} \ln \theta_{2}(i_0,i_1) &=  \zeta_0 +  \zeta_1  \ln \theta_{0}  + \zeta_3 \ln i_{0} +  \zeta_4 \ln i_{1} + \zeta_5 \ln{\theta_0}\ln i_{0}  
\end{align*} 
for identified parameters $\{\zeta_j\}^5_{j=1}$. Using $ \ln \theta_{0}  \sim N(0,s_0^2)$ it follows that
$$E[\theta_{2}(i_0,i_1)^{\lambda_{\theta,2,1}}] = \exp\left(\zeta_0  + \zeta_3 \ln i_{0} +  \zeta_4 \ln i_{1} + \left(   \zeta_1  + \zeta_5 \ln i_{0} \right)^2 s_0^2   \right)   $$
but 
$$E[\theta_{2}(i_0,i_1)] = \exp\left( \frac{1}{\lambda_{\theta,2,1}} \left(\zeta_0  + \zeta_3 \ln i_{0} +  \zeta_4 \ln i_{1} + \left(   \zeta_1  + \zeta_5 \ln i_{0} \right)^2 \frac{s_0^2}{\lambda_{\theta,2,1}} \right) \right)   $$
which implies that generally $\argmax_{(i_{0},i_{1}) \in \mathcal{I}} E \left[  \theta_{2}(i_0,i_1)   \right] \neq \argmax_{(i_{0},i_{1}) \in \mathcal{I}} E[\theta_{2}(i_0,i_1)^{\lambda_{\theta,2,1}}]     $
Hence, using different scaled versions of log-skills (or different fixed values of $\lambda_{\theta,t,1}$) results in different optimal investment sequences when studying the skill level. 

Finally, consider the CES production function
\[
\theta_{t+1} =  ( \gamma_{1t} \theta_{t}^{\sigma_{t}} + \gamma_{2t} I_{t}^{\sigma_{t}} )^{1/\sigma_{t}} 
\]
The parameter vector is now $\tau_0 = (   \lambda_{\theta,0,1},  \lambda_{\theta,1,1},  \lambda_{\theta,2,1}, \sigma_0, \sigma_1, \gamma_{10},  \gamma_{20}, \gamma_{11},  \gamma_{21}, s^2_0  )$ and we can write the production function in terms of observables as
\[
\exp(Z_{\theta,t+1,1}) =  ( \gamma_{1t} \exp(Z_{\theta,t,1})^{\sigma_{t}/\lambda_{\theta,t,1}} + \gamma_{2t} I_{t}^{\sigma_{t}} )^{\lambda_{\theta,t+1,1}/\sigma_{t}} 
\]
Due to the functional form restrictions of the CES production function, it can be shown that $\tau_0 $ is point identified without any additional restrictions.

Thus, the commonly imposed scale restrictions  ($\lambda_{\theta,t,1} = 1$ for all $t$) lead to misspecification unless the true values are all $1$. Consequently, imposing this restriction yields different conclusions for different scaled versions of log-skills (or different fixed values of $\lambda_{\theta,t,1}$) irrespective of the counterfactual. Similar as in the trans-log case, many important features are identified without additional restrictions and are invariant to scaling of  the measures.

These results generalize to more complicated settings as discussed in the next section.

\section{Skill formation models}
\label{s:skillform}

\subsection{Model}
\label{s:trans_log}

I now discuss issues arising from normalizations in a general class of skill formation models. As before, $\theta_{t}$ and $I_t$ denote skills and investment at time $t$, respectively. Now neither skills nor investment are directly observed and we denote the observed measurements by $Z_{\theta,t,m}$ and $Z_{I,t,m}$, respectively. Specifically, I consider the model:
\begin{eqnarray}
\theta_{t+1} &=& f(\theta_{t},I_{t},\delta_{t}, \eta_{\theta,t}) \hspace{46mm} t = 0, \ldots, T-1 \label{eq:prod_fn} \\
Z_{\theta,t,m} &=& \mu_{\theta,t,m} + \lambda_{\theta,t,m} \ln \theta_{t} + \eps_{\theta,t,m} \hspace{26mm} t = 0, \ldots, T, m = 1,2 \label{eq:measurement_eq} \\
Z_{I,t,m} &=& \mu_{I,t,m} + \lambda_{I,t,m} \ln I_{t} + \eps_{I,t,m} \hspace{26mm} t = 0, \ldots, T-1, m = 1,2  \label{eq:measurement_eq_invest}  
\end{eqnarray}
The first equation describes the production technology with a production function $f$ that depends on skills and investment at time $t$, a parameter vector $\delta_{t}$, and an unobserved shock $\eta_{\theta,t}$. The second and the third equation describe the measurement system for unobserved (latent) skills $\theta_{t}$ and unobserved investment $I_t$, respectively. Observed investment is a special case with $ \mu_{I,t,m} = 0$, $\lambda_{I,t,m} = 1$, and $\eps_{I,t,m} = 0$ for all $m$ and $t$ in which case $Z_{I,t,m} = \ln I_{t}$.

Next, I introduce two equations to allow for endogenous investment and anchoring at an adult outcome. If investment is exogenous, in the sense that $\eta_{\theta,t}$ is independent of $I_{t}$, then these equations are not needed for the main identification results. That is, let
\begin{eqnarray}
	\ln I_t &=& \beta_{0t} + \beta_{1t} \ln \theta_{t}  + \beta_{2t} \ln Y_{t}   + \eta_{I,t} \hspace{26mm} t = 0, \ldots, T-1  \label{eq:investment}   \\
	Q &=& \rho_{0} + \rho_{1} \ln \theta_{T} + \eta_Q \label{eq:anchor_eq}
\end{eqnarray}
Here  $Y_t$ is parental income (or another exogenous variable that affects investment) and $Q$ is an adult outcome, such as earnings or education. An adult outcome does not necessarily have to be available and we can simply use a skill measure in period $T$ in its place.

In summary, the observed variables are income $\{Y_t\}_{t=0}^{T-1}$, the measures  $\{Z_{\theta,t,m}\}_{t=0,\ldots,T, m= 1,2}$ and $\{ Z_{I,t,m} \}_{t=0,\ldots,T-1, m= 1,2}$,  and the adult outcome $Q$, but we neither observe  skills $\{\theta_{t}\}^{T}_{t=0}$ nor investment  $\{I_t\}_{t=0}^{T-1}$. We also do not observe any of the errors/shocks in the five equations. The parameters  are $\{\mu_{\theta,t,m},\lambda_{\theta,t,m}\}_{t=0,\ldots,T,m=1,2}$, $\{\mu_{I,t,m},\lambda_{I,t,m}\}_{t=0,\ldots,T-1,m=1,2}$, $\{\delta_t\}^{T-1}_{t=0}$, $\{ \beta_{0t}, \beta_{1t}, \beta_{2t}\}^{T-1}_{t=0}$, and $(\rho_0,\rho_1)$.

In the following analysis, I consider the two most commonly used forms for the production technology in the empirical literature, namely the trans-log production function with
	\begin{equation}
		\label{eq:translog}
	\ln \theta_{t+1} = a_t + \gamma_{1t}\ln \theta_{t} + \gamma_{2t} \ln I_{t} + \gamma_{3t}\ln \theta_{t} \ln I_{t}  + \eta_{\theta,t} 
  \end{equation}
	 and parameter vector $\delta_{t} = (a_t, \gamma_{1t},\gamma_{2t},\gamma_{3t})$ and the CES production function with
	\begin{equation}
	\label{eq:ces}
	\theta_{t+1} =  ( \gamma_{1t} \theta_{t}^{\sigma_{t}} + \gamma_{2t} I_{t}^{\sigma_{t}} )^{\psi_t/\sigma_{t}} \exp(\eta_{\theta,t})
  \end{equation}
	 and parameter vector $\delta_{t} = (\gamma_{1t},\gamma_{2t},\sigma_{t},\psi_t)$. When $\gamma_{3t} = 0$, the trans-log reduces to the Cobb-Douglas production function.

I now state several additional assumptions that are common in the literature.

\begin{assumption}\label{a:baseline}  \qquad 
	
	\begin{enumerate}[(a)]
		
		\item $\{\{\eps_{\theta,t,m}\}_{t=0,\ldots,T, m= 1,2}, \{\eps_{I,t,m} \}_{t=0,\ldots,T-1, m= 1,2}, \eta_Q\}$ are jointly independent and independent of $\{\{\theta_{t}\}^{T}_{t=0},\{I_t\}_{t=0}^{T-1}\}$ conditional on $\{Y_t\}_{t=0}^{T-1}$.

		\item All random variables have bounded first and second moments.

		\item $E[\eps_{\theta,t,m}] = E[\eps_{I,t,m}] = E[\eps_Q] =  0$ for all $t$ and $m$.
		  
		\item $\lambda_{\theta,t,m}, \lambda_{I,t,m} \neq 0$ for all $t$ and $m$.
		
		\item For all $t \in \{0,\ldots,T\}$, $cov(\ln \theta_t,\ln I_s) \neq 0$ for some $s \in \{0,\ldots, T-1\}$ or $cov(\ln \theta_t,\ln \theta_s) \neq 0$ for some $s \in \{0, \ldots, T\} \backslash t$ . For all $t \in \{0,1,\ldots,T-1\}$, $cov(\ln I_t,\ln \theta_s) \neq 0$ for some $s \in \{0, \ldots, T\}$ or $cov(\ln I_t,\ln I_s) \neq 0$ for some $s \in \{0, \ldots, T-1\} \backslash t$.
		
		\item For all $t$ and $m$ the real zeros of the characteristic functions of $\eps_{\theta,t,m}$ are isolated and are distinct from those of its derivatives. Identical conditions hold for the characteristic functions of $\eps_{I,t,m}$ and $\eta_{Q}$.
		
		\item The support of $(\theta_{t},I_{t},Y_t)$ includes an open ball in $\R^3$ for all $t$.
		
		\item $E[\eta_{I,t} \mid \theta_{t}, Y_t] = 0$ and $E[\eta_{\theta,t} \mid \theta_{t},\eta_{I,t} , Y_t] = \kappa_t \eta_{I,t}$ for all $t$.

	\end{enumerate}
	 
\end{assumption}

Part (a) imposes common independence assumptions on the measurement errors. Importantly, $I_t$ and $\theta_t$ are not independent and $I_t$ may be endogenous and contemporaneously correlated with $\eta_{\theta,t}$. Part (b) is a standard restriction, part (c) is needed because all measurement equations contain an intercept, and part (d) ensures that the skills actually affect the measures. Part (e) requires that skills and investment are correlated in some time periods. Sufficient conditions are that $cov(\ln \theta_{t+1}, \ln \theta_t) \neq 0$ and $cov(\ln I_{t+1}, \ln I_t) \neq 0$ for all $t$. Under parts (a) and (d) zero covariances of the latent variables are identified because, for example, $cov(\ln \theta_{t}, \ln \theta_s) = 0$ if and only if $cov( Z_{\theta,t,1}, Z_{\theta,s,1}) = 0$. Notice that I only require two measures in each period. One can drop part (e) by assuming that three measures are available.  Part (f) contains weak regularity conditions needed for nonparametric identification of the distributions of skills and investment and that hold for most common distributions.    Part (g) is a mild support condition that ensures sufficient variation of $(\theta_{t},I_{t},Y_t)$. Part (h) implies that $Y_t$ can serve as an instrument with identification based on a control function argument, as in \shortciteN{AMN:19}. Linearity of the conditional mean function can be relaxed to allow for more flexible functional forms. Exogenous investment is a special case with $\kappa_{t} = 0$.

Under parts (a)--(f) of Assumption \ref{a:baseline} we get the following result.
\begin{lemma}
	\label{l:identjoint}
	Suppose that parts (a)--(f) of Assumption \ref{a:baseline} hold. Then the joint distribution of 
	$$ \left( \{\mu_{\theta,t,m} + \lambda_{\theta,t,m} \ln \theta_{t}\}_{t=0,\dots, T,m =1,2},  \{ \mu_{I,t,m} + \lambda_{I,t,m} \ln I_{t} \}_{t=0,\dots, T-1,m =1,2}, \rho_0 + \rho_1 \ln \theta_{T} \right) $$ 
	is point identified conditional on $\{Y_t\}_{t=0}^{T-1}$.
\end{lemma}

The proof follows from an extension of Kotlarski's Lemma \cite{EW:12}. Lemma \ref{l:identjoint} shows that under Assumption \ref{a:baseline},  the distribution of linear combinations of log-skills and investments is identified. However, Assumption \ref{a:baseline} does not imply identification of the parameters in equations (\ref{eq:prod_fn})--(\ref{eq:measurement_eq_invest}), such as $\delta_t$. Below I discuss additional assumptions, which have been used in the literature, to achieve point identification of  two sets of parameters: (i) the primitive parameters of the model (\ref{eq:prod_fn})--(\ref{eq:anchor_eq}):   $\{\mu_{\theta,t,m},\lambda_{\theta,t,m}\}_{t=0,\ldots,T,m=1,2}$, $\{\mu_{I,t,m},\lambda_{I,t,m}\}_{t=0,\ldots,T-1,m=1,2}$, $\{\delta_t\}^{T-1}_{t=0}$, $\{ \beta_{0t}, \beta_{1t}, \beta_{2t}\}^{T-1}_{t=0}$, and $(\rho_0,\rho_1)$  and (ii) ``policy relevant'' parameters, such as how changes in investment or income affect  $Q$.  I consider various combinations of  assumptions and I discuss in what sense restrictions are normalizations.

There are two separate issues concerning identification and normalizations in this model. First, the previous literature has focused on sufficient conditions for point identification, which includes scale and location restrictions. However, it is unclear whether these restrictions are necessary. If they instead restrict identified parameters, the model may be misspecified and estimators are generally inconsistent. Second, even if the restrictions simply select an element of the identified set, it is important to understand which  features are invariant to arbitrary scale and location restrictions. Although the implications in the CES case are more interesting, I start with the more transparent trans-log case where only the second issue arises. As pointed out before, whether or not a restriction is a normalization depends on the model, and parameters may be invariant in some settings, but not in others.

\subsection{Trans-log production function}

In this section I consider the trans-log production function given in equation (\ref{eq:translog}). I first introduce additional assumptions that are commonly used in the literature.

\begin{assumption}\label{a:normalization} 
	$\lambda_{\theta,0,1}=1$ and $\mu_{\theta,0,1} = 0$. 
\end{assumption}

\begin{assumption}\label{a:ageinvariant_technology_skills} \quad
	
 \begin{enumerate}[\label=(a)]
 	
 	\item $\lambda_{\theta,t,1}=\lambda_{\theta,t+1,1}$ and $\mu_{\theta,t,1} = \mu_{\theta,t+1,1}$ for all $t = 0, \ldots, T-1$  
 	\item $a_t = 0$ and	$ \gamma_{1t} + \gamma_{2t} + \gamma_{3t} = 1$ for all $t=0,\ldots,T-1$.
 \end{enumerate}

\end{assumption}

\begin{assumption}\label{a:ageinvariant_technology_investment}\quad
	 \begin{enumerate}[\label=(a)]
		\item 	  $\lambda_{I,t,1} = 1$ and $\mu_{I,t,1} = 0$ for all $t = 0, \ldots, T-1$.
		\item 	$\beta_{0t} = 0$ and	$\beta_{1t} + \beta_{2t} = 1$ for all $t=0,\ldots,T-1$.
	\end{enumerate}

\end{assumption}

Assumption \ref{a:normalization} is usually thought of as a normalization, which is commonly imposed since log-skills are only identified up to scale and location. Here, I impose the restrictions on the first measure, which is without loss of generality. Instead of fixing the intercept and the slope coefficient in equation (\ref{eq:measurement_eq}) for $t=0$, one could set $\rho_{0} =0 $ and $\rho_{1}=1$ and thus ``anchor'' the skills at $Q$. Assumption \ref{a:normalization} anchors the skills at $Z_{\theta,0,1}$, but analogous issues discussed here arise with anchoring at $Q$ (see Section \ref{s:anchor_tl}). Without such an assumption, the parameters are not point identified. Assumption \ref{a:ageinvariant_technology_skills}(a) states that the skill measures are age-invariant (using the terminology of \citeN{AW:22} -- see their Definition 1 and footnote 10).  Assumption \ref{a:ageinvariant_technology_skills}(b) imposes restrictions on the technology, which \citeN{AW:22} refer to as a known scale and location assumption in a more general context.  Assumption \ref{a:ageinvariant_technology_investment}(a) says that an investment measure is age-invariant. Assumption \ref{a:ageinvariant_technology_investment}(b) states constant return to scale in equation (\ref{eq:investment}), which is a strong assumption and used for point identification without age-invariant investment measures. If investment is observed (i.e. $Z_{I,t,m} = \ln I_t$), Assumption \ref{a:ageinvariant_technology_investment}(a) is automatically satisfied.

I now characterize the identified set of the primitive parameters under Assumption \ref{a:baseline} only. I then discuss point identification under different combinations of Assumptions \ref{a:baseline}--\ref{a:ageinvariant_technology_investment} and show that several policy relevant parameters are invariant to the restrictions in Assumption \ref{a:normalization} and are in fact point identified under Assumption \ref{a:baseline}. Finally, I  illustrate why Assumption \ref{a:normalization} is generally not a normalization for the primitive as well as some policy relevant parameters.

\subsubsection{Identification}
\label{s:identification_tl}

Define $\tilde{I}_t = \exp(\mu_{I,t,1})I_{t}^{\lambda_{I,t,1}}$ and  $\tilde{\theta}_t = \exp(\mu_{\theta,t,1})\theta_{t}^{\lambda_{\theta,t,1}}$ so that
$$\ln \tilde{\theta}_{t} = \mu_{\theta,t,1} + \lambda_{\theta,t,1} \ln \theta_{t} \qquad \text{ and } \qquad \ln \theta_{t} = \frac{\ln \tilde{\theta}_t -  \mu_{\theta,t,1}}{\lambda_{\theta,t,1} }.$$
We can then rewrite the production function in terms of $\tilde{\theta}_t$ and $\tilde{I}_t$  because
$$\frac{\ln \tilde{\theta}_{t+1} -  \mu_{\theta,t+1,1}}{\lambda_{\theta,t+1,1} } = a_t + \gamma_{1t} \frac{\ln \tilde{\theta}_t -  \mu_{\theta,t,1}}{\lambda_{\theta,t,1} } +  \gamma_{2t} \frac{\ln \tilde{I}_t -  \mu_{I,t,1}}{\lambda_{I,t,1} } +  \gamma_{3t} \frac{\ln \tilde{\theta}_t -  \mu_{\theta,t,1}}{\lambda_{\theta,t,1} }\frac{\ln \tilde{I}_t -  \mu_{I,t,1}}{\lambda_{I,t,1} } + \eta_{\theta,t} .$$
After rearranging, we can then rewrite equations (\ref{eq:prod_fn})--(\ref{eq:anchor_eq}) as 
\begin{eqnarray}
\hspace{8mm} \ln \tilde{\theta}_{t+1} &=& \tilde{a}_t + \tilde{\gamma}_{1t} \ln \tilde{\theta}_{t} +  \tilde{\gamma}_{2t} \ln \tilde{I}_t  +  \tilde{\gamma}_{3t} \ln \tilde{\theta}_{t} \ln \tilde{I}_t + \tilde{\eta}_{\theta,t} \hspace{8mm} t = 0, \ldots, T-1 \label{eq:prod_fn_rw}\\
Z_{\theta,t,m} &=& \tilde{\mu}_{\theta,t,m} + \tilde{\lambda}_{\theta,t,m} \ln \tilde{\theta}_{t} + \eps_{\theta,t,m} \hspace{38mm} t = 0, \ldots, T, m = 1,2  \label{eq:measurement_eq_rw} \\
Z_{I,t,m} &=& \tilde{\mu}_{I,t,m} +\tilde{\lambda}_{I,t,m} \ln \tilde{I}_{t} + {\eps}_{I,t,m} \hspace{38mm} t = 0, \ldots, T-1, m = 1,2  \label{eq:measurement_eq_invest_rw}\\
\ln \tilde{I}_t &=& \tilde{\beta}_{0t} + \tilde{\beta}_{1t} \ln \tilde{\theta}_{t}  + \tilde{\beta}_{2t} \ln Y_{t}   + \tilde{\eta}_{I,t} \hspace{32mm} t = 0, \ldots, T-1  \label{eq:investment_rw} \\
Q &=& \tilde{\rho}_{0} + \tilde{\rho}_{1} \ln \tilde{\theta}_{T} + \eta_Q   \label{eq:anchor_eq_rw}
\end{eqnarray}
where  
$$ \tilde{\gamma}_{1t}  =  \frac{\lambda_{\theta,t+1,1}}{\lambda_{\theta,t,1}} \left(  \gamma_{1t}  - \frac{\mu_{I,t,1}}{\lambda_{I,t,1}}\gamma_{3t}  \right), \;\;   \tilde{\gamma}_{2t}  = \frac{ \lambda_{\theta,t+1,1} }{\lambda_{I,t,1}} \left( \gamma_{2t}  - \frac{\mu_{\theta,t,1}}{\lambda_{\theta,t,1}}\gamma_{3t} \right), \;\;  \tilde{\gamma}_{3t}  = \frac{\lambda_{\theta,t+1,1}}{\lambda_{\theta,t,1} \lambda_{I,t,1}}\gamma_{3t}$$
and expressions for all other parameters in (\ref{eq:prod_fn_rw})--(\ref{eq:anchor_eq_rw}) are provided in Appendix \ref{s:app_additional_parameters}. Importantly, by construction, $\tilde{\mu}_{\theta,t,1} = \tilde{\mu}_{I,t,1} = 0$ and $\tilde{\lambda}_{\theta,t,1} = \tilde{\lambda}_{I,t,1} = 1$ and Lemma \ref{l:identjoint} implies that the joint distribution of $(\{\ln \tilde{\theta}_{t}\}^T_{t=0}, \{\ln \tilde{I}_{t}\}^T_{t=0})$ is point identified conditional on $\{Y_t\}_{t=0}^{T-1}$, which then yields identification of all parameters in (\ref{eq:prod_fn_rw})--(\ref{eq:anchor_eq_rw}). These parameters are functions of the primitive parameters in (\ref{eq:prod_fn})--(\ref{eq:anchor_eq}). The next theorem, which characterizes the identified set, states that the identified set consists of all primitive parameters that imply the same values in (\ref{eq:prod_fn_rw})--(\ref{eq:anchor_eq_rw}) as the true parameters, which is very similar to the probit model.

\begin{theorem}
	\label{th:identparams_tl}
Suppose Assumption \ref{a:baseline}	holds. 

\begin{enumerate}
	\item The identified set of  $\{\mu_{\theta,t,m},\lambda_{\theta,t,m}\}_{t=0,\ldots,T,m=1,2}$, $\{\mu_{I,t,m},\lambda_{I,t,m}\}_{t=0,\ldots,T-1,m=1,2}$, $(\rho_0,\rho_1)$,  $\{ \beta_{0t}, \beta_{1t}, \beta_{2t}\}^{T-1}_{t=0}$, $\{a_t,\gamma_{1t},\gamma_{2t}, \gamma_{3t}\}^{T-1}_{t=0}$ consists of all vectors that yield the same values of
	$\{\tilde{\mu}_{\theta,t,m},\tilde{\lambda}_{\theta,t,m}\}_{t=0,\ldots,T,m=1,2}$, $\{\tilde{\mu}_{I,t,m},\tilde{\lambda}_{I,t,m}\}_{t=0,\ldots,T-1,m=1,2}$, $(\tilde{\rho}_0,\tilde{\rho}_1)$, $\{ \tilde{\beta}_{0t}, \tilde{\beta}_{1t}, \tilde{\beta}_{2t}\}^{T-1}_{t=0}$, and $\{\tilde{a}_t,\tilde{\gamma}_{1t},\tilde{\gamma}_{2t}, \tilde{\gamma}_{3t}\}^{T-1}_{t=0}$ as the true parameter vectors.
	
	\item Let  $\{\bar{\mu}_{\theta,t,1},\bar{\lambda}_{\theta,t,1}\}_{t=0}^{T}$, $\{\bar{\mu}_{I,t,1},\bar{\lambda}_{I,t,1}\}_{t=0}^{T-1}$, be fixed constants with $\bar{\lambda}_{\theta,t,1},\bar{\lambda}_{I,t,1}\neq0$ for all $t$.  If in addition $\{{\mu}_{\theta,t,1},{\lambda}_{\theta,t,1}\}_{t=0}^{T} = \{\bar{\mu}_{\theta,t,1},\bar{\lambda}_{\theta,t,1}\}_{t=0}^{T}$ and $\{{\mu}_{I,t,1},{\lambda}_{I,t,1}\}_{t=0}^{T-1} = \{\bar{\mu}_{I,t,1},\bar{\lambda}_{I,t,1}\}_{t=0}^{T-1}$, then
	the identified set is a singleton.

\end{enumerate}

\end{theorem}
 
Part 2 of the theorem shows that the parameters are indeed not point identified under Assumption \ref{a:baseline} only and that the sources of underidentification are the ambiguous scales and locations of skills and investments.  For example, without additional assumptions, equations (\ref{eq:prod_fn})--(\ref{eq:anchor_eq}) and (\ref{eq:prod_fn_rw})--(\ref{eq:anchor_eq_rw})
 are observationally equivalent, and we cannot distinguish between the skills $\theta_t$ and $\tilde{\theta}_t$ and the corresponding production functions. Even if investment was observed (and $\mu_{I,t,1} = 0$ for all $t$), we can then only identify $(\lambda_{\theta,t+1,m}/\lambda_{\theta,t,m})\gamma_{1t}$, but not $\lambda_{\theta,t,m}$ and $\gamma_{1t}$ separately. Hence, we cannot distinguish between changes in the quality of the measurements  ($\lambda_{\theta,t+1,m}/\lambda_{\theta,t,m}$) and changes in the technology (${\gamma}_{1t}$). For example, suppose $Z_{\theta,t,m}$ are test scores. We then cannot distinguish between all children getting smarter or tests becoming easier.  Similarly, we can at best identify $\gamma_{2t}$ up to scale, even with observed investment. 

The following corollary states that all parameters are point identified under additional assumptions. These results are an extension of those in \citeN{AW:22}, who assume that investment is exogenous (in the sense that it is uncorrelated with $ \eta_{\theta,t}$). 

\begin{corollary}
\label{c:pointident}
Suppose Assumptions \ref{a:baseline} and \ref{a:normalization} hold. Suppose either Assumption \ref{a:ageinvariant_technology_skills}(a) or Assumption \ref{a:ageinvariant_technology_skills}(b) holds. Suppose either Assumption \ref{a:ageinvariant_technology_investment}(a) or Assumption \ref{a:ageinvariant_technology_investment}(b) holds. Then all parameters are point identified. 
\end{corollary}

The corollary also immediately implies that Assumptions \ref{a:ageinvariant_technology_skills}(a) and \ref{a:ageinvariant_technology_skills}(b) together impose additional testable restrictions, which is one of the main contributions of Agostinelli and Wiswall (2016a, 2016b, 2024). Contrarily, as shown in Theorem \ref{th:obseq} in the appendix and illustrated in examples below, if the model is correctly specified and Assumption \ref{a:baseline} holds, then there always exist sets of parameters which are consistent with the data and satisfy Assumptions \ref{a:baseline}, \ref{a:normalization}, either \ref{a:ageinvariant_technology_skills}(a) or \ref{a:ageinvariant_technology_skills}(b), and either \ref{a:ageinvariant_technology_investment}(a) or \ref{a:ageinvariant_technology_investment}(b). These different sets of assumptions therefore impose no additional restrictions on the distribution of observables. While different sets of assumptions yield point identification and are observationally equivalent, the estimated primitive parameters are usually quite different, as illustrated in Section \ref{s:nonivariant_tl}.

 \subsubsection{Invariant parameters}
 \label{s:invariant_tl}

I now show that many interesting features are point identified under Assumption \ref{a:baseline} only because they can all be rewritten in terms of the identified parameters in (\ref{eq:prod_fn_rw})--(\ref{eq:anchor_eq_rw}), as in Section \ref{s:examples}. They include summaries of the productions functions and effects of investment and income on skills and adult outcomes. These features do not constitute an exhaustive list and there may be many others.  To state the formal results, recall that $Q_{\alpha}(\theta_t)$ is the $\alpha$ quantile of the skill distribution at time $t$ and $F_{\ln(\theta_{t+1})}(\cdot)$ is the cdf of log-skills at time $t+1$. Define $s_{1t}(\alpha_1,\alpha_2,\alpha_3) =  a_t + \gamma_{1t} \ln Q_{\alpha_1}(\theta_t) +  \gamma_{2t} \ln  Q_{\alpha_2}(I_t)  +  \gamma_{3t} \ln Q_{\alpha_1}(\theta_t) \ln  Q_{\alpha_2}(I_t) +  Q_{\alpha_3}(\eta_{\theta,t}) $ which are the log-skills in period $t+1$ for specific quantiles of the inputs in period $t$. Similarly, let $s_{2t}(\alpha_1,\alpha_2,\alpha_3,y) = a_t + \gamma_{1t} \ln Q_{\alpha_1}(\theta_t) +  \gamma_{2t} \ln  I_t(y)   +  \gamma_{3t} \ln Q_{\alpha_1}(\theta_t) \ln  I_t(y)  +  Q_{\alpha_2}(\eta_{\theta,t})$, where $\ln I_t(y) = \beta_{0t} + \beta_{1t} \ln Q_{\alpha_1}(\theta_t)  + \beta_{2t} \ln y    +  Q_{\alpha_3}(\eta_{I,t}) $, which are log-skills and log-investment, that depend on a specific exogenously set value of $Y$.

\begin{theorem}
	\label{th:identfunctions}
	Suppose Assumption \ref{a:baseline}	holds.  
	
	\begin{enumerate}
	 	
		\item 		$F_{\ln \theta_{t+1}}(s(\alpha_1,\alpha_2,\alpha_3) )$ and  $\mu_{\theta,t+1,m} +  \lambda_{\theta,t+1,m}s_{1t}(\alpha_1,\alpha_2,\alpha_3) +  Q_{\alpha_4}(\eps_{\theta,t+1,m})$ are point identified for all $\{\alpha_j\}^4_{j=1}\in (0,1)$. $\frac{\partial \lambda_{\theta,t+1 ,m} \ln \theta_{t+1}}{\partial \lambda_{\theta,t ,m'} \ln \theta_t } \mid_{ I_t = Q_{\alpha_2}(I_t)  }$ and $\frac{\partial \lambda_{\theta,t+1 ,m} \ln \theta_{t+1}}{\partial \lambda_{I,t ,m'} \ln I_t } \mid_{ \theta_t = Q_{\alpha_1}(\theta_t)  }$ are point identified for all $m,m'$.

		\item $F_{\ln \theta_{t+1}}(s_{2t}(\alpha_1,\alpha_2,\alpha_3,y) )$ and 	$\mu_{\theta,t+1,m} +  \lambda_{\theta,t+1,m}s_{2t}(\alpha_1,\alpha_2,\alpha_3,y) +  Q_{\alpha_4}(\eps_{\theta,t+1,m})$  are point identified for all $\{\alpha_j\}^4_{j=1}\in (0,1)$.  
		
		\item $  P\left(Q \leq q   \mid \theta_s= Q_{\alpha_1}(\theta_s), \{I_t = Q_{\alpha_{2t}}(I_t) \}^{T-1}_{t=0}, \{\eta_{\theta,t} = Q_{\alpha_{3t}}(\eta_{\theta,t})   \}^{T-1}_{t=s}  \right)$ is point identified for all $\alpha_1,\{\alpha_{2t},\alpha_{3t} \}^{T-1}_{t=s} \in (0,1)$.
		
 		\item $  P\left(Q \leq q   \mid \theta_s= Q_{\alpha}(\theta_s), \{Y_t = y_t \}^{T-1}_{t=s}     \right)$ is point identified for all $\alpha \in (0,1)$.

	\item Suppose Assumption \ref{a:ageinvariant_technology_skills}(a) also holds. Then $\gamma_{1t} + \gamma_{3t}\ln Q_{\alpha}(I_t)$ is point identified for all $\alpha$ and  $P(\gamma_{1t} + \gamma_{3t}\ln I_t \leq q)$ is point identified for all $q \in \R$.

	\end{enumerate}

\end{theorem}

The function $F_{\ln \theta_{t+1}}(a_t + \gamma_{1t} \ln Q_{\alpha_1}(\theta_t) +  \gamma_{2t} \ln  Q_{\alpha_2}(I_t)  +  \gamma_{3t} \ln Q_{\alpha_1}(\theta_t) \ln  Q_{\alpha_2}(I_t) +  Q_{\alpha_3}(\eta_{\theta,t}) )$ measures how changes in skills and investment changes the relative standing in the skill distribution. For example, consider an individual with $\theta_t = Q_{0.1}(\theta_t)$, meaning that the person is at lowest $10\%$ of the skill distribution at time $t$. Then, given investment $I_t = Q_{0.25}(I_t)$ and a median production function shock, $ \eta_{\theta,t} = Q_{0.5}(\eta_{\theta,t})$,  
$F_{\ln \theta_{t+1}}(s(0.1,0.25,0.5))$ tells us the relative rank (or the quantile) in the skill distribution at time $t+1$. We can then for example vary the investment quantile and analyze how future skill ranks are affected. This feature is identified because I show in the appendix that
\begin{eqnarray*}
	&& \hspace{-1cm} F_{\ln  {\theta}_{t+1}}\left(    {a}_t +  {\gamma}_{1t}  Q_{\alpha_1} \ln( \theta_{t})   +  {\gamma}_{2t}  Q_{\alpha_2} (\ln I_{t})   +  {\gamma}_{3t}  Q_{\alpha_1} \ln( \theta_{t})  Q_{\alpha_2} (\ln I_{t}) +   Q_{\alpha_3}(   \eta_{\theta,t})     \right) \\
	&=&     F_{\ln \tilde{\theta}_{t+1}}( \tilde{a}_t + \tilde{\gamma}_{1t} \ln Q_{\alpha_1}(\tilde{\theta}_{t}) +  \tilde{\gamma}_{2t} \ln Q_{\alpha_2}(\tilde{I}_{t})  +  \tilde{\gamma}_{3t} \ln Q_{\alpha_1}(\tilde{\theta}_{t})\ln Q_{\alpha_2}(\tilde{I}_{t})  + Q_{\alpha_3}(\tilde{\eta}_{\theta,t})    )   
\end{eqnarray*}
Thus, one can estimate the model based on (\ref{eq:prod_fn_rw})--(\ref{eq:anchor_eq_rw}) and calculate the  feature using $\tilde{\theta}_{t}$ instead of ${\theta}_{t}$. The estimand then corresponds to the true effect, even if Assumptions \ref{a:normalization}--\ref{a:ageinvariant_technology_investment} do not hold.

Once we know the rank at time $t+1$ and fix investment and production shock quantiles in that period, we can identify the skill rank at time $t+2$. Using recursive arguments, we can identify the relative rank in period $T$, given investment and production shock quantiles in all period and a skill quantile in period $0$. We could then make statements such as: ``A person at lowest $10\%$ of the skill distribution in period $0$ would end up at the $30\%$ quantile of the original skill distribution in the final period with a particular investment strategy and median production function shocks.'' These statements would allow comparisons of investment strategies, assessing heterogeneous effects, and choosing optimal investments.

Instead of considering ranks of skills, one can interpret the effects in the units of any of the measures. For example, changing investment from $Q_{\alpha_2}(I_t)$ to  $Q_{\alpha_2'}(I_t)$ changes skills at time $t+1$ such that skill measure $m$ changes by $\lambda_{\theta,t+1,m}(s_1(\alpha_1,\alpha_2',\alpha_3) -  s_1(\alpha_1,\alpha_2,\alpha_3))$ (holding everything else equal). We can also identify that a change in investment, which corresponds to a 1 unit increase in investment measure $m'$ at time $t$, affects skills at time $t+1$ in a way that skill measure $m$ changes by $\frac{\partial \lambda_{\theta,t+1 ,m} \ln \theta_{t+1}}{\partial \lambda_{I,t ,m'} \ln I_t } \mid_{ \theta_t = Q_{\alpha_1}(\theta_t)  }$ (see Section \ref{s:application} for specific examples). 

The second part is very similar, but it also considers exogenous changes in income.

Instead of considering the skill rank in the final period, we could also look at the distribution of the adult outcome $Q$ (or, alternatively, one of the skill measures in the final period). Notice that it is important to condition on the production function shocks, because investment in endogenous. We can either fix a quantile or average over its marginal distribution since we obtain identification for all quantiles. The fourth part focuses on the effect of income on skills. Again, we can point identify averages and features of the distribution such as $\int E\left(Q \mid \theta_s = \theta, \{Y_t = y_t \}^{T-1}_{t=s}   \right) f_{\theta_s}(\theta) d\theta$ 
 (which differs from $E\left(Q \mid \{Y_t = y_t \}^{T-1}_{t=s}  \right)$ due to a potential dependence between $ \{Y_t = y_t \}^{T-1}_{t=s} $ and   skills). Also notice that 
\begin{align*}
	\int E\left( Q \mid \theta_s = \theta, \{Y_t = y_t \}^{T-1}_{t=s}   \right)  f_{\theta_s}(\theta)  d\theta & = \rho_0 + \rho_1 \int E\left(\ln \theta_T \mid \theta_s = \theta, \{Y_t = y_t \}^{T-1}_{t=s}   \right)  f_{\theta_s}(\theta)  d\theta
\end{align*}
Thus, we can identify the sequence $\{Y_t = y_t \}^{T-1}_{t=s}$ that maximizes the conditional expected value of  $\ln \theta_T$. For this sequence, we do not necessarily need to observe an adult outcome because we can instead use one of the skill measures in period $T$. To identify these features, one only has to identify the joint distribution of $(Q,Y_1, \ldots, Y_{T-1},\tilde{\theta}_s)$. For example, when $s = 0$, we do not require any skill measures in periods $1,2,\ldots,T$.\footnote{\label{f:DKP}Part 4 of Theorem \ref{th:identfunctions} also implies identification of $$\int\int\left( \frac{\partial}{\partial y_s} E(Q \mid \theta_s = \theta, \{Y_t = y_t \}^{T-1}_{t=s})\right)f_{\theta_s,Y_s, \ldots, Y_{T-1}}(\theta,y_s, \ldots, y_{T-1}) d\theta dy_s \cdots d y_{T-1} $$
	which \shortciteN{DKP:20} refer to as ``anchored treatment effects''. \shortciteN{DKP:20} show that these effects are identified without age-invariant measures. My results show that one in fact does not even need any measures of investments or measures of skills in periods $s+1,2,\ldots, T$ (and therefore also no assumptions on the measurement systems, including age-invariance and independence). In addition, my arguments are not specific to the trans-log production function and carry over to other settings.}  Importantly objects such as $ \int E\left(\theta_T \mid \theta_s = \theta, \{Y_t = y_t \}^{T-1}_{t=s}   \right)  f_{\theta_s}(\theta)  d\theta$ are not point identified without the scale and location restrictions and are sensitive to the specific values used (see Example \ref{e:levels_mc} below).

\begin{remark}
\label{r:DKP_2}
To summarize the production technology, \citeN{DKP:20} show that standardizing skills can lead to features that are invariant to scale and location restrictions and age-invariance. In particular, they show identification of the distribution of $\left(\frac{\partial \ln \theta_{t+1} }{\partial \ln  I_t}\right)/ std(\ln \theta_{t+1}) $. One can then make statements such as ``increasing investment by 1\%, increases log-skills by $x \times std(\ln \theta_{t+1})$''. Part 1 of Theorem \ref{th:identfunctions} offers an alternative interpretation in terms of ranks or units of the measures.
\end{remark}

\begin{remark}
	\label{r:AW}
	Instead of using a two-step approach, where the distribution of a linear combination of skills and investment is identified first, \citeN{AW:22} substitute measures into the production function equation and use IV arguments  with exogenous investment (i.e. $I_t$ and $\eta_{\theta,t}$ are independent in (\ref{eq:investment})). Aside from exogenous investment, there are no substantial differences between the required assumptions, as both approaches are based on the joint distribution of the measures. My main contributions are to study the roles of the scale and location restrictions on the parameters of the model, to show which restrictions select an element of the identified set, and to provide features that are invariant to them and are identified without age-invariant measures and restrictions on the production function. 
\end{remark}

\subsubsection{Non-invariant parameters}
\label{s:nonivariant_tl}

As shown above, under Assumption \ref{a:baseline},  there exist different sets of observationally equivalent parameters that cannot be distinguished using the data. I now illustrate with two examples that the resulting primitive parameters and counterfactuals might differ considerably.

\begin{example}
	\label{e:rescale}
 For simplicity, I assume that investment is observed and exogenous. In this case, Assumption \ref{a:ageinvariant_technology_investment}(a) holds. I first consider a data generating process (DGP) satisfying Assumptions \ref{a:baseline}, \ref{a:ageinvariant_technology_skills}(a), and \ref{a:ageinvariant_technology_skills}(b), but not Assumption \ref{a:normalization}. I then construct two alternative sets of parameters, both of which are observationally equivalent to the original DGP. One of these sets of parameters satisfies Assumptions \ref{a:baseline}, \ref{a:normalization}, and \ref{a:ageinvariant_technology_skills}(a) and the other set satisfies Assumptions \ref{a:baseline}, \ref{a:normalization}, and \ref{a:ageinvariant_technology_skills}(b). Specifically, first assume that
	\begin{eqnarray*}
		\ln \theta_{t+1} &=& \frac{1}{2} \ln \theta_{t} +  \frac{1}{2} \ln I_t + \eta_{\theta,t} \\
		\tilde{Z}_{\theta,t,1} &=&  \ln \theta_{t} + \tilde{\eps}_{\theta,t,1}   
	\end{eqnarray*}
	Here, I set $a_t = \mu_{\theta,t,m}  = \rho_{0} =  0$ and focus on the scale restriction in Assumption \ref{a:normalization} only. For brevity, I omit equations for investment, additional measurements, and the adult outcome.  
	
	Measures often do not have a natural scale. For example, we could divide all test scores by its standard deviation or we could measure education in months rather than years. One would then hope that changing the units of measurement does not affect the economic interpretation of the results. As a specific example, suppose we estimate the model using a scaled version of the measures, namely $Z_{\theta,t,1}= 12 \tilde{Z}_{\theta,t,1}$. Then for all $t$
	\begin{eqnarray*}
		\ln \theta_{t+1} &=& \frac{1}{2} \ln \theta_{t} +  \frac{1}{2} \ln I_t + \eta_{\theta,t} \\
		Z_{\theta,t,1} &=& 12 \ln \theta_{t} + {\eps}_{\theta,t,1} 
	\end{eqnarray*}
	where ${\eps}_{\theta,t,1}  = 12 \tilde{\eps}_{\theta,t,1} $. Without knowing the true DGP, one would typically impose assumptions that yield point identification when estimating the model. By Corollary \ref{c:pointident}, one could use either Assumption \ref{a:ageinvariant_technology_skills}(a) or \ref{a:ageinvariant_technology_skills}(b) next to Assumptions \ref{a:baseline} and \ref{a:normalization}. 
	
	I first construct parameters satisfying Assumptions \ref{a:baseline}, \ref{a:normalization}, and \ref{a:ageinvariant_technology_skills}(b) using Theorem \ref{th:obseq} in the appendix, which implies that there are $\{\bar{\theta}_{t}\}^T_{t=0}$ such that 
	\begin{eqnarray*}
		\ln \bar{\theta}_{t+1} &=&  \bar{\gamma}_{1t} \ln \bar{\theta}_{t} +  (1- \bar{\gamma}_{1t})  \ln I_t + \bar{\eta}_{\theta,t} \\
		Z_{\theta,t,1} &=&   \bar{\lambda}_{\theta,t,1} \ln \bar{\theta}_{t} + \eps_{\theta,t,1}  
	\end{eqnarray*}
	where  
	\begin{eqnarray*}
		(\bar{\gamma}_{10},\bar{\gamma}_{11},\bar{\gamma}_{12},\bar{\gamma}_{13},\bar{\gamma}_{14}, \ldots ) &=& (0.077, 0.351, 	0.435, 	0.470, 0.485,\ldots) \\
	(\bar{\lambda}_{\theta,0,1}, \bar{\lambda}_{\theta,1,1},\bar{\lambda}_{\theta,2,1},\bar{\lambda}_{\theta,3,1},\bar{\lambda}_{\theta,4,1}, \ldots ) &=& (1, 	6.5,	9.25,	10.625, 	11.3125, \ldots	)
	\end{eqnarray*}
			Moreover $\bar{\gamma}_{1t} \rightarrow {\gamma}_{1t} = 1/2$ and $\bar{\lambda}_{\theta,t,1}\rightarrow {\lambda}_{\theta,t,1} = 12$ as $t \rightarrow \infty$. Imposing the restriction $\bar{\lambda}_{\theta,0,1}= 1$ means that we estimate a model with alternative skills $\ln \bar{\theta}_{0} = 12 \ln \theta_{0}$ and that we obtain different parameters and skill distributions. Although these two models suggest very different dynamics, they generate identical measurements.	Since the true DGP is unknown, the primitive parameters are therefore hard to interpret. For example the coefficient in front of investment in the original model might be interpreted as: ``increasing investment by 1\%, increases skills in the next period by 0.5\% and the effect is the same for all time periods'' (see e.g. \citeN{AW:22} for these interpretations). Contrarily, one might interpret the coefficients in the alternative model with  Assumption \ref{a:normalization} as: ``increasing investment by $1\%$ in period $0$, increases skills in period $1$ by $0.923\%$ and increasing investment by $1\%$ in period $4$, increases skills in period $5$ by $0.515\%$'', suggesting that investment is more beneficial in early periods. To obtain interpretable parameters, notice that a one unit increase in $\ln(I_t)$ leads to an increase in skills that corresponds to a $\lambda_{\theta,t+1,1}\gamma_{2t} = 6$ units increase in $Z_{\theta,t+1,1}$ or a $1/2$ unit increase in $\tilde{Z}_{\theta,t+1,1}$. This effect  does not depend on the set of restrictions and its interpretation adapts to the units of measurement.

	The consequences of imposing Assumptions \ref{a:baseline}, \ref{a:normalization} and \ref{a:ageinvariant_technology_skills}(a)     are discussed in Section \ref{s:examples}, which shows that the production function becomes $	\ln \tilde{\theta}_{t+1} =  \frac{1}{2}  \ln \tilde{\theta}_{t} +  6  \ln I_t  + \tilde{\eta}_{\theta,t}$.

\end{example}

\begin{example}
	\label{e:levels_mc}
	 I now use a more involved DGP from Section \ref{s:montecarlo}, which is based on the simulations of \shortciteN{AMN:19}.\footnote{In Section \ref{s:montecarlo} I use a CES production function, as do \shortciteN{AMN:19}, but for this numerical examples I use a trans-log production function that leads to similar observed data.} In this setup $T=2$ and Assumptions \ref{a:baseline}, \ref{a:normalization}, \ref{a:ageinvariant_technology_skills}(a), and \ref{a:ageinvariant_technology_investment}(a) hold. Importantly, one skill and one investment measure have loadings of $1$.  Now suppose $Y_t$ represents income, we can exogenously increase the sum of income of each individual in periods $0$ and $1$  by two standard deviations, and we want to distribute that income optimally across the two periods. Denote the skills in period $2$ as a function of income by $\theta_2(Y_0 + w x, Y_1 + (1-w)x)$, where $Y_0$ and $Y_1$ are the original incomes, $x$ is the additional amount to be distributed, and $w$ is the share invested in period $0$. The left panel of Figure \ref{f:mean_share_translog} shows 
	$$\frac{E[\theta_2(Y_0 + w x, Y_1 + (1-w)x) ] - E[\theta_2(Y_0, Y_1)]}{sd(\theta_2(Y_0, Y_1))}$$ 
	as a function of $w$. That is, the $y$-axis shows the increase in the mean measured in standard deviations. The black line shows the effect using the true parameters, leading to an optimal income share of around $32\%$ in period $0$. Next, I multiply all skill measures $Z_{\theta,t,1}$ by $s_{\theta}$ and reestimate the model. 	Scaling the measures has identical effects to imposing $\lambda_{\theta,0,1} = 1$, when the data is generated with a loading of $s_{\theta}$.  The implied effects of increasing income for $s_{\theta} = 2/3$ and $s_{\theta} = 2$ can be seen in the left panel of Figure \ref{f:mean_share_translog} as well. Depending on the scale, we obtain inefficient optimal investment choices or erroneous benefits of investment.   Intuitively, the reason is that we now maximize $E[\tilde{\theta}_t(Y_0 + w x, Y_1 + (1-w)x)] = E[\theta_t(Y_0 + w x, Y_1 + (1-w)x)^{s_{\theta}}]$ which is not invariant to $s_{\theta}$. While such counterfactuals could be relevant if a welfare function is a function of the skill level, they are not identified without fixing the scales and locations  and then depend on the units of measurements of the data.

\begin{figure}[t!]
	\begin{center}
	\caption{Mean response for different weights}
	\label{f:mean_share_translog} 
	\includegraphics[trim=3cm 7.5cm 3cm 6.5cm, width=13.7cm]{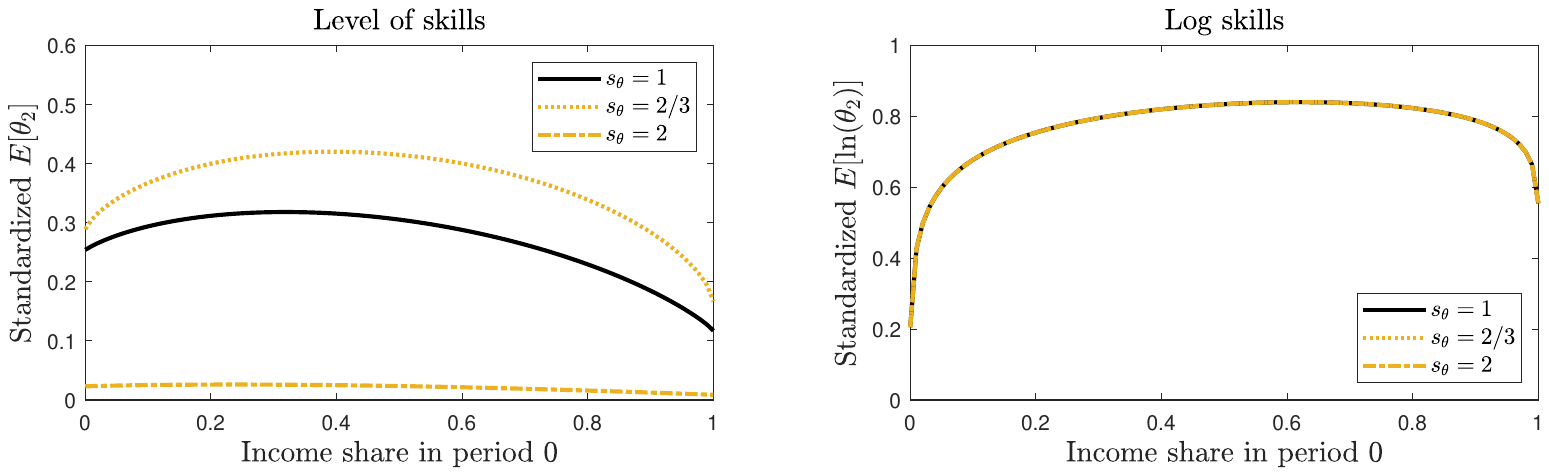}
	\end{center}
		\newgeometry{textwidth=16.4cm}
	{\footnotesize \begin{singlespace}  Notes: This figure shows changes in standardized means of the level (left panel) and the log (right panel) of skills for different income transfers and different scales of log-skills.  \end{singlespace}}
\end{figure}

	The right panel contains results for the log-skills rather than the level. As discussed below Theorem \ref{th:identfunctions}, in this case, the optimal income sequence does not depend on the scale of the measures, because the production function is sufficiently flexible in log-skills. In addition,  dividing by the standard deviation, implies that the objective   is  invariant to the scale.

\end{example}

\subsubsection{Anchoring }
\label{s:anchor_tl}

Consider equations (\ref{eq:prod_fn_rw})--(\ref{eq:anchor_eq_rw}) and recall that $\tilde{\mu}_{\theta,t,0} = \tilde{\mu}_{I,t,0} = 0$ and $\tilde{\lambda}_{\theta,t,0} = \tilde{\lambda}_{I,t,0} = 1$, and the parameters in this system of equations are point identified by Theorem \ref{th:identparams_tl}. Now define $\tilde{\vartheta}_{t}$ such that $\ln  \tilde{\vartheta}_{t} = \tilde{\rho}_0 + \tilde{\rho}_1 \ln  \tilde{\theta}_{t}$. We then get
\begin{eqnarray*}
	\ln \tilde{\vartheta}_{t+1} &=&  \tilde{\rho}_0 + \tilde{\rho}_1 \tilde{a}_t - \tilde{\rho}_0 \tilde{\gamma}_{1t}  + \tilde{\gamma}_{1t} \ln \tilde{\vartheta}_{t} +  \tilde{\rho}_1 \tilde{\gamma}_{2t} \ln I_t  + \tilde{\gamma}_{3t} \ln \tilde{\vartheta}_{t} \ln I_t +  \tilde{\rho}_1 \tilde{\eta}_{\theta,t} \\
	Z_{\theta,t,m} &=& \tilde{\mu}_{\theta,t,m} - \frac{ \tilde{\rho}_0 \tilde{\lambda}_{\theta,t,m}}{\tilde{\rho}_1} + \frac{\tilde{\lambda}_{\theta,t,m}}{\tilde{\rho}_1}\ln \tilde{\vartheta}_{t} + \eps_t \\
	Z_{I,t,m} &=& \tilde{\mu}_{I,t,m} +\tilde{\lambda}_{I,t,m} \ln \tilde{I}_{t} + {\eps}_{I,t,m} \hspace{26mm}   \\
	\ln \tilde{I}_t &=& \tilde{\beta}_{0t} - \frac{ \tilde{\rho}_0 \tilde{\beta}_{1t}}{\tilde{\rho}_1} + \frac{\tilde{\beta}_{1t} }{\tilde{\rho}_1} \ln \tilde{\vartheta}_{t}  + \tilde{\beta}_{2t} \ln Y_{t}   + \tilde{\eta}_{I,t}    \\	
	Q &=& \ln  \tilde{\vartheta}_{T}  + \eta_Q
\end{eqnarray*}
\citeN{CH:08} estimate a model for $\ln  \tilde{\vartheta}_{t}$ which anchors the skills at $Q$. Their identification strategy is equivalent to using Assumptions \ref{a:normalization} and \ref{a:ageinvariant_technology_skills}(a) and imposing $\rho_{0} = 0$ and $\rho_1 = 1$ instead of Assumption \ref{a:normalization}. Anchoring can help with the interpretation of certain parameters of the model. For example, when $E(Q \mid \ln  \tilde{\vartheta}_{T} ) = \ln  \tilde{\vartheta}_{T} $, then an increase of $\ln  \tilde{\vartheta}_{T}$ by one corresponds to a one unit increase in the conditional expectation of $Q$. However, as illustrated in Example \ref{e:levels_mc}, investment or income sequences that maximize the expected levels of skills depend on the units of measurements of the anchor.\footnote{\label{f:func_outcome}The identification arguments are fundamentally different if the anchoring equations was in levels instead of logs of the skills. That is, if $	Q =  {\rho}_{0} +  {\rho}_{1} {\theta}_{T} + \eta_Q   =  {\rho}_{0} +  \frac{{\rho}_{1} }{\exp(\mu_{\theta,T,1})^{1/\lambda_{\theta,T,1}} }\tilde{\theta}_T^{1/\lambda_{\theta,T,1}}   + \eta_Q    $. In this case it can be shown that the joint distribution of $(Q,\tilde{ \theta}_T)$ is identified, which identifies $\lambda_{\theta,T,1}$. Thus, the distribution of $ \theta_T$ is identified up to a scaling factors, which implies that even sequences of income or investment that maximize the level of skills are identified.}  

As discussed by \citeN{CH:08}, $\tilde{\gamma}_{1t}$ is invariant to the anchor. However, without Assumption \ref{a:ageinvariant_technology_skills}(a), $\tilde{\gamma}_{1t}$ is typically not equal to $\gamma_{1t}$.  The coefficient in front of investment is hard to interpreted because it depends on the specific anchor and its units of measurements (as also noted by \citeN{CH:08}). Moreover, skills can be anchored at the adult outcome, but not investment, and many of the estimated parameters also depend on the units of the investment measures. While anchoring makes most sense under Assumption \ref{a:ageinvariant_technology_skills}(a), in which case the skills in all periods are in the units of $Q$, we could instead use Assumption \ref{a:ageinvariant_technology_skills}(b) to achieve point identification. In this case different adult outcomes or different units of measurements  lead to different coefficients that possibly suggest very different dynamics, just as in Example \ref{e:rescale}.

\subsection{CES production function}
\label{s:ces}

I now discuss the CES production technology, where  
\begin{eqnarray}
	\theta_{t+1} &=& \left(\gamma_{1t} \theta_{t}^{\sigma_t} +   \gamma_{2t} I_t^{\sigma_t}\right)^{\frac{\psi_t}{\sigma_t}}\exp(\eta_{\theta,t})  \hspace{25mm} t = 0, \ldots, T-1 \label{eq:prod_fn_ces} \\
	Z_{\theta,t,m} &=& \mu_{\theta,t,m} + \lambda_{\theta,t,m} \ln \theta_{t} + \eps_{\theta,t,m} \hspace{27mm} t = 0, \ldots, T, m = 1,2 \label{eq:measurement_eq_ces} \\
	Z_{I,t,m} &=& \mu_{I,t,m} + \lambda_{I,t,m} \ln I_{t} + \eps_{I,t,m} \hspace{27mm} t = 0, \ldots, T-1, m = 1,2  \label{eq:measurement_eq_ces_invest}  \\
		\ln I_t &=& \beta_{0t} + \beta_{1t} \ln \theta_{t}  + \beta_{2t} \ln Y_{t}   + \eta_{I,t} \hspace{21mm} t = 0, \ldots, T-1  \\
	Q &=& \rho_{0} + \rho_{1} \ln \theta_{T} + \eta_Q \label{eq:anchor_eq_ces}
\end{eqnarray}
where $\sigma_t \neq 0$, $\gamma_{1t}\neq 0$, and $\gamma_{2t}\neq 0$ for all $t$. The measurement system is linear in $\ln \theta_{t}$ (as in \shortciteN{CHS:10} or \shortciteN{AMN:19}) to ensure that estimated skills are positive.

\subsubsection{Identification}

Similar to the trans-log case, define $\tilde{\theta}_t = \exp(\mu_{\theta,t,1})\theta_{t}^{\lambda_{\theta,t,1}}$ so that
$$\theta_{t}  = \exp \left(-\frac{\mu_{\theta,t,1} }{\lambda_{\theta,t,1}}\right)\tilde{\theta}_t^{\frac{1}{\lambda_{\theta,t,1}}}.$$ 
Similarly, let $\tilde{I}_t = \exp(\mu_{I,t,1})I_{t}^{\lambda_{I,t,1}}$. We can again rewrite the production technology in terms of $\tilde{\theta}_t$ and $\tilde{I}_t$. That is, we can rewrite equations (\ref{eq:prod_fn_ces})--(\ref{eq:anchor_eq_ces}) to 
\begin{eqnarray}
\tilde{\theta}_{t+1}  &=&  \left( \tilde{\gamma}_{1t}    \tilde{\theta}_t^{\frac{\sigma_{t} }{\lambda_{\theta,t,1}}} + \tilde{\gamma}_{2t}  \tilde{I}_t^{\frac{\sigma_{t} }{\lambda_{I,t,1}}}  \right)^{ \frac{\lambda_{\theta,t+1,1}\psi_t}{\sigma_{t}} } \exp( \tilde{\eta}_{\theta,t}) \hspace{13mm} t = 0, \ldots, T-1 \label{eq:prod_fn_ces_norm} \\
	Z_{\theta,t,m} &=& \tilde{\mu}_{\theta,t,m} + \tilde{\lambda}_{\theta,t,m} \ln \tilde{\theta}_{t} + \eps_{\theta,t,m} \hspace{38mm} t = 0, \ldots, T, m = 1,2  \label{eq:measurement_eq_rw_ces} \\
 \hspace{12mm}  Z_{I,t,m} &=& \tilde{\mu}_{I,t,m} +\tilde{\lambda}_{I,t,m} \ln \tilde{I}_{t} + {\eps}_{I,t,m} \hspace{38mm} t = 0, \ldots, T-1, m = 1,2 \label{eq:measurement_eq_invest_rw_ces}\\
\ln \tilde{I}_t &=& \tilde{\beta}_{0t} + \tilde{\beta}_{1t} \ln \tilde{\theta}_{t}  + \tilde{\beta}_{2t} \ln Y_{t}   + \tilde{\eta}_{I,t} \hspace{32mm} t = 0, \ldots, T-1  \label{eq:investment_rw_ces} \\
Q &=& \tilde{\rho}_{0} + \tilde{\rho}_{1} \ln \tilde{\theta}_{T} + \eta_Q   \label{eq:anchor_eq_rw_ces}
\end{eqnarray}
where  
$$  \tilde{\gamma}_{1t}  =  \gamma_{1t}   \exp\left( \sigma_{t} \left( \frac{ \mu_{\theta,t+1,1} }{ \lambda_{\theta,t+1,1}\psi_t } - \frac{\mu_{\theta,t,1}}{\lambda_{\theta,t,1}}  \right) \right) \quad  \text{and} \quad \tilde{\gamma}_{2t}  =  \gamma_{2t} \exp\left( \sigma_{t} \left( \frac{ \mu_{\theta,t+1,1} }{ \lambda_{\theta,t+1,1}\psi_t } - \frac{\mu_{I,t,1}}{\lambda_{I,t,1}}  \right) \right)  $$
and the other parameters are defined as in the trans-log case. Using the identified joint distribution of $(\{\ln \tilde{\theta}_{t}\}^T_{t=0}, \{\ln \tilde{I}_{t}\}^T_{t=0})$ and the restrictions imposed by the production function, the following theorem characterizes the identified set of the finite dimensional parameters.
 
\begin{theorem}
	\label{th:identparams_ces}	
	Suppose Assumption \ref{a:baseline} holds. 
	
	\begin{enumerate}[(a)]

	\item The identified set of  $\{\mu_{\theta,t,m},\lambda_{\theta,t,m}\}_{t=0,\ldots,T,m=1,2}$, $\{\mu_{I,t,m},\lambda_{I,t,m}\}_{t=0,\ldots,T-1,m=1,2}$, $(\rho_0,\rho_1)$,  $\{ \beta_{0t}, \beta_{1t}, \beta_{2t}\}^{T-1}_{t=0}$, $\{ \gamma_{1t},\gamma_{2t}, \sigma_{t}, \psi_t\}^{T-1}_{t=0}$ consists of all vectors that yield the same values of
$\{\tilde{\mu}_{\theta,t,m},\tilde{\lambda}_{\theta,t,m}\}_{t=0,\ldots,T,m=1,2}$, $\{\tilde{\mu}_{I,t,m},\tilde{\lambda}_{I,t,m}\}_{t=0,\ldots,T-1,m=1,2}$, $(\tilde{\rho}_0,\tilde{\rho}_1)$, $\{ \tilde{\gamma}_{1t},\tilde{\gamma}_{2t}, \frac{\sigma_{t} }{\lambda_{\theta,t,1}},\frac{\sigma_{t} }{\lambda_{I,t,1}} \}^{T-1}_{t=0}$, $\{\frac{\sigma_{t} \psi_t }{\lambda_{\theta,t+1,1}}\}^{T-1}_{t=0}$, and $\{ \tilde{\beta}_{0t}, \tilde{\beta}_{1t}, \tilde{\beta}_{2t}\}^{T-1}_{t=0}$ as the true parameter vectors.

\item Let  $\{\bar{\mu}_{\theta,t,1} \}_{t=0}^{T}$, $\{\bar{\mu}_{I,t,1} \}_{t=0}^{T-1}$, and $\{\bar{\lambda}_{\theta,t,1} \}_{t=0}^{T}$ be fixed constants with $\bar{\lambda}_{\theta,t,1} \neq 0$ for all $t$.  Under the additional restrictions $\{{\mu}_{\theta,t,1} \}_{t=0}^{T} = \{\bar{\mu}_{\theta,t,1} \}_{t=0}^{T}$ and $\{{\mu}_{I,t,1} \}_{t=0}^{T-1} = \{\bar{\mu}_{I,t,1} \}_{t=0}^{T-1}$ and $\{{\lambda}_{\theta,t,1} \}_{t=0}^{T} = \{\bar{\lambda}_{\theta,t,1} \}_{t=0}^{T}$
the identified set is a singleton.

	\end{enumerate} 
	
\end{theorem}

An important implications of part (a) is that the fraction $\frac{\lambda_{\theta,t,1}}{\lambda_{I,t,1}} = \frac{\lambda_{\theta,t,1}}{\sigma_t}\frac{\sigma_t}{\lambda_{I,t,1}}$ is  point identified for all $t = 1,\ldots, T-1$. Hence, if we restrict $\lambda_{\theta,t,1}$ to a constant, $\lambda_{I,t,1}$ is identified, which is very different to the trans-log case.   Intuitively, since skills and investment have the same exponent in the CES production functions, the relative scale is identified by the functional form restrictions.

As before, I now introduce additional assumptions to achieve point identification.
 
\setcounter{assumptionp}{1}

\begin{assumptionp} \label{a:normalization_ces} 
$\lambda_{\theta,0,1}=1$  and $\mu_{\theta,0,1} = 0$. 
\end{assumptionp}
 
\begin{assumptionp}\label{a:ageinvariant_technology_skills_ces} \quad
	
	\begin{enumerate}[\label=(a)]
		
		\item   $\mu_{\theta,t,1} = \mu_{\theta,t+1,1}$ for all $t = 0, \ldots, T-1$  
		\item $\gamma_{1t} + \gamma_{2t} =1 $ for all $t=0,\ldots,T-1$.
	\end{enumerate}
	
\end{assumptionp}

\begin{assumptionp}\label{a:ageinvariant_technology_investment_ces}\quad
	\begin{enumerate}[\label=(a)]
		\item 	$\mu_{I,t,1} = \mu_{I,t+1,1} = 0$ for all $t = 0, \ldots, T-2$.
		\item 	$\beta_{0t} = 0$  for all $t=0,\ldots,T-1$.
	\end{enumerate}

\end{assumptionp}

\begin{assumptionp}\label{a:add_restrictions_ces}\quad
	\begin{enumerate}[\label=(a)]
		\item  $\lambda_{\theta,t,1} = \lambda_{\theta,t+1,1}$ for all $t = 0, \ldots, T-1$. 
		\item $\psi_{t} = 1$  for all $t=0,\ldots,T-1$.
	\end{enumerate}

\end{assumptionp}

The following result shows how point identification can be established.

\begin{corollary}
	\label{c:ces_ident}
	Suppose Assumptions \ref{a:baseline} and \ref{a:normalization_ces} hold. Suppose either Assumption \ref{a:ageinvariant_technology_skills_ces}(a) or Assumption \ref{a:ageinvariant_technology_skills_ces}(b) holds. Suppose either Assumption \ref{a:ageinvariant_technology_investment_ces}(a) or Assumption \ref{a:ageinvariant_technology_investment_ces}(b) holds. Suppose either Assumption \ref{a:add_restrictions_ces}(a) or Assumption \ref{a:add_restrictions_ces}(b) holds. Then all parameters are point identified. 
\end{corollary}

A common restriction with the CES production function is to set $\psi_{t} = 1$ for all $t$ (see e.g. \shortciteN{CHS:10} and \shortciteN{AMN:19}). In this case, Theorem \ref{th:identparams_ces} implies that   $\frac{\lambda_{\theta,t+1,1}}{\lambda_{\theta,t,1}}$ is point identified. Hence, assuming age-invariance (i.e. $\lambda_{\theta,t+1,1}  = \lambda_{\theta,t,1}$) is not required, and it is in fact testable. Moreover, with $\psi_{t} = 1$,  Corollary \ref{c:ces_ident} implies that the only scale restriction needed is $\lambda_{\theta,0,1} = 1$ (or alternatively $\lambda_{I,0,1}=1$). Nevertheless, it is common practice to set   ${\lambda}_{\theta,t,1}= {\lambda}_{I,t,1} = 1$ for all $t$, which are not normalizations, but restrictive assumptions (even if $\psi_t \neq 1$).  Setting the scales to different numbers or changing the units of measurement of the data then affects all estimated parameters, optimal investment sequences, and other counterfactuals. Similarly, if the scale of investment is fixed, we cannot anchor the skills at the adult outcome, unless $\rho_1 = 1$, which can only be true for one specific unit of measurements (see Appendix \ref{s:anchor_ces}). The exact consequences of imposing unnecessary   restrictions depend on the estimation methods used, and I discuss specific examples in Sections \ref{s:montecarlo} and \ref{s:application}.

As shown in Theorem \ref{th:obseq_ces} in the appendix and illustrated in examples below, if the model is correctly specified and Assumption \ref{a:baseline} holds, then there always exist sets of parameters which are consistent with the data and satisfy Assumptions \ref{a:baseline}, \ref{a:normalization_ces}, either \ref{a:ageinvariant_technology_skills_ces}(a) or \ref{a:ageinvariant_technology_skills_ces}(b), either \ref{a:ageinvariant_technology_investment_ces}(a) or \ref{a:ageinvariant_technology_investment_ces}(b), and either \ref{a:add_restrictions_ces}(a) or \ref{a:add_restrictions_ces}(b). Similar to the trans-log case, different sets of assumptions yield observationally equivalent models with potentially very different primitive parameters.

\subsubsection{Invariant parameters}

As in the trans-log case, important policy relevant parameters are point identified under Assumption \ref{a:baseline} only, because they can all be written in terms of identified features. Similar to before, now define  $s_{1t}(\alpha_1,\alpha_2,\alpha_3) =   \left(\gamma_{1t} Q_{\alpha_1}(\theta_t)^{\sigma_t}   +  \gamma_{2t} Q_{\alpha_2}(I_t)^{\sigma_t} \right)^{\frac{\psi_t}{\sigma_t}}\exp(Q_{\alpha_3}(\eta_{\theta,t}))  $ and $s_{2t}(\alpha_1,\alpha_2,\alpha_3,y) =  \left(\gamma_{1t} Q_{\alpha_1}(\theta_t)^{\sigma_t}   +  \gamma_{2t} I_t(y)^{\sigma_t} \right)^{\frac{\psi_t}{\sigma_t}}\exp(Q_{\alpha_2}(\eta_{\theta,t}))$  where $\ln I_t(y) = \beta_{0t} + \beta_{1t} \ln Q_{\alpha_1}(\theta_t)  + \beta_{2t} \ln y    +  Q_{\alpha_3}(\eta_{I,t}) $.

\begin{theorem}
	\label{th:identfunctionsces}
	Suppose Assumption \ref{a:baseline}	holds. 
	
	\begin{enumerate}
	 
		\item 
		
		$F_{\theta_{t+1}}\big(s_{1t}(\alpha_1,\alpha_2,\alpha_3)  \big)$ and	$\mu_{\theta,t+1,m} +  \lambda_{\theta,t+1,m}\ln s_{1t}(\alpha_1,\alpha_2,\alpha_3) +  Q_{\alpha_4}(\eps_{\theta,t+1,m})$  are point identified for all $\{\alpha_j\}^4_{j=1}\in (0,1)$. Moreover, $\frac{\partial \lambda_{\theta,t+1 ,m} \ln \theta_{t+1}}{\partial \lambda_{\theta,t ,m'} \ln \theta_t } \mid_{\theta_t = Q_{\alpha_1}(\theta_t), I_t = Q_{\alpha_2}(I_t)  }$ and $\frac{\partial \lambda_{\theta,t+1 ,m} \ln \theta_{t+1}}{\partial \lambda_{I,t ,m'} \ln I_t } \mid_{ \theta_t = Q_{\alpha_1}(\theta_t),I_t = Q_{\alpha_2}(I_t)  }$ are point identified for all $m,m'$ and $\{\alpha_j\}^2_{j=1}\in (0,1)$.

		\item   $F_{ \theta_{t+1}}(s_{2t}(\alpha_1,\alpha_2,\alpha_3,y) )$ and 	$\mu_{\theta,t+1,m} +  \lambda_{\theta,t+1,m} \ln s_{2t}(\alpha_1,\alpha_2,\alpha_3,y) +  Q_{\alpha_4}(\eps_{\theta,t+1,m})$  are point identified $\{\alpha_j\}^4_{j=1}\in (0,1)$.  
		\item $  P\left(Q \leq q   \mid \theta_t= Q_{\alpha_1}(\theta_t), \{I_s = Q_{\alpha_{2s}}(I_s) \}^{T-1}_{s=t}, \{\eta_{\theta,s} = Q_{\alpha_{3s}}(\eta_{\theta,s})   \}^{T-1}_{s=t}  \right)$ is point identified for all $\alpha_1,\{\alpha_{2s},\alpha_{3s} \}^{T-1}_{s=t} \in (0,1)$.  
		
		\item $  P\left(Q \leq q   \mid \theta_0= Q_{\alpha}(\theta_0), \{Y_t = y_t \}^{T-1}_{t=0}     \right)$ is point identified for all $\alpha \in (0,1)$.

		\item If in addition either Assumption \ref{a:add_restrictions_ces}(a) or Assumption \ref{a:add_restrictions_ces}(b) holds,  the distributions of 
		$$\frac{\partial \ln \theta_{t+1} }{\partial \ln \theta_{t} } = \frac{\partial  }{\partial \ln \theta_{t} } \ln \left(\gamma_{1t} \theta_{t}^{\sigma_t} +   \gamma_{2t} I_t^{\sigma_t}\right)^{\frac{1}{\sigma_t}} \; \text{ and } \;  \;   \frac{\partial \ln \theta_{t+1} }{\partial \ln I_{t} } = \frac{\partial  }{\partial \ln I_{t} } \ln \left(\gamma_{1t} \theta_{t}^{\sigma_t} +   \gamma_{2t} I_t^{\sigma_t}\right)^{\frac{1}{\sigma_t}}  $$
		are point identified and
		$$\left.\frac{\partial \ln \theta_{t+1} }{\partial \ln \theta_{t} } \right|_{ \ln \theta_{t} = Q_{\alpha_1}(\ln \theta_{t}), \ln I_{t} = Q_{\alpha_2}(\ln I_{t})}  = \left.\frac{\partial  }{\partial \ln \theta_{t} } \ln  \left(\gamma_{1t} \theta_{t}^{\sigma_t} +   \gamma_{2t} I_t^{\sigma_t}\right)^{\frac{1}{\sigma_t}}  \right|_{ \ln \theta_{t} = Q_{\alpha_1}(\ln \theta_{t}), \ln I_{t} = Q_{\alpha_2}(\ln I_{t})}   $$
		and
		$$\left.\frac{\partial \ln \theta_{t+1} }{\partial \ln I_{t} } \right|_{ \ln \theta_{t} = Q_{\alpha_1}(\ln \theta_{t}), \ln I_{t} = Q_{\alpha_2}(\ln I_{t})}  = \left.\frac{\partial  }{\partial \ln I_{t} }  \ln \left(\gamma_{1t} \theta_{t}^{\sigma_t} +   \gamma_{2t} I_t^{\sigma_t}\right)^{\frac{1}{\sigma_t}}  \right|_{ \ln \theta_{t} = Q_{\alpha_1}(\ln \theta_{t}), \ln I_{t} = Q_{\alpha_2}(\ln I_{t})}   $$
		are point identified for all $\alpha_1, \alpha_2 \in (0,1)$.
		
	\end{enumerate}

\end{theorem}

  The features in parts (1)--(4) are analogous to those in the trans-log case and can be used to calculate optimal investment/income strategies and anchored treatment effects, as discussed after Theorem \ref{th:identfunctions}. As opposed to the trans-log case, now the relative scales of skills and investment are point identified. Consequently, elasticities are identified under Assumption \ref{a:baseline} and either age-invariant skill measures or $\psi_t = 1$ only.

\subsubsection{Non-invariant parameters}

To achieve point identification of all parameters, we need to fix the levels of the logs of skills and investment, e.g. by setting $\mu_{\theta,0,1}$ and $\mu_{0,I,1}$ to $0$. Similar to the trans-log case, these restrictions are generally not normalizations for the primitive parameters. In the example below, I illustrate that with the restriction $\gamma_{1t} + \gamma_{2t} =1$, setting $\mu_{\theta,0,1} = 0$   can imply different dynamics. Here $\mu_{\theta,0,1}$ (and not $\lambda_{\theta,0,1}$) affects the scale because we can identify the distribution of $\tilde{\theta}_{t} = \exp(\mu_{\theta,t,1})\theta_t^{\lambda_{\theta,t,1}} $ and the production function is in levels rather than logs of $\theta_t$.

\begin{example}
	
The issues with the restriction $\mu_{\theta,0,1} = 0$ in the CES case are analogous to the  issues with the restriction $\lambda_{\theta,0,1} = 1$ in the trans-log case discussed in Example \ref{e:rescale}. I now consider a numerical example analogous to Example \ref{e:rescale} and focus on the measurement and the production function only. That is, suppose that 
\begin{eqnarray*}
	\theta_{t+1} &=& \left(\frac{1}{2} \theta_{t} +  \frac{1}{2}   I_t\right) \exp(\eta_{\theta,t})  \\
	Z_{\theta,t,m} &=& \ln(12) + \ln \theta_{t} + \eps_{\theta,t,m} 
\end{eqnarray*}
Notice that $Z_{\theta,t,m} =  \ln (12 \theta_{t}) + \eps_{\theta,t,m}$. Let $\tilde{\theta}_t = 12 \theta_{t}$. Just as in Example \ref{e:rescale}, we get
\begin{eqnarray*}
	\tilde{\theta}_{t+1} &=&  \left(\tilde{\gamma}_{1t}   \tilde{\theta}_{t} +  (1- \tilde{\gamma}_{1t})   I_t\right)\exp(\tilde{\eta}_{\theta,t}) \\
	Z_{\theta,t,m} &=&  \tilde{\mu}_{\theta,t,m} + \ln \tilde{\theta}_{t} + \eps_{\theta,t,m}  		 
\end{eqnarray*}
where  
$(\tilde{\gamma}_{10},\tilde{\gamma}_{11},\tilde{\gamma}_{12},\tilde{\gamma}_{13},\ldots ) = (0.077, 0.351, 	0.435, 	0.470,\ldots) $
and
$$(\exp(\tilde{\mu}_{0,\theta,m}),\exp(\tilde{\mu}_{1,\theta,m}),\exp(\tilde{\mu}_{2,\theta,m}),\exp(\tilde{\mu}_{3,\theta,m}), \ldots ) = (1, 	6.5,	9.25,	10.625,  \ldots	).$$	
Just as in Example \ref{e:rescale}, these two models are observationally equivalent, but suggest very different dynamics.

\end{example}

\section{Monte Carlo simulations}
\label{s:montecarlo}

I use a very similar data generating process as \shortciteN{AMN:19}. In particular, I use 
$$
\theta_{t+1}= A_{t}\left(\gamma_{t} \theta_{t}^{\sigma_{t}}+\left(1-\gamma_{t}\right) I_t^{\sigma_{t}}\right)^{\frac{1}{\sigma_{t}} }\exp(\eta_{\theta,t})
$$
for $t=1,2$. Allowing for $A_t \neq 1$ is equivalent to not imposing the restriction that the coefficients in front of $\theta_t$ and $I_t$ sum to $1$. Since I want to study income transfers as counterfactuals, I augment the setup of \shortciteN{AMN:19} and add $\ln I_t = \beta_{1t} \ln \theta_t + \beta_{2t} \ln Y_t  + \eta_{I,t}$, where $Y_0 = Y_1$. To simulate data, I first draw $(\ln(\theta_{0}),\ln(Y))$ from a normal mixture distribution. Given $(\ln(\theta_{0}),\ln(Y))$ and normally distributed $\eta_{\theta,t}$ and $\eta_{I,t}$  I then generate $I_0$, $\theta_1$,  $I_1$, and $\theta_2$ using the model. If $\ln I_t$ was equal to $\ln Y_t$, the setup would be exactly equal to that of \shortciteN{AMN:19} with parameters as in their Table 9, which are based on their empirical results, and I use $\sigma_{0} = \sigma_{1} = -0.5$.\footnote{I use slightly different notation to be consistent with the notation above. Specifically, the periods are $t = 0,1,2$ instead of $t = 1,2,3$, I use $I_t$ instead of $X_t$ for the second latent variable, and I use $\sigma_t$ instead of $\rho_t$ to denote the elasticity of substitution. The distribution of $(\ln(\theta_{0}),\ln(Y))$  is the same as the distribution of $(\ln(\theta_{0}),\ln(X))$ in \shortciteN{AMN:19}.}  I deviate slightly from their setting by using the additional investment equation with $\beta_{1t} = 0.1$, $\beta_{2t} = 0.9$, and $\eta_{I,t}  \sim N(0,0.1^2)$. I simulate three measures each for $\theta_t$ and $I_t$, which have a factor structure with $\mu_{\theta,t,m} = \mu_{I,t,m} = 0$ for all $m,t$. In addition $\lambda_{\theta,t,1} = \lambda_{I,t,1} = 1$ for all $t$, which imposes the scale restrictions and the age-invariance assumptions. In \shortciteN{AMN:19}, all of these loadings are also ``normalized'' to $1$ in their estimation procedure. In particular, they first estimate the distribution of the measures and of log-income using a normal mixture model. Then, assuming that $(\ln(\theta_{0}),\ln(\theta_{1}),\ln(\theta_{2}),\ln(I_0),\ln(I_1),\ln(Y))$ also has a normal mixture distribution,  they use the factor structure and the restrictions to estimate that distribution.\footnote{Interestingly, one can see from the DGP that $(\ln(\theta_{0}),\ln(\theta_{1}),\ln(\theta_{2}))$ does not have a normal mixture distribution due to the nonlinear production function, but this misspecification bias seems to be relatively small in this simulation setup.} Lastly, they take draws from the distribution and estimate the production function parameters by nonlinear least squares. I implement their estimator and my more flexible approach.

For my approach, I use the estimation procedure explained in Appendix \ref{s:estimation_ces}, which simplifies in this setup because investment is exogenous (and thus, $\kappa_t = 0$). Moreover, to focus on the scale restriction only, I set $\mu_{\theta,t,m} = \mu_{I,t,m} = 0$ for all $m$ and $t$.   Finally, I impose that the first skill measure and the first investment measure are age-invariant. I  then set $\lambda_{I,t,1} = 1$ for all $t$, and estimate $\lambda_{\theta,t,1}$. I therefore impose Assumptions \ref{a:normalization_ces}, \ref{a:ageinvariant_technology_skills_ces}(a), and \ref{a:ageinvariant_technology_investment_ces}(a), as well as both parts of Assumption \ref{a:add_restrictions_ces}. While only one part of the last assumption is needed (i.e. age-invariance of the measures could be dropped), they are both satisfied in this setup.  I then estimate  $\lambda_{\theta} =  \lambda_{\theta,t,1}$ along with the production function parameters by solving
\begin{align*}
	\argmin_{ {\lambda}_{\theta,1},  {\sigma}_1, {\sigma}_2, {\gamma}_{11},{\gamma}_{21}, {\gamma}_{11},{\gamma}_{21}}   \sum^T_{t=1} \sum^J_{j=1}  \left(\ln \tilde{\theta}_{t+1,j} -   \ln(A_t)  \lambda_{\theta} - { \frac{ \lambda_{\theta}  }{\sigma_{t}} } \ln \left(  {\gamma}_{1t}    \tilde{\theta}_{t,j}^{\frac{\sigma_{t}}{\lambda_{\theta}}} +  (1-{\gamma}_{2t})  \tilde{I}_{t,j}^{\sigma_{t} }  \right)  \right)^2 
\end{align*}
where $\tilde{\theta}_{t,j}$ and $\tilde{I}_{t,j}$ are draws from the estimated distribution of skills and investment using the estimator of \shortciteN{AMN:19}. Similarly, we can estimate $\beta_{1t}$ and $\beta_{2t}$ from a regression of $\tilde{I}_{t,j}$ on $\tilde{\theta}_{t,j}$ and $Y_{j}$. 

In the following, I will investigate the effect of multiplying the skill measures by a single constant $s_{\theta}$ in all periods, which changes the loadings, but not the age-invariance assumption. For example, the first measure, say $\tilde{Z}_{\theta,t,1}$, is generated by 
$$\tilde{Z}_{\theta,t,1} = \log(\theta_{t}) + \eps_{\theta,t,1}$$
but when I estimate the model, I use 
$$Z_{\theta,t,1} = s_{\theta} \tilde{Z}_{\theta,t,1} = s_{\theta}\log(\theta_{t}) + s_{\theta}\eps_{\theta,t,1},$$
which is also an age-invariant measure. The estimators still impose that the loadings are equal to 1 ($\lambda_{I,t,1} = 1$ for all $t$  with my estimator and $\lambda_{\theta,t,1} = 1$ and $\lambda_{I,t,1} = 1$ for all $t$  with the estimator of \shortciteN{AMN:19}). Of course, in practice, we do not know the DGP and there is no reason to believe that the true loading is $1$. Ideally, the restriction should be a normalization in which case the results would be invariant to scaling the data or changing its units of measurement. However, Corollary \ref{c:ces_ident} implies that the estimator of \shortciteN{AMN:19} is inconsistent. I consider the implications for elasticities and counterfactual predictions, which are point identified (as shown in Theorem \ref{th:identfunctionsces}) and are invariant to the scaling when using the new estimator.\footnote{These features are also invariant to scaling the investment measures with my estimator and I could have set $\lambda_{\theta,t,1}$ instead of $\lambda_{I,t,1}$ to 1.}  I take $s_{\theta} = 1$, which is the correctly specified model, as well as $s_{\theta} = 2/3$ and $s_{\theta} = 2$. I report average estimates from 500 simulated data sets. 

To summarize the production function estimates I report $ {\partial \ln(\theta_{1}) }/{\partial \ln(\theta_{0})}$ as a function of quantiles of $\ln(\theta_{0})$ and evaluated at the 25th and 75th percentile of $I_0$ as well as $ {\partial \ln(\theta_{2}) }/{\partial \ln(I_1)}$ as a function of quantiles of $\ln(I_1)$ and evaluated at the 25th and 75th percentile of $\theta_{1}$.  Figure \ref{f:marg_eff_1} displays these partial derivatives for the true parameters, the 2-step estimator of \shortciteN{AMN:19} with different values of $s_{\theta}$, and the invariant estimator that also estimates the scale. In this and the following figures, the results obtained with the invariant estimator are always almost identical to those with the 2-step estimator and $s_{\theta} = 1$ and very similar to those with the true parameter values. Moreover, the invariant estimator adapts to the scale change. Contrarily, the figure illustrates that small changes in the units of measurements can have a large effect on the results when using the 2-step estimator. One interesting finding is that the 2-step estimator underestimates the partial derivative with respect to $\ln(\theta_0)$ at the  25th percentile of $I_0$ for $s_{\theta} = 2/3$ and overestimates it for $s_{\theta} = 2$.

\begin{figure}[t!]
	\begin{center}
		\caption{Partial derivatives}
	\label{f:marg_eff_1}
	\includegraphics[trim=3cm 5cm 3cm 1.2cm, ,clip, width=16cm]{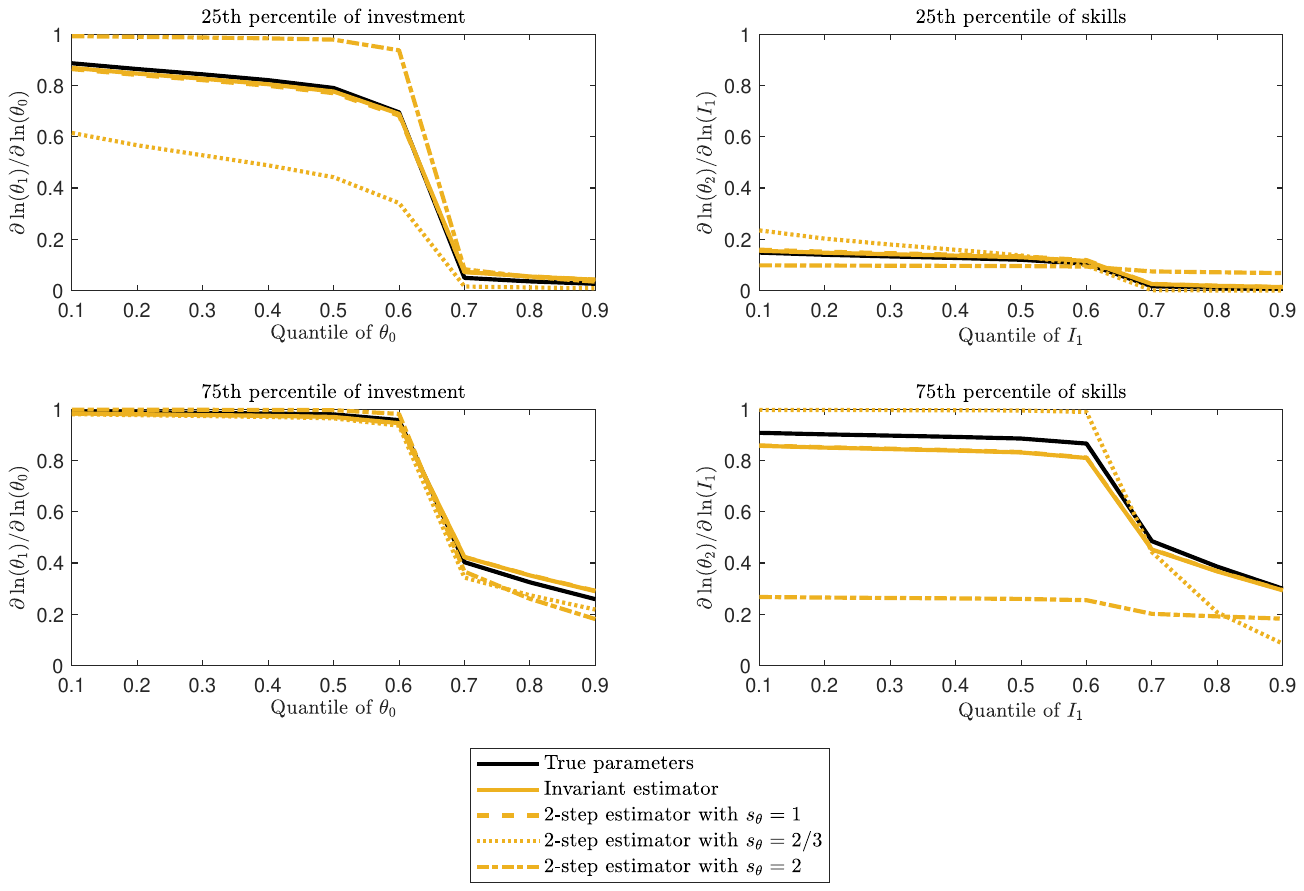}
		\end{center}
\newgeometry{textwidth=16.4cm}
\vspace{-7mm}
{\footnotesize \begin{singlespace}  Notes: This figure shows $ {\partial \ln(\theta_{1}) }/{\partial \ln(\theta_{0})}$ as a function of quantiles of $\ln(\theta_{0})$ and evaluated at the 25th and 75th percentile of $I_0$ as well as $ {\partial \ln(\theta_{2}) }/{\partial \ln(I_1)}$ as a function of quantiles of $\ln(I_1)$ and evaluated at the 25th and 75th percentile of $\theta_{1}$ using different estimators.  \end{singlespace}}
\end{figure}

As a first counterfactual, consider an individual, whose value of $\theta_{0}$ equals $Q_{0.1}(\theta_0)$ and all unobservables are $0$. I then exogenously change the income sequence of that individual and check the implied quantile in the skill distribution in period $2$. Of course, the larger income, the higher the relative rank/quantile in the last period. I report results where a feasible choice of income in each period is a given quantile (that will be varied). Instead of using the feasible choices, I distribute the total income among the two periods to maximize the skill quantile in the last period. The second part of Theorem \ref{th:identfunctionsces} implies that these counterfactuals are point identified, can be consistently estimated with the invariant estimator, but the results with the 2-step estimator of \shortciteN{AMN:19} will depend on $s_{\theta}$.

\begin{figure}[t]
	\begin{center}
	\caption{Response to income changes and optimal shares }
	\label{f:opt_inv_1}
	\includegraphics[trim=3cm 8cm 3cm 3.5cm, clip, width=16cm]{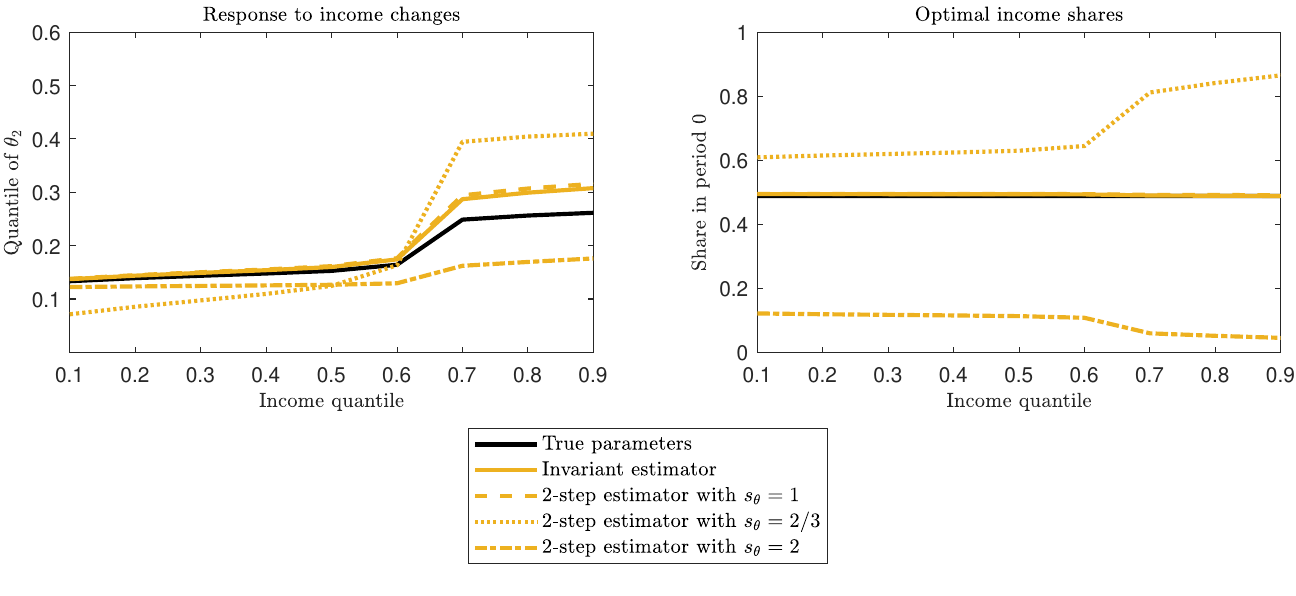}
	\end{center}
\newgeometry{textwidth=16.4cm}
	\vspace{-4mm}
	{\footnotesize \begin{singlespace}  Notes: The figure shows the estimated effects of increases in income on the quantiles of skills in the last periods using different estimators. The left panel shows the quantiles for different income levels and the right panels shows the corresponding optimal income shares in period $0$.  \end{singlespace}}
\end{figure}

The left panel of Figure \ref{f:opt_inv_1} shows the quantile as a function of the feasible income quantile ($0.5$ is the median, etc.).  Using the true parameters, we can see that, even for large income, the quantile in period $2$ is at most around $0.26$  and is almost flat below the median.  The right panel of Figure \ref{f:opt_inv_1} displays the corresponding optimal income shares in period $0$. With the true parameters, income should be similar in both periods. The 2-step estimator with $s_{\theta} =1$ and the invariant estimator (irrespective of the scale) yield similar conclusions for the optimal income sequence, but have a small bias for the estimated quantile in period $2$ (due the approximation error of the joint distribution of the measures/skills as mixtures of normals). Figure \ref{f:opt_inv_1} also demonstrates that changing the units of measurement can lead to inefficient income transfers when using the 2-step estimator. For example, with $s_{\theta} = 2/3$ the results suggest that we should mainly invest in $t=0$, and with $s_{\theta} = 2$ it appears that we should mainly invest in $t=1$. Moreover, the estimated gains of income transfers are misleading in this case. The inconsistent estimates in the left panel suggest that large income transfers can increase the quantile to almost $0.5$ in period $2$ when $s_{\theta} = 2/3$. Contrarily, with  $s_{\theta} = 2$ we would underestimate the effect. Importantly, the results for the new estimator  are invariant to changes in the units of measurements.

\begin{figure}[b!]
		\begin{center}
	\caption{Outcome distributions under different investment sequences}  
	\label{f:count_inv_l}  
	\includegraphics[trim=2cm 3cm 0cm 3.4cm, width=18cm]{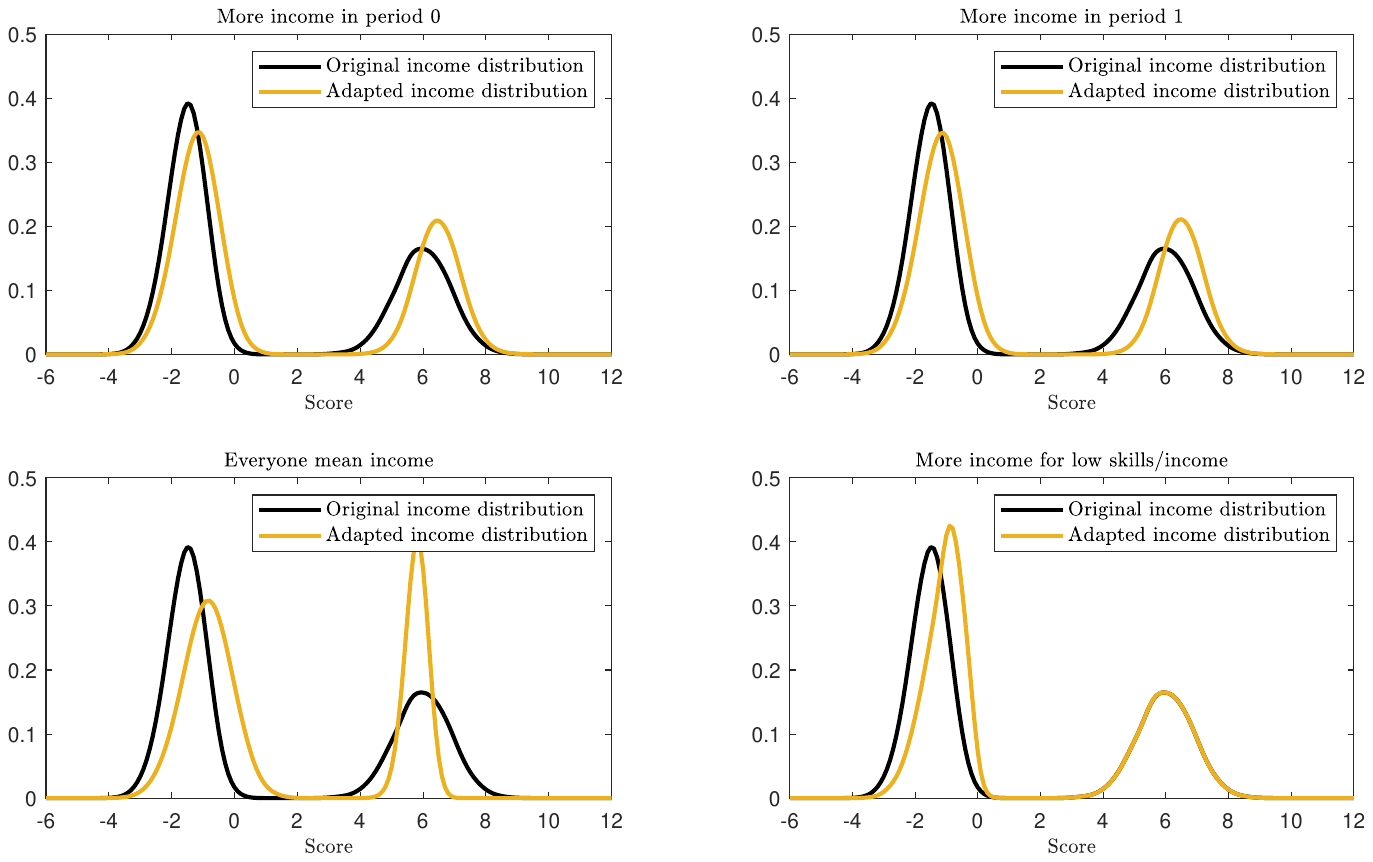}  
		\end{center}
	\newgeometry{textwidth=16.4cm}
\vspace{-7mm}
{\footnotesize \begin{singlespace}  Notes: The figure shows counterfactual outcome distributions. Each panel displays the distribution with the original income distribution and a counterfactual income distribution.  \end{singlespace}}
\end{figure}  

Next, I consider how exogenous income changes affect the skill distribution. To do so, I take draws from the estimated joint distribution of income and  skills in period $0$  (based on the average of the estimated parameters to get representative results) and consider four counterfactual marginal income distributions. First, I increase everyone's income by two standard deviations in period 0. Second, I increase everyone's income by two standard deviations in period 1. Third, I set income to the median for everyone in both periods. Fourth, I increase income by two standard deviations in both periods, but only if the initial skill and income quantiles are below $0.5$. I set all unobservables to their median values. I report results for the invariant estimator only. 
 
Figure \ref{f:count_inv_l} shows the implied distributions of one of the skill measures in the final period (which could be a test score or an adult outcome in an application). These results depend on and should be interpreted relative to the units of measurements of that measures. Figure \ref{f:count_inv_l} is based on $s_{\theta} = 1$. One can see that increasing income in either of the first two periods has very similar effects and leads to an increase in skills. Since everyone's income increases, everyone is better off. If everyone receives the mean income, the variance of the outcome distribution decreases considerably. If we only increase income for people with low initial skills and income, then predictably only the lower tail of the distribution will be affected.

\section{Empirical illustration}
\label{s:application}

In this section, I illustrate the previous findings by replicating some of the estimates of \shortciteN{AMN:19} and showing how changing the units of measurements of the data (or, equivalently, changing the scale restriction) affects parameter estimates and counterfactuals.  

\shortciteN{AMN:19} estimate production functions for cognitive skills and 
health of children with past skills and health as well as parental investment, health, and skills as inputs. They use the CES production function with $\psi_t = 1$. Theorem 3 and Corollary 2 imply that most scale parameters are identified and one only has to restrict the loading of \underline{one} latent variable in \underline{one} time period to 1. \shortciteN{AMN:19} set a loading for \underline{each}  observed factor in \underline{each} time period to 1. Specifically, they write on page 2520: ``One way to meet this condition is to normalize each factor on the same measure every period, assuming that the mapping from measure to factor is invariant with respect to the age of the subject. Fortunately, our data are sufficiently rich that we are able to do this for our model. For child cognitive skills, we always normalize the loading on PPVT to one. Similarly, child health is always normalized on height z scores, investments are normalized on amount spent on books, parental health is normalized on mother’s weight, parental cognitive skills is normalized on mother’s years of schooling, and resources are normalized on income.'' These restrictions already take the considerations of Agostinelli and Wiswall (2016a, 2024) into account. 

\nocite{AW:16a} \nocite{AW:16b}

I consider two changes of the units of measurements. First, I scale book expenditures used to anchor investment. In the descriptive statistics in their Table 2, \shortciteN{AMN:19} report expenditures in USD, but estimation is based on a standardized version of the expenditures in INR.\footnote{While not discussed in their paper or in the code, their data set contains a standardized book expenditures variable, which turns out to be expenditures in INR, but standardized such that the mean and the variance of the pooled expenditures over all time periods are 0 and 1, respectively.} Since the authors write ``we use the same measurements at every age and we normalize on the amount spent on books'' (p. 2524) and ``our investment measure is measured in monetary units, reflecting cost'' (p. 2526), it might not be clear to readers which units the results are based on. Ideally, scaling a measure in this way does not affect the main results, but I show that some of the main results are sensitive to using these standardized expenditures instead of book expenditures in 100 INR.\footnote{This transformation amounts to multiplying the standardized test scores by $7.833$ and adding $4.802$.} Second the estimates in the paper are based on a standardized PPVT test score. I reestimate the model after multiplying all test scores by 3, which implies a variance of the scaled scores that is still well below the variance of the actual PPVT test scores. I then compare some of the estimates of \shortciteN{AMN:19} with the ones obtained using the two scaled versions of the data.\footnote{In the first step, \shortciteN{AMN:19} estimate the joint distribution of skills and then simulate from the estimated distribution to estimate production functions and to calculate counterfactuals. This first step of their code is sensitive to the starting values. To ensure that the results below are based on the same (local) minimum, I simply rescale the simulated variables instead of reestimating that part of the model. }

\begin{figure}[t!]

	\caption{Marginal Product of Investment on Cognition Age $t+1$}
	
	\label{fig:amn_4}
	
	\newgeometry{textwidth=17.2cm}

	\begin{center}

		\vspace{-2mm}
		
		\begin{subfigure}[h]{0.32\textwidth}
			\centering
			\hspace{-6mm} \caption{AMN scale \hspace{4mm}}
			\vspace{-4mm}
			\hspace{-6mm}\includegraphics[width=0.95\textwidth, trim={4cm, 8cm, 5cm,  7.5cm},clip]{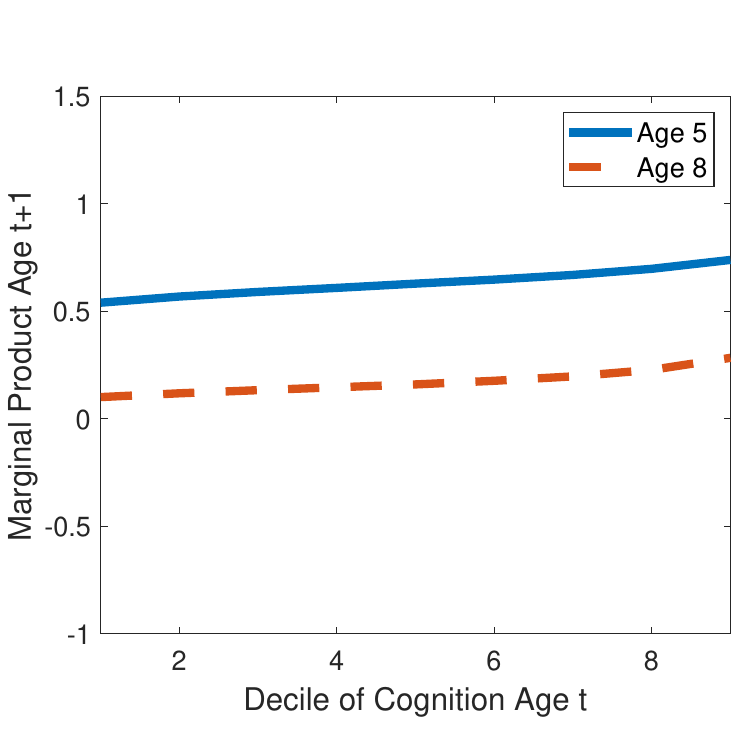} 
		\end{subfigure}
		\begin{subfigure}[h]{0.32\textwidth}
			\centering
			\hspace{-6mm}\caption{Investment in 100 INR 	\hspace{4mm}}
			\vspace{-4mm}
			\hspace{-6mm}\includegraphics[width=0.95\textwidth, trim={4cm, 8cm, 5cm,  7.5cm},clip]{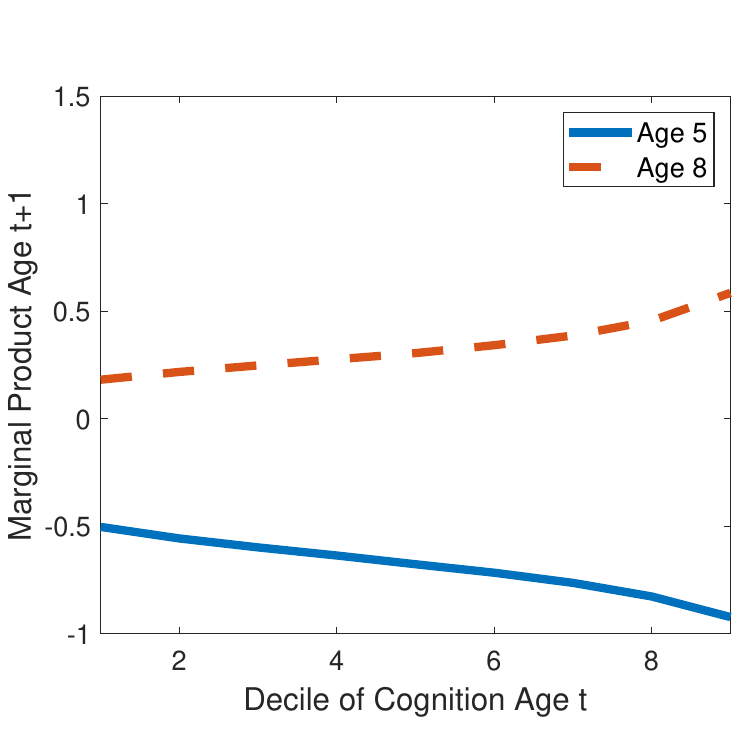}
		\end{subfigure}
		\begin{subfigure}[h]{0.32\textwidth}
			\centering
			\hspace{-6mm}\caption{Scaled test scores \hspace{4mm}}
			\vspace{-4mm}
			\hspace{-6mm}\includegraphics[width= 0.95\textwidth, trim={4cm, 8cm, 5cm,  7.5cm},clip]{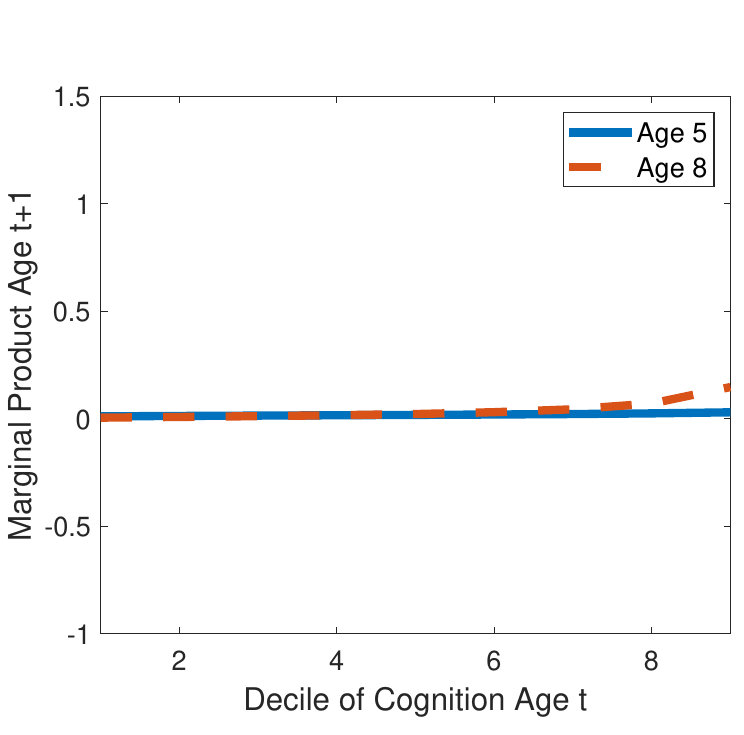}
		\end{subfigure}
		
	\end{center}

	\restoregeometry
	
 	{\footnotesize \begin{singlespace}  Notes: The figure shows marginal products of investment on cognition. The y-axis represents the impact on cognition, in standard deviation units, of increasing investment by one standard deviation. Panel (a) corresponds to the left panel of Figure 4(a) of AMN. Panels (b) and (c) are based on  book expenditures in 100 INR and scaled test scores, respectively. \end{singlespace}}
	\end{figure}

\shortciteN{AMN:19} analyze the marginal product of investment on cognition and write: ``When considering the production of cognitive skills (\ldots) the productivity of investments is much higher at younger ages: investments are more able to affect cognition earlier on.'' This conclusion is based on the left panel of Figure \ref{fig:amn_4} below, which corresponds to Figure 4(a) of \shortciteN{AMN:19}.\footnote{When calculating counterfactuals, they exclude the constant in the investment equation and incorrectly order some columns of the covariate matrix. I correct these issues and thus obtain slightly different estimates.} The middle panel of Figure \ref{fig:amn_4} shows the results based on book expenditures in 100 INR. In this case, the marginal effects are higher at age 8 than at age 5 (and even negative in the latter case).\footnote{I follow \shortciteN{AMN:19} and do not restrict the signs of the coefficients (and many of their specifications result in some negative coefficients as well). } When I rescale test scores instead of book expenditures, the estimated effects are much lower, as can be seen in panel (c).

By Theorem \ref{th:identfunctionsces}, there are different, invariant ways to  interpret the effect of investment on skills. For example, here we can conclude: Fixing all variables at their median values, 
\begin{enumerate}
	\item[(i)] if we increase someones investment from the 40\% quantile to the 60\% quantile at age 5, we increase her skills from the 46\% quantile to the 51\% quantile at age 8, or 
	\item[(ii)]  if we increase investment by 1\% at age 5, we increase skills by 0.004\% at age 8, or
	\item[(iii)] if we increase investment s.t. book expenditures increase by 1 std (or 783 INR) at age 5, we increase skills at age 8 s.t. the standardized PPVT score increases by 0.26 std (or 0.23 units).
\end{enumerate} 

Although the effect described in interpretation (ii) may appear smaller than that in (i), the magnitudes are in fact consistent with each other. To understand their relationship, note that the logs of the latent variables follow mixtures of normal distributions. Here, the estimated variance of log-investment is substantially larger than that of log-skills. As a result, a 10\% increase in investment at the median corresponds to a shift to the 50.18\% quantile. In contrast, a 0.04\% increase in skills at the median corresponds to a shift to the 50.07\% quantile. The relative change in these quantiles ($0.0007/0.0018=0.39$), is comparable to the effect in interpretation (i), where the relative quantile change is $0.05/0.2=0.25$.

\begin{figure}[t!]
	\caption{Estimated effects of income transfers - AMN estimator - skill level}
	\label{fig:amn_5_restricted_level}

	\newgeometry{textwidth=17.2cm}
	
	\begin{center}
		
		\vspace{-6mm}
		
		\begin{subfigure}[h]{0.32\textwidth}
			\centering
			\hspace{-6mm} \caption{Age 5 - AMN scale \, \hspace{5mm}}
			\vspace{-4mm}
			\hspace{-6mm}\includegraphics[width=0.95\textwidth, trim={4cm, 8cm, 5cm,  7.5cm},clip  ]{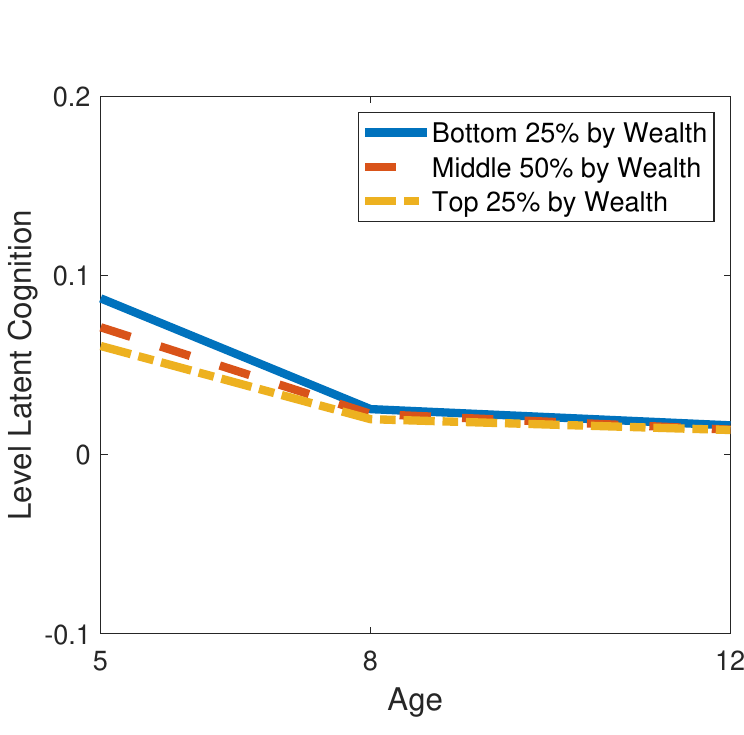}
		\end{subfigure}
		\begin{subfigure}[h]{0.32\textwidth}
			\centering
			\hspace{-6mm} \caption{Age 5 - invest. in 100 INR \,	\hspace{7mm}}
			\vspace{-4mm}
			\hspace{-6mm}\includegraphics[width=0.95\textwidth, trim={4cm, 8cm, 5cm,  7.5cm},clip  ]{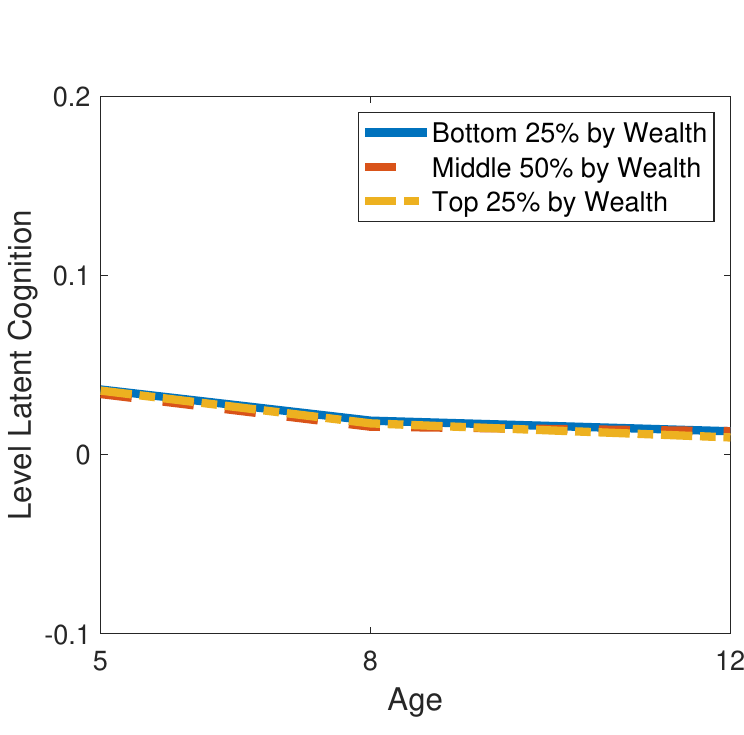}
		\end{subfigure}
		\begin{subfigure}[h]{0.32\textwidth}
			\centering
			\hspace{-6mm}\caption{Age 5 - scaled tests \, \hspace{5mm}}
			\vspace{-4mm}
			\hspace{-6mm}\includegraphics[width=0.95\textwidth, trim={4cm, 8cm, 5cm,   7.5cm},clip  ]{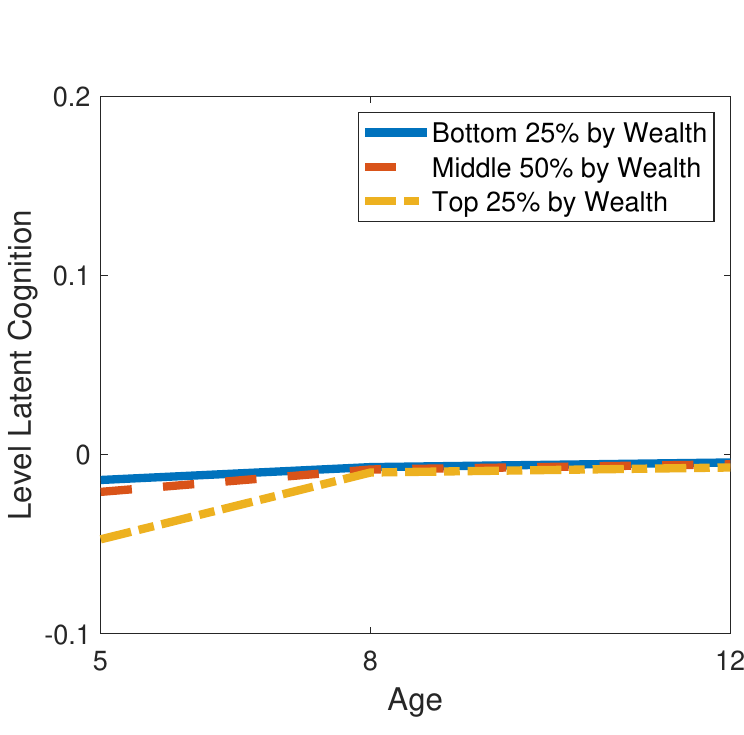}
		\end{subfigure}
		
		\begin{subfigure}[h]{0.32\textwidth}
			\centering
			\hspace{-6mm} \caption{Age 8 - AMN scale \, \hspace{5mm}}
			\vspace{-4mm}
			\hspace{-6mm}\includegraphics[width=0.95\textwidth, trim={4cm, 8cm, 5cm,  7.5cm},clip  ]{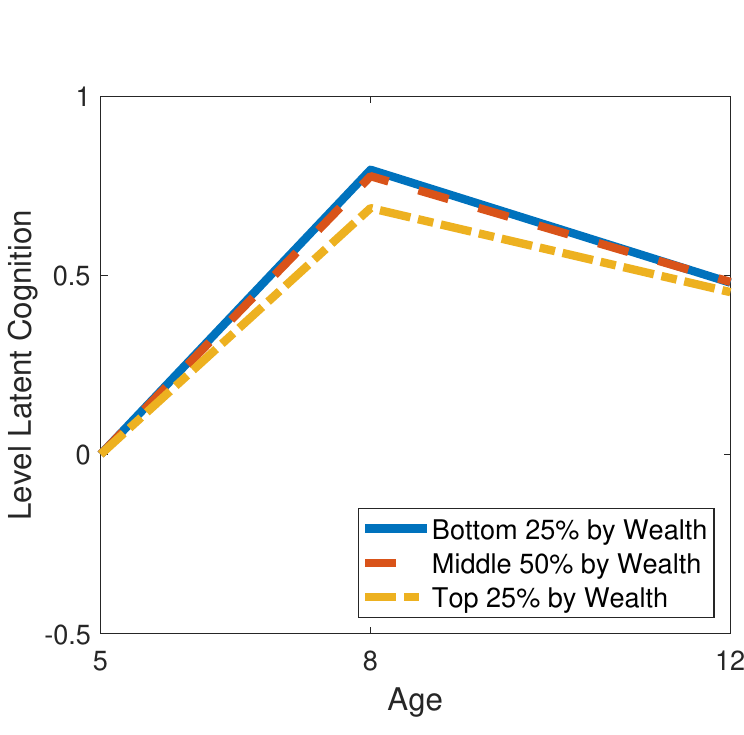}	
		\end{subfigure}
		\begin{subfigure}[h]{0.32\textwidth}
			\centering
			\hspace{-6mm}\caption{Age 8 - invest. in 100 INR \,	\hspace{7mm}}
			\vspace{-4mm}
			\hspace{-6mm}\includegraphics[width=0.95\textwidth, trim={4cm, 8cm, 5cm,   7.5cm},clip  ]{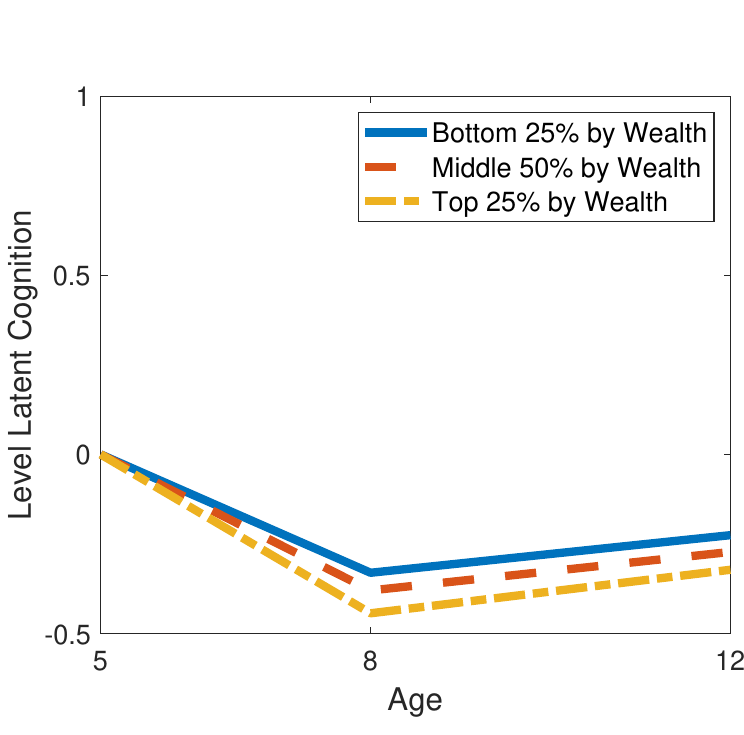}
		\end{subfigure}
		\begin{subfigure}[h]{0.32\textwidth}
			\centering
			\hspace{-6mm}\caption{Age 8 - scaled tests \, \hspace{5mm}}
			\vspace{-4mm}
			\hspace{-6mm}\includegraphics[width=0.95\textwidth, trim={4cm, 8cm, 5cm,  7.5cm},clip  ]{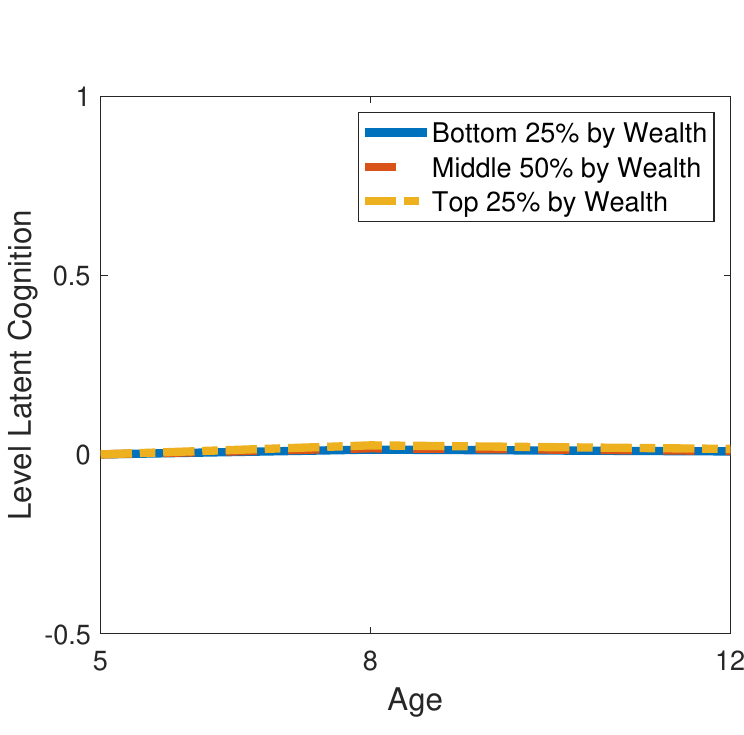}
		\end{subfigure}
		
	\end{center}

	\newgeometry{textwidth=16.6cm}
	{\footnotesize \begin{singlespace}  Notes:  The figure shows the impact on cognition of an income transfer equal to 25\% of mean income in the entire sample. The y-axis represents the impact on cognition, measured as the change in the median level of cognition in standard deviation units, of increasing investment by one standard deviation. In the top three graphs, the transfer is made before age 5. In the lower three graphs, it is made between ages 5 and 8. Panels (a) and (d) correspond to the left two panels of Figure 5 of AMN. Panels (b) and (e) are based on   book expenditures in 100 INR. Panels (c) and (e) are based in scaled test scores. \end{singlespace}}
	\end{figure}

The authors also study the impact on cognition of an income transfer equal to 25\% of mean income in the entire sample. Specifically, they consider the change in the median level of cognition in standard deviation units. The results for the different scales are shown in Figure \ref{fig:amn_5_restricted_level}. The income transfer is made before age 5 in the upper panels and  between ages 5 and 8 in the lower panels.  The authors write: ``In terms of timing, the largest impact is obtained if the transfer takes place when the children are between 5 and 8 for cognition'', which is based on the much larger response in the lower panel at age 12. This conclusion does not hold when using book expenditures in 100 INR, because the effect of investment is then estimated to be negative in panel (e). When scaling test scores instead, the effects are very close to 0 (see panels (c) and (f)).

There are two problems with the counterfactuals reported in Figure \ref{fig:amn_5_restricted_level} using the specification of \shortciteN{AMN:19}. First, identified parameters are set to 1, which leads to misspecification. Second, this counterfactual does not correspond to those that are invariant to the scale and location restrictions. To fix the second problem, one can focus on the change in the median log of skills rather than the level of skills.

\begin{figure}[t!]
	\caption{Estimated effects of income transfers - flexible estimator - log-skills}
	\label{fig:amn_5_flexible_logs}

	\newgeometry{textwidth=17.2cm}
	
	\begin{center}
		
		\vspace{-6mm}
		
		\begin{subfigure}[h]{0.4\textwidth}
			\centering
			\hspace{-6mm} \caption{Age 5 - original scale}
			\vspace{-4mm}
			\hspace{-6mm}\includegraphics[width=0.8\textwidth, trim={4cm, 8cm, 5cm,  7.5cm},clip  ]{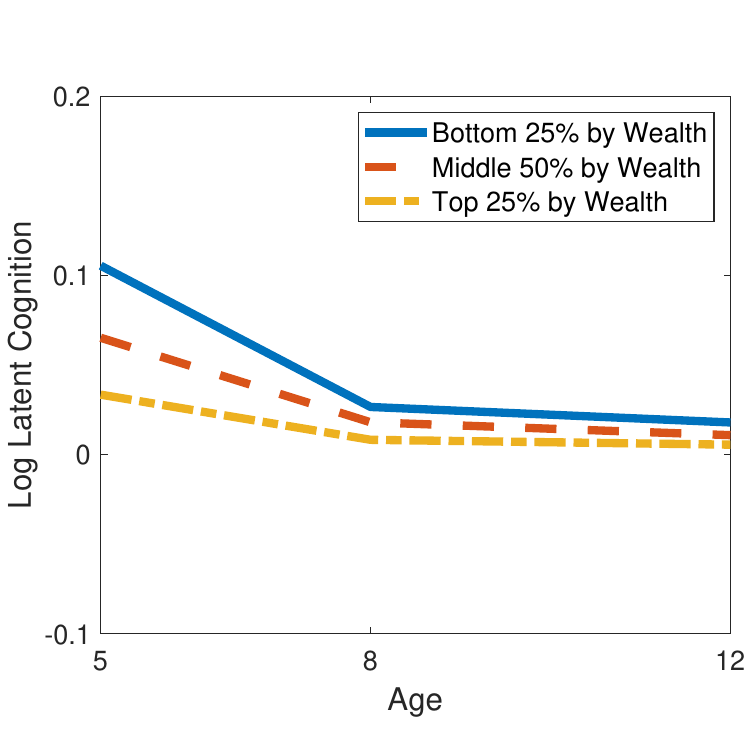}
		\end{subfigure}
		\begin{subfigure}[h]{0.4\textwidth}
			\centering
			\hspace{-6mm}\caption{Age 8 - original scale}
			\vspace{-4mm}
			\hspace{-6mm}\includegraphics[width=0.8\textwidth, trim={4cm, 8cm, 5cm,  7.5cm},clip  ]{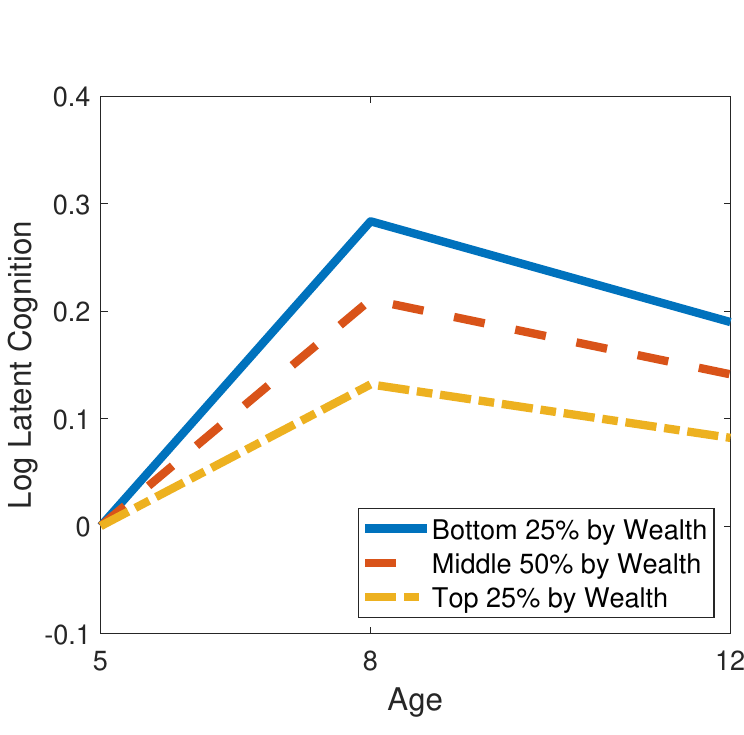}	
		\end{subfigure}
		
		\vspace{4mm}
		\begin{minipage}{0.85\textwidth}
			{\footnotesize \begin{singlespace} Notes:  The figure shows the impact on cognition of an income transfer equal to 25\% of mean income in the entire sample. The y-axis represents the impact on cognition, measured as the change in the median of the log of cognition in standard deviation units, of increasing investment by one standard deviation. In the left graph, the transfer is made before age 5. In the lower graph, it is made  between ages 5 and 8. Estimates are based on the flexible estimator.
			\end{singlespace}} 
		\end{minipage}
		
	\end{center}

\end{figure}

The results for the log-skills using the estimator of \shortciteN{AMN:19} are shown in Figure \ref{fig:amn_5_restricted_log} in Appendix \ref{s:app_figures_application}. Since the estimates are based on setting identified parameters to 1, the results are still sensitive to the scales. Figure \ref{fig:amn_5_flexible_level} in Appendix \ref{s:app_figures_application}  shows the estimates corresponding to the counterfactual of \shortciteN{AMN:19}, but based on the estimator that also estimates the identified scales, which also depend on the the scales of the test scores (but not investment). Finally, Figure \ref{fig:amn_5_flexible_logs} displays the counterfactual in terms of log-skill based on the flexible estimator, which is invariant to the units of measurements (and therefore, the other scales are omitted). Based on Figure \ref{fig:amn_5_flexible_logs}, we can still conclude that investment between ages 5 and 8 is more beneficial. However, compared to panels (d) of Figure \ref{fig:amn_5_restricted_level}, the new estimates imply much more heterogeneity by wealth.  Figure \ref{fig:amn_densities} in Appendix \ref{s:app_figures_application} shows the effects of income transfers on the entire counterfactual distribution of the test scores.

\section{Conclusion} 

This paper demonstrates that in an important class of skill formation models, seemingly innocuous scale and location restrictions may not function as mere normalizations. Instead, they can constrain identified parameters and influence counterfactuals. The specific implications depend on the feature of interest, the production function, the measurement system, and the estimation method. In a new identification analysis, I pool all restrictions of the model and characterize the identified set without imposing additional scale and location restrictions. Notably, many key features remain invariant to these restrictions, are identified under weaker assumptions, yield robust policy implications, and are comparable across different studies.
	
Researchers often impose ``normalizations'' to achieve point identification in models that are otherwise only partially identified. To avoid unintended consequences and potentially misleading conclusions, it is critical to clarify which parameters and features are invariant to these restrictions, ensuring they are genuine normalizations as defined in Definition \ref{d:normalization}. While proving such invariance properties theoretically can be challenging, it is typically possible to demonstrate the robustness of key results to these restrictions. 

One example closely related to the model considered here are skill formation models that explicitly solve household optimization problems to determine optimal investment functions (see e.g. \citeN{DFW:13} or \citeN{CL:20}). These models, which include a production function, are then estimated using the optimality conditions alongside a measurement system for latent variables. While these models have a different structure compared to the ones I study and a detailed analysis is left for future research, it appears that analogous normalization  issues may arise, with consequences that depend on the model specification and the types of counterfactuals examined.

\section{Data availability statement}
	
The data and code underlying this research is available on Zenodo at 

\noindent \url{https://doi.org/10.5281/zenodo.16437164}.

\appendix

\setcounter{figure}{0}
\setcounter{table}{0}
\setcounter{theorem}{0}
\setcounter{assumption}{0}
\renewcommand{\thetable}{A.\arabic{table}}
\renewcommand{\thefigure}{A.\arabic{figure}}
\renewcommand{\thetheorem}{A.\arabic{theorem}}
\renewcommand{\theassumption}{A.\arabic{assumption}}

\allowdisplaybreaks

\begin{center}	
	\textbf{\Large{Online Appendix for \\``Normalizations and misspecification in skill formation models''}}

	\vspace{5mm}

	Joachim Freyberger\symbolfootnote[3]{University of Bonn. Email: freyberger@uni-bonn.de. } 
	
	\vspace{5mm}
	
	August 2025
	
\end{center}

\section{Additional identification results}

\subsection{Parameters of the transformed trans-log model}
\label{s:app_additional_parameters}

The remaining parameters in equation (\ref{eq:prod_fn_rw}) are   
$$\tilde{a}_t =  \lambda_{\theta,t+1,1} a_t + \mu_{\theta,t+1,1} - \frac{\lambda_{\theta,t+1,1}}{\lambda_{\theta,t,1}} \mu_{\theta,t,1} \gamma_{1t} - \frac{\lambda_{\theta,t+1,1}}{\lambda_{I,t,1}} \mu_{I,t,1} \gamma_{2t} + \frac{\lambda_{\theta,t+1,1}}{\lambda_{\theta,t,1} \lambda_{I,t,1}} \mu_{I,t,1}  \mu_{\theta,t,1}\gamma_{3t}  $$
with unobservable $ \tilde{\eta}_{\theta,t} =   \lambda_{\theta,t+1,1}{\eta}_{\theta,t}$, the parameters in the measurement equations (\ref{eq:measurement_eq_rw}) and (\ref{eq:measurement_eq_invest_rw}) are  $ \tilde{\mu}_{\theta,t,1} = 0$, $\tilde{\lambda}_{\theta,t,1} = 1$,  $ \tilde{\mu}_{I,t,1} = 0$, $\tilde{\lambda}_{I,t,1} = 1$, 
$$ \tilde{\lambda}_{\theta,t,m} = \lambda_{\theta,t,m}/\lambda_{\theta,t,1}, \quad \tilde{\mu}_{\theta,t,m} = \mu_{\theta,t,m} - \left(\lambda_{\theta,t,m}/\lambda_{\theta,t,1}\right) \mu_{\theta,t,1}$$
$$ \tilde{\lambda}_{I,t,m} = \lambda_{I,t,m}/\lambda_{I,t,1}, \quad \tilde{\mu}_{I,t,m} = \mu_{I,t,m} - \left(\lambda_{I,t,m}/\lambda_{I,t,1}\right) \mu_{I,t,1},$$
the parameters in equation (\ref{eq:investment_rw}) are
$$ \tilde{\beta}_{0t} = \lambda_{I,t,1} {\beta}_{0t} +  \mu_{I,t,1} - (\lambda_{I,t,1}/\lambda_{\theta,t,1}) \mu_{\theta,t,1} {\beta}_{1t}   , \quad \tilde{\beta}_{1t} = (\lambda_{I,t,1}/\lambda_{\theta,t,1}) {\beta}_{1t},   \quad \tilde{\beta}_{2t} = \lambda_{I,t,1} {\beta}_{2t} $$
with unobservable $ \tilde{\eta}_{I,t} =  \lambda_{I,t,1} {\eta}_{I,t}$, and the parameters in equation (\ref{eq:anchor_eq_rw}) are
$$\tilde{\rho}_0 = \rho_0 -   \frac{\rho_1 \mu_{\theta,T,1}}{\lambda_{\theta,T,1}}, \quad   \tilde{\rho}_1 = \frac{ \rho_1}{\lambda_{\theta,T,1}}.$$

\subsection{Further theorems}

\begin{theorem} 
	\label{th:obseq}	
	Suppose Assumption \ref{a:baseline} holds and consider the model
	\begin{eqnarray*}
		\ln \theta_{t+1} &=& a_t + \gamma_{1t} \ln \theta_{t} +  \gamma_{2t} \ln I_t +  \gamma_{3t} \ln \theta_{t}  \ln I_t + \eta_{\theta,t}   \hspace{18mm} t = 0, \ldots, T-1   \\
		Z_{\theta,t,m} &=& \mu_{\theta,t,m} + \lambda_{\theta,t,m} \ln \theta_{t} + \eps_{\theta,t,m} \hspace{46mm} t = 0, \ldots, T, m = 1, 2   \\
		Z_{I,t,m} &=& \mu_{I,t,m} + \lambda_{I,t,m} \ln I_{t} + \eps_{I,t,m} \hspace{46mm} t = 0, \ldots, T-1, m = 1,2    \\
		\ln I_t &=& \beta_{0t} + \beta_{1t} \ln \theta_{t}  + \beta_{2t} \ln Y_{t}   + \eta_{I,t} \hspace{41mm} t = 0, \ldots, T-1  \\
		Q &=& \rho_{0} + \rho_{1} \ln \theta_{T} + \eta_Q  
	\end{eqnarray*}
	Then there always exist sets of observationally equivalent parameters which are consistent with the data and satisfy Assumptions \ref{a:baseline}, \ref{a:normalization}, either \ref{a:ageinvariant_technology_skills}(a) or \ref{a:ageinvariant_technology_skills}(b)  and either \ref{a:ageinvariant_technology_investment}(a) or \ref{a:ageinvariant_technology_investment}(b). 
	
\end{theorem}

\begin{theorem} 
	\label{th:obseq_ces}	
	Suppose Assumption \ref{a:baseline} holds and consider the model
	\begin{eqnarray*}
		\theta_{t+1} &=& \left(\gamma_{1t} \theta_{t}^{\sigma_t} +   \gamma_{2t} I_t^{\sigma_t}\right)^{\frac{\psi_t}{\sigma_t}}\exp(\eta_{\theta,t})  \hspace{44mm} t = 0, \ldots, T-1  \\
		Z_{\theta,t,m} &=& \mu_{\theta,t,m} + \lambda_{\theta,t,m} \ln \theta_{t} + \eps_{\theta,t,m} \hspace{46mm} t = 0, \ldots, T, m = 1,  2  \\
		Z_{I,t,m} &=& \mu_{I,t,m} + \lambda_{I,t,m} \ln I_{t} + \eps_{I,t,m} \hspace{46mm} t = 0, \ldots, T-1   \\
		\ln I_t &=& \beta_{0t} + \beta_{1t} \ln \theta_{t}  + \beta_{2t} \ln Y_{t}   + \eta_{I,t} \hspace{41mm} t = 0, \ldots, T-1  \\
		Q &=& \rho_{0} + \rho_{1} \ln \theta_{T} + \eta_Q 
	\end{eqnarray*}
	where $\sigma_t \neq 0$ and $\gamma_{1t},\gamma_{2t} >0$ for all $t$. Then there always exist sets of observationally equivalent parameters which are consistent with the data and satisfy Assumptions \ref{a:baseline}, \ref{a:normalization_ces}, either \ref{a:ageinvariant_technology_skills_ces}(a) or \ref{a:ageinvariant_technology_skills_ces}(b)  and either \ref{a:ageinvariant_technology_investment_ces}(a) or \ref{a:ageinvariant_technology_investment_ces}(b)  and either \ref{a:add_restrictions_ces}(a) or \ref{a:add_restrictions_ces}(b).  
	
\end{theorem}

\section{Anchoring and standardizing}

\subsection{Anchoring in the CES case}
\label{s:anchor_ces}

As in the trans-log case, we can write the measurement system and the adult outcome as 
\begin{eqnarray*}
	Z_{\theta,t,m} &=& \tilde{\mu}_{\theta,t,m} - \frac{ \tilde{\rho}_0 \tilde{\lambda}_{\theta,t,m}}{\tilde{\rho}_1} + \frac{\tilde{\lambda}_{\theta,t,m}}{\tilde{\rho}_1}\ln \tilde{\vartheta}_{t} + \eps_t \\
	Q &=& \ln  \tilde{\vartheta}_{T}  + \eta_Q
\end{eqnarray*}
where the joint distribution of $\{\ln  \tilde{\vartheta}_{t}\}^T_{t=1}$ and $\tilde{\rho}_0$ and  $\tilde{\rho}_1$ are point identified. In general, the implications of anchoring are similar to those in the trans-log case and it replaces Assumption \ref{a:normalization_ces}. However, recall that the relative scale of investment and skills is identified and we therefore cannot anchor both variables to a measure or adult outcome. Moreover, if investment was observed and either Assumption \ref{a:add_restrictions_ces}(a) or Assumption \ref{a:add_restrictions_ces}(b) holds, then $\{\lambda_{\theta,t,1}\}^T_{t=0}$ is point identified and the model can only be consistent with one particular (identified) scale of the skills. Thus, we can only have a CES production technology for $\{\tilde{\vartheta}_t\}^T_{t=0}$ if $\rho_1 = 1$ and the model cannot be consistent with different anchors or different units of measurement.

\subsection{Standardizing the CES inputs}

There is also a large macroeconomic literature on normalized CES production functions; see for example \citeN{KG:00}, \citeN{KMW:12}, \citeN{Temple:12} and references therein. They note that one can always write
$$\left(\gamma_{1t} \theta_{t}^{\sigma_t} +   \gamma_{2t} I_t^{\sigma_t}\right)^{\frac{1}{\sigma_t}}  = \bar{A}_t\left(\bar{\gamma}_{t} \left(\frac{\theta_t}{\bar{\theta}_{t}}\right)^{\sigma_t} +  (1- \bar{\gamma}_{t}) \left(\frac{I_t}{\bar{I}_t}\right)^{\sigma_t}\right)^{\frac{1}{\sigma_t}}  $$ 
where $\bar{\theta}_{t}$ and $\bar{I}_{t}$ are fixed constants to be chosen by the researcher to standardize the inputs, 
$$\bar{\gamma}_{t} = \frac{\gamma_{1t} \bar{\theta}_t^{\sigma_t}}{\gamma_{1t} \bar{\theta}_t^{\sigma_t} + \gamma_{2t} \bar{I}_t^{\sigma_t} }  \quad \text{ and } \quad \bar{A}_t = \left( \gamma_{1t} \bar{\theta}_t^{\sigma_t} + \gamma_{2t} \bar{I}_t^{\sigma_t} \right)^{\frac{1}{\sigma_t}} .$$
Estimating the production function with these two different specifications yields observationally equivalent models with the same elasticity of substitution. These papers discuss that with known units of measurements of the inputs, such standardizations can help interpret the parameters and calibrate the model. 

One potential choice of standardizations could be $\bar{\theta}_{t} = E[\theta_t]$ and $\bar{I}_{t} = E[I_t]$  (see e.g. \citeN{Embrey:19}). However, since the inputs are not observed, $E[\theta_t]$ and $E[I_t]$ are not identified. Instead, we can write the production function in terms of $\tilde{ \theta}_{t+1}$ as
\begin{align*}
	\tilde{\theta}_{t+1}  &=  e^{\mu_{\theta,t+1,1}}\bar{A}_t \left( \bar{\gamma}_t  \left(\frac{\tilde{\theta}_t}{E [\tilde{\theta}_{t}^{1/\lambda_{\theta,t,1}} ]^{\lambda_{\theta,t,1} }  }\right)^{\frac{\sigma_{t} }{\lambda_{\theta,t,1}}} +  (1-\bar{\gamma}_t)\left(\frac{\tilde{I}_t}{ E [\tilde{I}_{t}^{1/\lambda_{I,t,1}} ]^{\lambda_{I,t,1} }  }\right)^{\frac{\sigma_{t} }{\lambda_{I,t,1}}}  \right)^{ \frac{\lambda_{\theta,t+1,1}\psi_{t} }{\sigma_{t}} } \hspace{-2mm} e^{\tilde{\eta}_{\theta,t}}
\end{align*}
To achieve the desired standardization, we would have to standardize the inputs in equation (\ref{eq:prod_fn_ces_norm}) by $E[\tilde{\theta}_{t}^{1/\lambda_{\theta,t,1}}]^{\lambda_{\theta,t,1} } $ and $E[\tilde{I}_{t}^{1/\lambda_{I,t,1}}]^{\lambda_{I,t,1} } $, respectively, which are not identified. We could instead standardize by $E[\tilde{\theta}_{t} ]$ and $E[\tilde{I}_{t} ]$, but is not clear if such a standardization would yield a useful interpretation of the parameters.    Finally, notice that the most important issue with the current specification of the CES production function is misspecification due to setting the scales, which is not mitigated by standardizing inputs.

\subsection{Standardizing the measures}
\label{s:stand_measures}
It can be convenient to standardize the measures before estimating the model due to potential numerical issues resulting from vastly different scales. There are two ways to so. First, we one could subtract the mean of a measure and divide by its standard deviation in each period. That is, we transform $Z_{\theta,t,m} $ to 
$$\tilde{Z}_{\theta,t,m} \equiv \frac{Z_{\theta,t,m}  - E[Z_{\theta,t,m} ]}{std(Z_{\theta,t,m} )} = - \frac{\lambda_{\theta,t,m}  E[ \ln \theta_{t}]}{std(Z_{\theta,t,m} )} + \frac{\lambda_{\theta,t,m}}{std(Z_{\theta,t,m} )}  \ln \theta_{t} + \frac{\eps_{\theta,t,m}}{std(Z_{\theta,t,m} )} $$
with $std(Z_{\theta,t,m} ) = \sqrt{\lambda_{\theta,t,m}^2 var(\ln \theta_{t}) + var(\eps_{\theta,t,m})}$. Such transformations do not necessarily imply that it is more reasonable to have age-invariant measures or that the new loadings, $\lambda_{\theta,t ,m}/\sqrt{\lambda_{\theta,t,m}^2 var(\ln \theta_{t}) + var(\eps_{\theta,t,m})}$, are equal to $1$. For example, if the mean or the variance of skills increases over time, then the mean and the variance of the measures are not time-invariant, even with age-invariant raw measures. Moreover, the standard deviations of the measures also depend on the variances of the measurement errors, and if these variances change over time, then the new loadings also change over time, even if $\lambda_{\theta,t,m}$ and $var(\ln \theta_{t})$ do not. A second possibility is to  subtract the mean of the pooled measure over all time periods in each period and to divide by the standard deviation of the pooled measures. This transformation is a simple change in the units of measurements and certain features are not invariant to it, as discussed in Section \ref{s:skillform}. However, all features in Theorems \ref{th:identfunctions} and \ref{th:identfunctionsces} are invariant to either way of standardization of the measures.

\section{Estimation}
\label{s:estimation}

\subsection{Trans-log }
\label{s:estimation_tl}

To estimate the features in Theorem \ref{th:identfunctions}, one can use any estimator that relies on either Assumptions \ref{a:baseline}, \ref{a:normalization}, \ref{a:ageinvariant_technology_skills}(a) and \ref{a:ageinvariant_technology_investment}(a) or another set of identifying assumptions of Corollary \ref{c:pointident}. For example, the estimator of \citeN{AW:22} (if investment is exogenous) is computationally attractive.	All of these sets of assumptions yield observationally equivalent models with potentially very different primitive parameters, but all features described in Theorem \ref{th:identfunctions} will be identical. It is typically most convenient to set $\mu_{\theta,t,1} = \mu_{I,t,1} = 0$ and  $\lambda_{\theta,t,1} = \lambda_{I,t,1} = 1$, estimate the implied primitive parameters, and then calculate the features in Theorem \ref{th:identfunctions}.  \citeN{AW:22} use Assumption \ref{a:ageinvariant_technology_investment}(b) instead of \ref{a:ageinvariant_technology_investment}(a) because they are concerned that their investment measures are not age-invariant. However, the restrictions in Assumption \ref{a:ageinvariant_technology_investment}(b) are hard to interpret economically (in terms of constant returns to scale) as the parameters depend on the units of measurement of the data (see Example \ref{e:rescale}). For the features described in Theorem \ref{th:identfunctions}, these assumptions are not required and the estimates do not depend on which set of assumptions is employed.   

An alternative could be to use Assumption \ref{a:baseline} only and an estimator that allows for partial identification. However, these methods can be computationally demanding with many parameters and seem to offer little benefits in this setting, because under Assumption \ref{a:baseline} only, the identified sets are typically unbounded. Moreover, the features in Theorem \ref{th:identfunctions} are point identified and can be recovered from an estimator in a point identified model.

\subsection{CES}
\label{s:estimation_ces}

As in the trans-log case, we can estimate the model using any combinations of assumptions  in Corollary \ref{c:ces_ident} that yield point identification. No matter which combination is used, estimates of the features in Theorem \ref{th:identfunctionsces}, will be identical. In this section I outline one particular way to do so based on equations  (\ref{eq:prod_fn_ces_norm})--(\ref{eq:anchor_eq_rw_ces}) together with Assumptions \ref{a:baseline}, \ref{a:normalization_ces}, \ref{a:ageinvariant_technology_skills_ces}(a), and \ref{a:ageinvariant_technology_investment_ces}(a), and \ref{a:add_restrictions_ces}(b) and an adaptation of the estimation approach of \shortciteN{AMN:19}. 

In the first step, we can use the restriction $\tilde{\lambda}_{\theta,t,1} = \tilde{\lambda}_{I,t,1} =1$ for all $t$ and $\tilde{\mu}_{\theta,t,m} = \tilde{\mu}_{I,t,m} = 0$ for all $t$ and $m$ to estimate the joint distribution of $(\{\tilde{\theta}_t\}^{T}_{t=0},\{\tilde{I}_t,Y_t\}^{T}_{t=0},Q)$ as well as $\tilde{\lambda}_{\theta,t,m}$, $\tilde{\lambda}_{I,t,m}$ for all $t$ and $m>1$, $\tilde{\rho}_0$, and $\tilde{\rho}_1$. For example, \shortciteN{AMN:19} assume that the measures, log-skills,  log-investment, and log-income, have a normal mixture distribution. In the second step, we can take draws from the estimated joint distribution and estimate the remaining parameters. To do so, let  $\{ \{\tilde{\theta}_{t,j}\}^{T}_{t=0},\{\tilde{I}_{t,j}, Y_{t,j}\}^{T}_{t=0},Q_j\}^J_{j=1}$ be these draws and let
$$(\hat{\tilde{\beta}}_{0t},\hat{\tilde{\beta}}_{1t},\hat{\tilde{\beta}}_{2t}) = \argmin_{ \{\tilde{\beta}_{0t}, \tilde{\beta}_{1t}, \tilde{\beta}_{2t}\}}  \sum^J_{j=1}\left( \ln \tilde{I}_{t,j} - \tilde{\beta}_{0t} - \tilde{\beta}_{1t} \ln \tilde{\theta}_{t,j}  - \tilde{\beta}_{2t} \ln Y_{t,j} \right)^2  $$
and $\hat{\tilde{\eta}}_{I,t,j} =\ln \tilde{I}_{t,j} - \hat{\tilde{\beta}}_{0t} - \hat{\tilde{\beta}}_{1t} \ln \tilde{\theta}_{t,j}  - \hat{\tilde{\beta}}_{2t} \ln Y_{t,j}$. Next  set $\hat{\lambda}_{\theta,0,1} = {\lambda}_{\theta,0,1} = 1$ and let
\begin{align*}
	&( \{ \hat{\lambda}_{\theta,t,1}\}^{T-1}_{t=1}, \{\hat{\lambda}_{I,t,1},\hat{\sigma}_t,\hat{\gamma}_{1t},\hat{\gamma}_{2t},\hat{\tilde{\kappa}}_{t}\}^{T-1}_{t=0} ) \\
	& = \argmin_{  \{ \lambda_{\theta,t,1}\}^{T-1}_{t=1}, \{\lambda_{I,t,1},\sigma_t,{\gamma}_{1t},{\gamma}_{2t},\tilde{\kappa}_{t}\}^{T-1}_{t=0} }   \sum^T_{t=1} \sum^J_{j=1}  \left(\ln \tilde{\theta}_{t+1,j} -   { \frac{\lambda_{\theta,t+1,1}}{\sigma_{t}} } \ln \left(  {\gamma}_{1t}    \tilde{\theta}_{t,j}^{\frac{\sigma_{t} }{\lambda_{\theta,t,1}}} +  {\gamma}_{2t}  \tilde{I}_{t,j}^{\frac{\sigma_{t} }{\lambda_{I,t,1}}}  \right) - \tilde{\kappa}_{t}\hat{\tilde{\eta}}_{I,t,j}\right)^2 
\end{align*}
Using these estimates, we can then recover $\hat{\lambda}_{\theta,t,m} =  \hat{\tilde{\lambda}}_{\theta,t,m} \hat{\lambda}_{\theta,t,1} $, $\hat{\lambda}_{I,t,m} =  \hat{\tilde{\lambda}}_{I,t,m} \hat{\lambda}_{I,t,1} $, $\hat{\beta}_{0t} = \frac{\tilde{\beta}_{0t}}{\hat{\lambda}_{I,t,1}} $, $\hat{\rho}_0 = \hat{\tilde{\rho}}_0$, $\hat{\rho}_1 = \hat{\tilde{\rho}}_1 \hat{\lambda}_{\theta,t,1}$, $\hat{\beta}_{1t} = \tilde{\beta}_{1t} \frac{\hat{\lambda}_{\theta,t,1}}{\hat{\lambda}_{I,t,1}}$, and $\hat{\beta}_{2t} = \frac{\tilde{\beta}_{2t}}{\hat{\lambda}_{I,t,1}} $, as well as the estimated distributions of skills and investment using the relationship $\ln\theta_t = \frac{\ln\tilde{\theta}_t}{{\lambda}_{\theta,t,1}}$ and  $\ln I_t = \frac{\ln\tilde{I}_t}{{\lambda}_{I,t,1}}$. The estimation procedure also easily allows imposing additional assumptions, such as age-invariance of the first skill measure in which case $ {\lambda}_{\theta,1,1} = {\lambda}_{\theta,t,1}$ for all $t$.

As an alternative interpretation of the estimator, consider a more general CES type production function, namely
$\theta_{t+1} = \left(\gamma_{1t} \theta_{t}^{\sigma_{1t}} +   \gamma_{2t} I_t^{\sigma_{2t}}\right)^{\frac{1}{\sigma_{3t}}}\exp(\eta_{\theta,t})$.
This model is observationally equivalent to the one with a more restricted production function because the scales of log-skills are also free parameters. The features in parts 1--4 of Theorem \ref{th:identfunctionsces} are thus identified under Assumption \ref{a:baseline} only.   However, now the production function is sufficiently flexible and setting $ {\lambda}_{\theta,t,1} =  {\lambda}_{I,t,1} =1$ for all $t$ and $ {\mu}_{\theta,t,m} =  {\mu}_{I,t,m} = 0$ simply selects an element of the identified set and one can estimate the invariant features using these restrictions.

\section{Nonparametric identification}
\label{s:genident}

I now extend these results to a general nonparametric model where 
\begin{eqnarray}
	\theta_{t+1} &=& f_t(\theta_{t},I_{t},\eta_{I,t},\varsigma_{\theta,t} )  \hspace{32mm} t = 0, \ldots, T-1 \label{eq:prod_fn_np} \\
	Z_{\theta,t,m} &=& g_{\theta,t,m}(\theta_{t},\eps_{\theta,t,m}) \hspace{35mm} t = 0, \ldots, T, \, m = 1,2,3 \label{eq:measurement_eq_np} \\
	Z_{I,t,m} &=& g_{I,t,m}(I_{t},\eps_{I,t,m}) \hspace{35mm} t = 0, \ldots, T-1, T, \,m = 1,2,3 \label{eq:measurement_eq_invest_np} \\ 
	I_t &=& h_t(\theta_{t},Y_{t},\eta_{I,t})  \hspace{38mm} t = 0, \ldots, T-1  \\
	Q &=& r(\theta_{T},\eta_Q) \label{eq:anchor_eq_np}
\end{eqnarray}

Similar to \shortciteN{CHS:10} and \citeN{AW:22}, in the nonparametric model, we need three measures in each period. I adapt Assumption \ref{a:baseline}  as follows.

\begin{assumption}\label{a:general}  \qquad 
	
	\begin{enumerate}[(a)]

		\item $\{\{\eps_{\theta,t,m}\}_{t=0,\ldots,T, m= 1,2}, \{\eps_{I,t,m}\}_{t=0,\ldots,T-1, m= 1,2} , \eta_Q\}$ are jointly independent and independent of $\{\{\theta_{t}\}^{T}_{t=0}, \{I_t\}^{T-1}_{t=0}\}$ conditional on $\{Y_t\}^{T-1}_{t=0}$.   
		
		\item All random variables are continuously distributed with strictly increasing cumulative distribution functions and have bounded first and second moments.

		\item The joint density of $\{\{Z_{\theta,t,m}\}_{t=0,\ldots,T, m= 1,2,3}, \{Z_{I,t,m}\}_{t=0,\ldots,T-1, m= 1,2,3}, Q, \{\theta_{t}\}^{T}_{t=0}, \{I_{t}\}^{T}_{t=0}\}$ is bounded conditional on $\{Y_t\}^{T-1}_{t=0}$ and so are all their marginal and conditional densities. $\{\{Z_{\theta,t,1}\}^T_{t=0},\{Z_{I,t,1}\}^{T-1}_{t=0}\}$ is bounded complete for $\{\{Z_{\theta,t,2}\}^T_{t=0},\{Z_{I,t,2}\}^{T-1}_{t=0}\}$ and $\{ \{\theta_{t}\}^{T}_{t=0}, \{I_{t}\}^{T}_{t=0} \} $ is bounded complete for $\{\{Z_{\theta,t,1}\}^T_{t=0},\{Z_{I,t,1}\}^{T-1}_{t=0}\}$.
		
		\item  $g_{\theta,t,m}$, $g_{I,t,m}$, and $r$ are strictly increasing in both arguments for all $m$ and $t$. $f_t$ and $h_t$ are strictly increasing in the last argument for all $t$.
		
		\item  $\theta_{t}$ and  $I_{t}$ have strictly positive support for all $t$.

		\item $\eta_{I,t}$ is independent of $(\theta_{t},Y_{t})$ for all $t$ and	$\varsigma_{\theta,t}$ is independent of $(I_\theta,\theta_{t},Y_{t})$ for all $t$.

	\end{enumerate}
	
\end{assumption}

These assumptions are similar to those of \shortciteN{CHS:10}, where now $\eta_{I,t}$ enters the production function directly. As a special case, suppose $ f_t(\theta_{t},I_{t},\eta_{I,t},\varsigma_{\theta,t} ) = f_t(\theta_{t},I_{t},\kappa_{t} \eta_{I,t}+\varsigma_{\theta,t})$. Then with $\eta_{\theta,t} = \kappa_{t} \eta_{I,t}+\varsigma_{\theta,t}$, we have $E[\eta_{\theta,t} \mid \theta_{t},\eta_{I,t} , Y_t] = \kappa_t \eta_{I,t}$, which is part (h) of Assumption \ref{a:baseline}. Here, the unobservables are allowed to enter much more flexibly. Parts (a) -- (c) are analogous to assumptions made in \shortciteN{CHS:10} and I build on their results to identify the joint distribution of 
$$\{\{Z_{\theta,t,m}\}_{t=0,\ldots,T, m= 1,2}, \{Z_{I,t,m}\}_{t=0,\ldots,T-1, m= 1,2}, Q,   \{\tilde{g}_{\theta,t}(\theta_t)\}^{T}_{t=0}, \{\tilde{g}_{I,t}(I_t)\}^{T-1}_{t=0}  \}$$
up to unknown and strictly increasing functions $\tilde{g}_{\theta,t}$ and $\tilde{g}_{I,t}$. This result is similar to the identification result in Lemma \ref{l:identjoint}, which shows identification up to linear transformations. \shortciteN{CHS:10} make an additional assumption, which ensures that the functions $\tilde{g}_{\theta,t}$ and $\tilde{g}_{I,t}$ can be pinned down (i.e. condition (v) of their Theorem 2). This assumption is similar to Assumption \ref{a:normalization} and it is not a normalization with respect to $f_t$. I do not use this assumption and instead focus on  features that are invariant to these monotone transformations. 

\begin{theorem}
	\label{th:indentgen}
	Suppose Assumption \ref{a:general} holds.    
	
	\begin{enumerate}
		
		\item Let $\{\alpha_j\}^5_{j=1} \in (0,1)$ be such that $( Q_{\alpha_1}(\theta_{t})  ,  Q_{\alpha_3}( I_t)   , Q_{\alpha_2}(\eta_{I,t}),    Q_{\alpha_4}(\varsigma_{\theta,t}), Q_{\alpha_5}(\epsilon_{\theta, t+1,m}) )$ is on the support of $(  \theta_{t},  I_t    , \eta_{I,t} ,   \varsigma_{\theta,t}, \epsilon_{\theta, t+1,m}  )$. Then $F_{ \theta_{t+1} }(   f_t(   Q_{\alpha_1}(\theta_t)    ,  Q_{\alpha_2}(I_t), Q_{\alpha_3}(\eta_{I,t}), Q_{\alpha_4}(\varsigma_{\theta,t}) ) )$ and $g_{ \theta, {t+1},m }(   f_t(   Q_{\alpha_1}(\theta_t)    ,  Q_{\alpha_2}(I_t), Q_{\alpha_3}(\eta_{I,t}), Q_{\alpha_4}(\varsigma_{\theta,t}) ), Q_{\alpha_5}(\epsilon_{\theta, t+1,m})   )$  are point identified.

		\item Let $  I_t(y) = h_t(Q_{\alpha_1}(\theta_t),y,Q_{\alpha_2}(\eta_{I,t}))    $
		Then
		$F_{ \theta_{t+1} }(   f_t(   Q_{\alpha_1}(\theta_t)    ,  I_t(y), Q_{\alpha_2}(\eta_{I,t}), Q_{\alpha_3}(\varsigma_{\theta,t}) ) )$
		and   $g_{ \theta, {t+1},m }(   f_t(   Q_{\alpha_1}(\theta_t)    ,I_t(y), Q_{\alpha_2}(\eta_{I,t}), Q_{\alpha_3}(\varsigma_{\theta,t}) ), Q_{\alpha_4}(\epsilon_{\theta, t+1,m})   )$  are point identified. 
		
		\item $  P\left(Q \leq q   \mid \theta_s= Q_{\alpha_1}(\theta_s), \{I_t = Q_{\alpha_{2t}}(I_t) \}^{T-1}_{t=0}, \{\eta_{I,t} = Q_{\alpha_{3t}}(\eta_{I,t})   \}^{T-1}_{t=s}, \{\varsigma_{\theta,t} = Q_{\alpha_{4t}}(\varsigma_{\theta,t})   \}^{T-1}_{t=s}  \right)$ is point identified for all $\alpha_1,\{\alpha_{2t},\alpha_{3t},\alpha_{4t} \}^{T-1}_{t=s} \in (0,1)$ such that the quantiles are on the joint support of the random variables.

		\item $  P\left(Q \leq q   \mid \theta_s= Q_{\alpha}(\theta_s), \{Y_t = y_t \}^{T-1}_{t=s}     \right)$ is point identified for all $\alpha \in (0,1)$.

	\end{enumerate}

\end{theorem}

Just as before, we can identify how investment/income and shocks affect the relative standing in the skill distribution.  We can also identify the effect of a sequence of investment/income on adult outcomes, given a quantile of the initial skills. Notice that these parameters are point identified without any normalizations and they neither require scale and location restrictions on the production function nor age-invariant measures. These features are now not only invariant to changes in the units of measurement, but to any monotone transformations of the measures. Notice that we can only identify a nonlinear transformation of the skills. Therefore, the sequence of investment that maximizes $E(\theta_T \mid    \theta_s= Q_{\alpha}(\theta_s), \{Y_t = y_t \}^{T-1}_{t=s}  )$ is not point  identified and choosing a particular transformation can yield erroneous conclusions regarding the role of investment. It is also important to mention that in nonadditive models, support conditions play an important role because we cannot extrapolate using the functional form. The results also implicitly contain an instrument relevance condition because if $h_t$ is constant in $Y_t$, then   $Q_{\alpha_2}(I_t) $ is completely determined by $ (Q_{\alpha_1}(\theta_t),Q_{\alpha_3}(\eta_{I,t}) )$.

Another advantage of stating results without any seemingly innocuous normalizations is that one can easily impose more structure on the model without having to check and potentially adjust the normalization. For example, to reduce the dimensionality of the model, we might want to simplify the measurement system to
\begin{eqnarray*}
	Z_{\theta,t,m} &=& g_{\theta,t,m}(\theta_{t} + \eps_{\theta,t,m}) \hspace{26mm} t = 0, \ldots, T, m = 1,2,3   \\
	Z_{I,t,m} &=& g_{I,t,m}(I_{t} + \eps_{I,t,m}) \hspace{26mm} t = 0, \ldots, T-1  , m = 1,2,3  
\end{eqnarray*}
A normalization in the more general model might then not be a normalization in the more restrictive model. Also here, the results from Theorem \ref{th:indentgen} apply and we can identify features that are invariant to monotone transformations of the measures.

\clearpage

\section{Additional figures application}
\label{s:app_figures_application}

\begin{figure}[h!]
	\caption{Estimated effects of income transfers - AMN estimator - log-skills}
	\label{fig:amn_5_restricted_log}

	\newgeometry{textwidth=17.2cm}
	
	\begin{center}
		
		\vspace{-6mm}
		
		\begin{subfigure}[h]{0.32\textwidth}
			\centering
			\hspace{-6mm} \caption{Age 5 - original scale}
			\vspace{-4mm}
			\hspace{-6mm}\includegraphics[width=0.95\textwidth, trim={4cm, 8cm, 5cm,  7.5cm},clip  ]{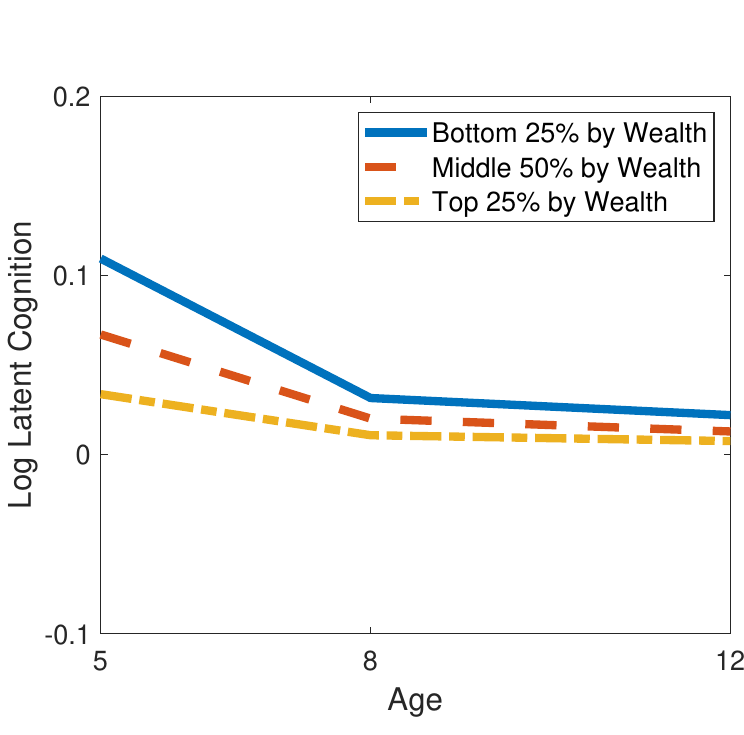}
		\end{subfigure}
		\begin{subfigure}[h]{0.32\textwidth}
			\centering
			\hspace{-6mm} \caption{Age 5 - scaled investment}
			\vspace{-4mm}
			\hspace{-6mm}\includegraphics[width=0.95\textwidth, trim={4cm, 8cm, 5cm, 7.5cm},clip  ]{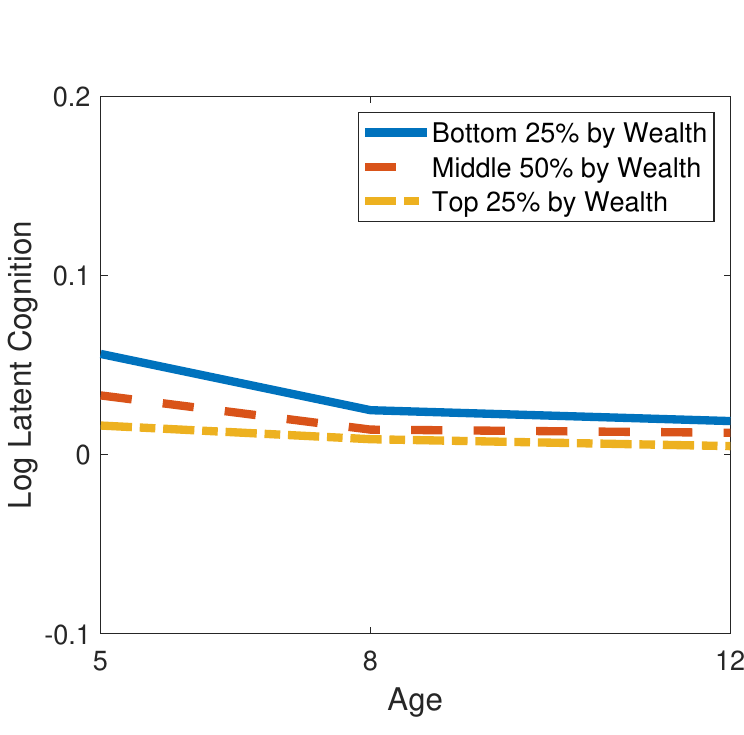}
		\end{subfigure}
		\begin{subfigure}[h]{0.32\textwidth}
			\centering
			\hspace{-6mm}\caption{Age 5 - scaled tests}
			\vspace{-4mm}
			\hspace{-6mm}\includegraphics[width=0.95\textwidth, trim={4cm, 8cm, 5cm,  7.5cm},clip  ]{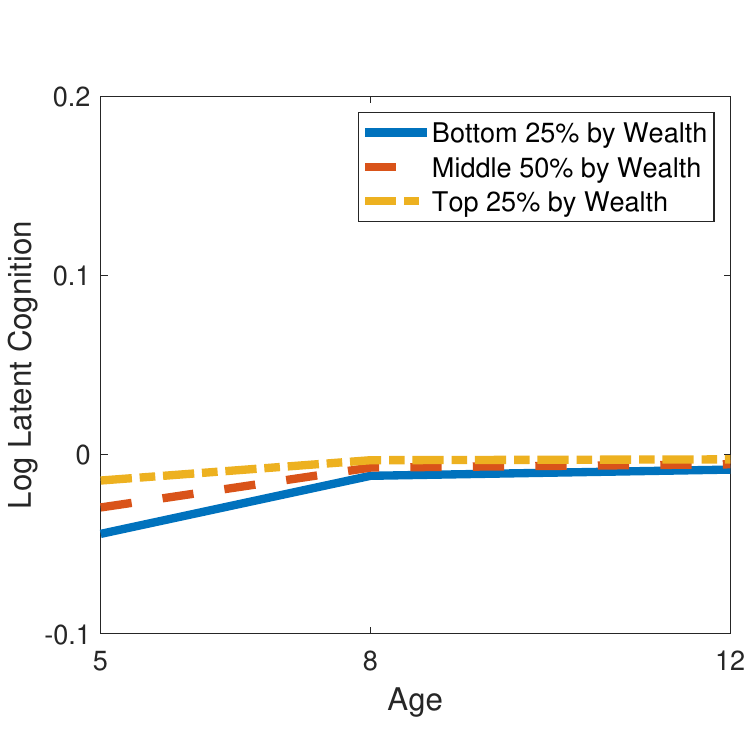}
		\end{subfigure}
		
		\begin{subfigure}[h]{0.32\textwidth}
			\centering
			\hspace{-6mm}\caption{Age 8 - original scale}
			\vspace{-4mm}
			\hspace{-6mm}\includegraphics[width=0.95\textwidth, trim={4cm, 8cm, 5cm,  7.5cm},clip  ]{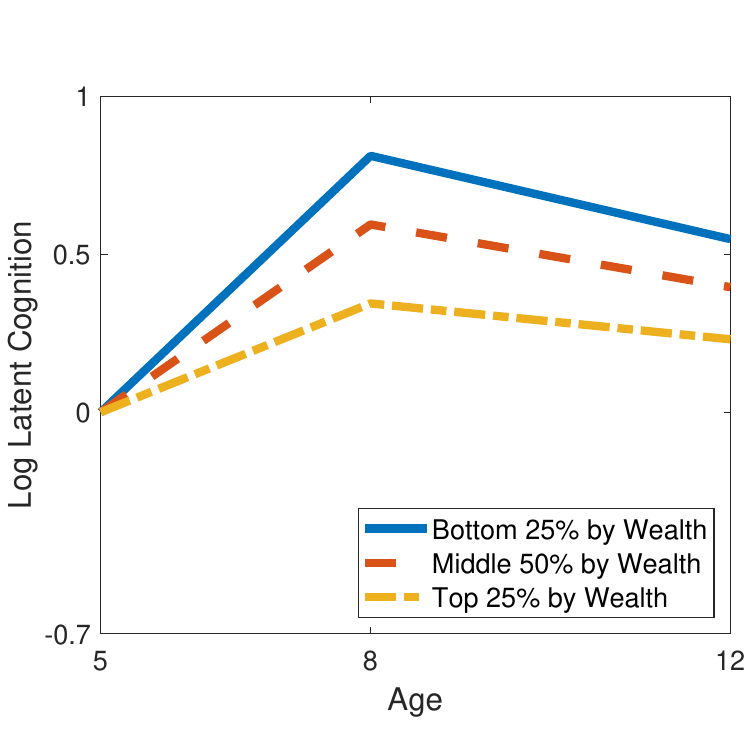}	
		\end{subfigure}
		\begin{subfigure}[h]{0.32\textwidth}
			\centering
			\hspace{-6mm}\caption{Age 8 - scaled investment}
			\vspace{-4mm}
			\hspace{-6mm}\includegraphics[width=0.95\textwidth, trim={4cm, 8cm, 5cm,  7.5cm},clip  ]{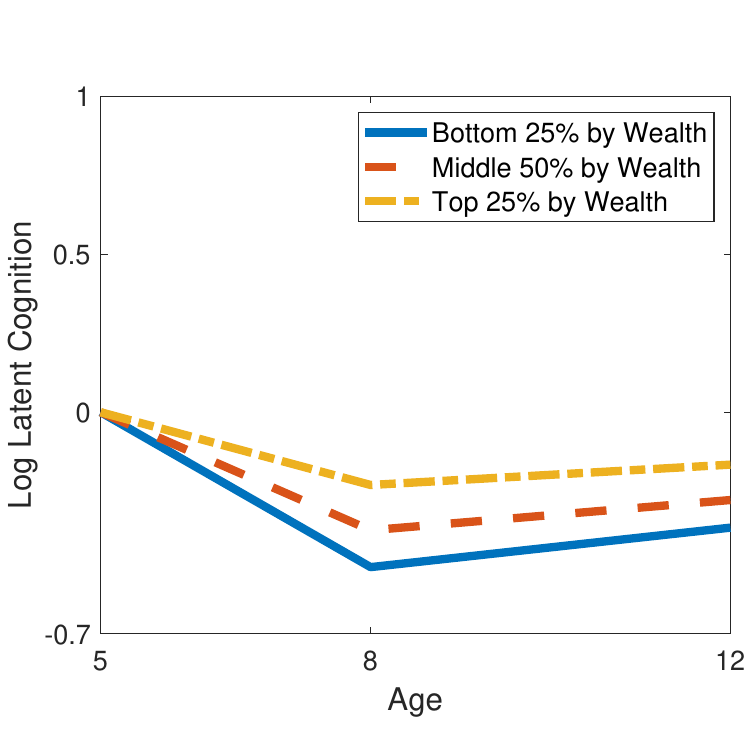}
		\end{subfigure}
		\begin{subfigure}[h]{0.32\textwidth}
			\centering
			\hspace{-6mm}\caption{Age 8 - scaled tests}
			\vspace{-4mm}
			\hspace{-6mm}\includegraphics[width=0.95\textwidth, trim={4cm, 8cm, 5cm,  7.5cm},clip  ]{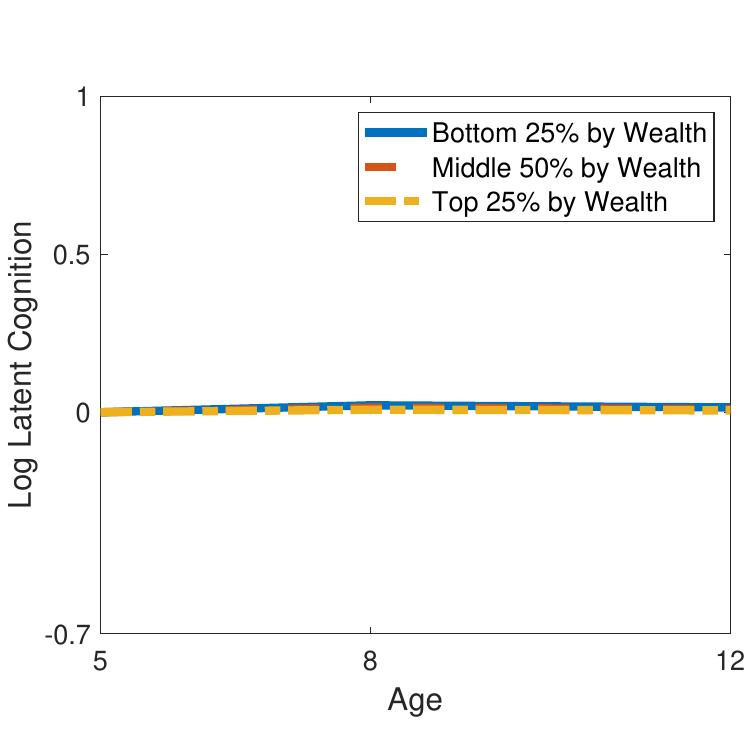}
		\end{subfigure}
		
	\end{center}

	\newgeometry{textwidth=16.6cm}
	{\footnotesize \begin{singlespace} Notes:  The figure shows the impact on cognition of an income transfer equal to 25\% of mean income in the entire sample. The y-axis represents the impact on cognition, measured as the change in the median of the log of cognition in standard deviation units, of increasing investment by one standard deviation. In the top three graphs, the transfer is made before age 5. In the lower three graphs, it is made between ages 5 and 8. Estimates are based on AMN.    \end{singlespace}}
\end{figure}

\begin{figure}[h!]
	\caption{Estimated effects of income transfers - flexible estimator - skill level}
	\label{fig:amn_5_flexible_level}

	\newgeometry{textwidth=17.2cm}
	
	\begin{center}
		
		\vspace{-6mm}
		
		\begin{subfigure}[h]{0.32\textwidth}
			\centering
			\hspace{-6mm} \caption{Age 5 - original scale}
			\vspace{-4mm}
			\hspace{-6mm}\includegraphics[width=0.95\textwidth, trim={4cm, 8cm, 5cm,  7.5cm},clip  ]{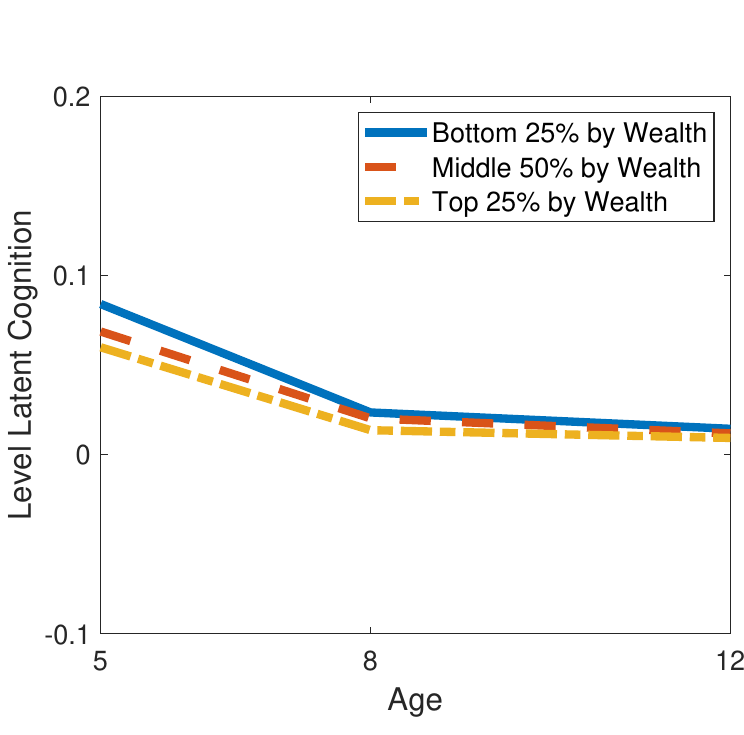}
		\end{subfigure}
		\begin{subfigure}[h]{0.32\textwidth}
			\centering
			\hspace{-6mm} \caption{Age 5 - scaled investment}
			\vspace{-4mm}
			\hspace{-6mm}\includegraphics[width=0.95\textwidth, trim={4cm, 8cm, 5cm,  7.5cm},clip  ]{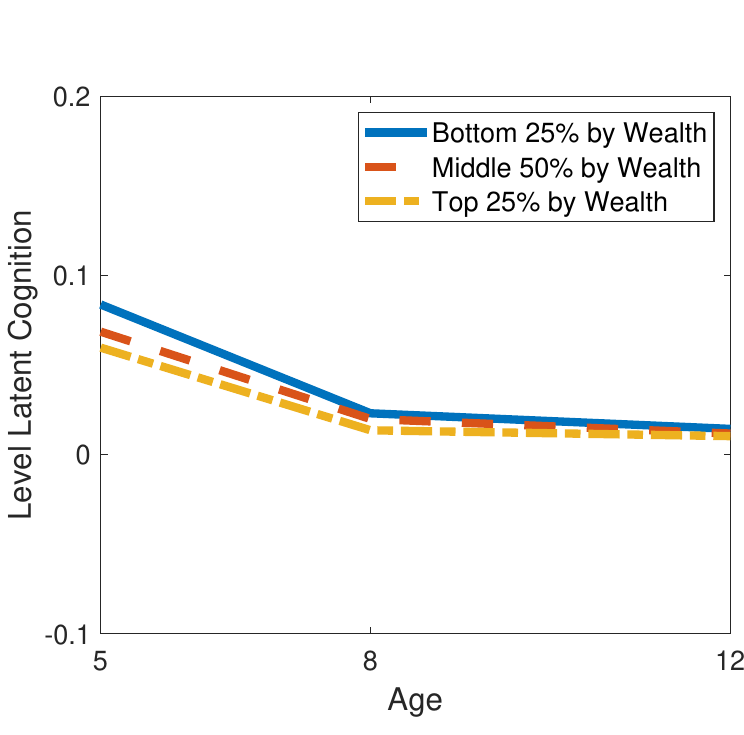}
		\end{subfigure}
		\begin{subfigure}[h]{0.32\textwidth}
			\centering
			\hspace{-6mm}\caption{Age 5 - scaled tests}
			\vspace{-4mm}
			\hspace{-6mm}\includegraphics[width=0.95\textwidth, trim={4cm, 8cm, 5cm,  7.5cm},clip  ]{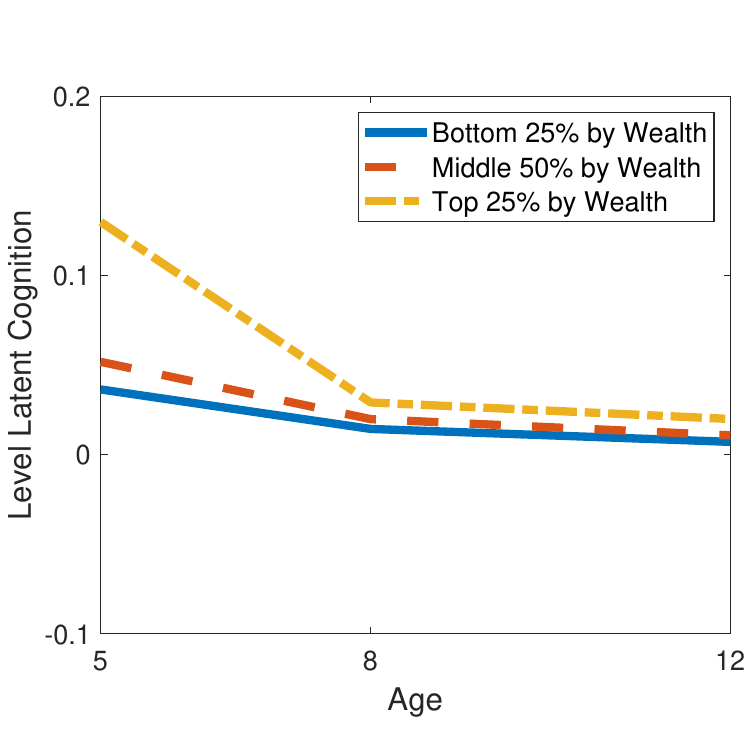}
		\end{subfigure}
		
		\begin{subfigure}[h]{0.32\textwidth}
			\centering
			\hspace{-6mm}\caption{Age 8 - original scale}
			\vspace{-4mm}
			\hspace{-6mm}\includegraphics[width=0.95\textwidth, trim={4cm, 8cm, 5cm,  7.5cm},clip  ]{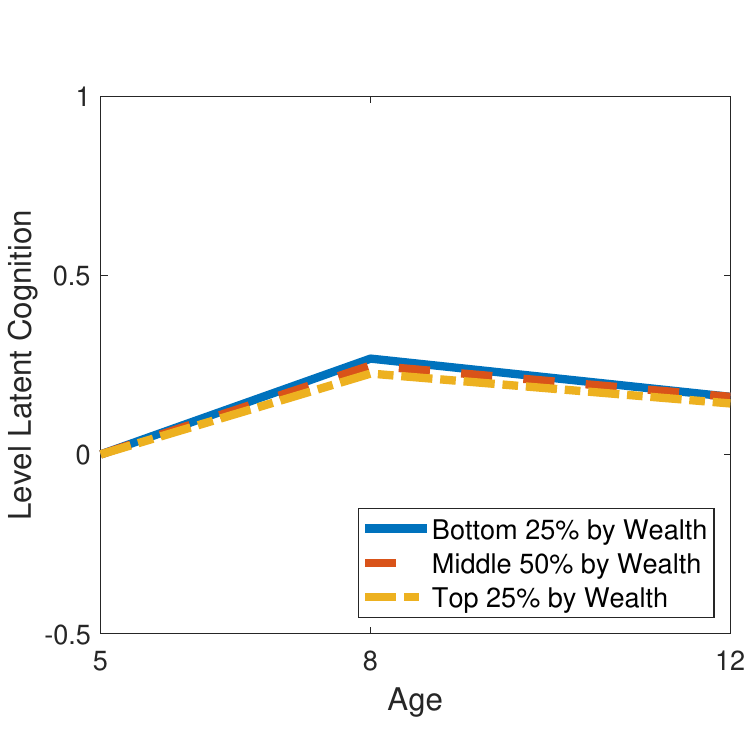}	
		\end{subfigure}
		\begin{subfigure}[h]{0.32\textwidth}
			\centering
			\hspace{-6mm}\caption{Age 8 - scaled investment}
			\vspace{-4mm}
			\hspace{-6mm}\includegraphics[width=0.95\textwidth, trim={4cm, 8cm, 5cm,  7.5cm},clip  ]{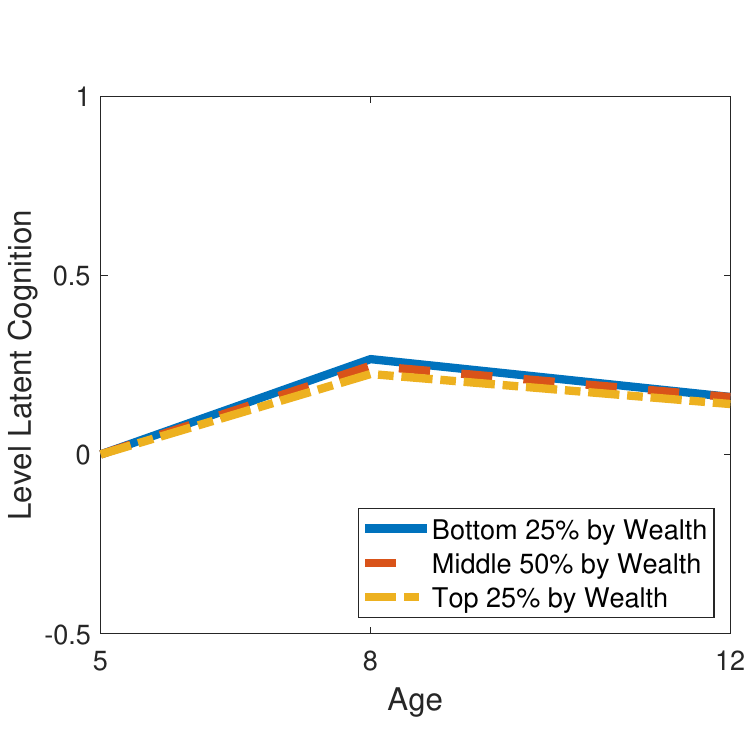}
		\end{subfigure}
		\begin{subfigure}[h]{0.32\textwidth}
			\centering
			\hspace{-6mm}\caption{Age 8 - scaled tests}
			\vspace{-4mm}
			\hspace{-6mm}\includegraphics[width=0.95\textwidth, trim={4cm, 8cm, 5cm,  7.5cm},clip  ]{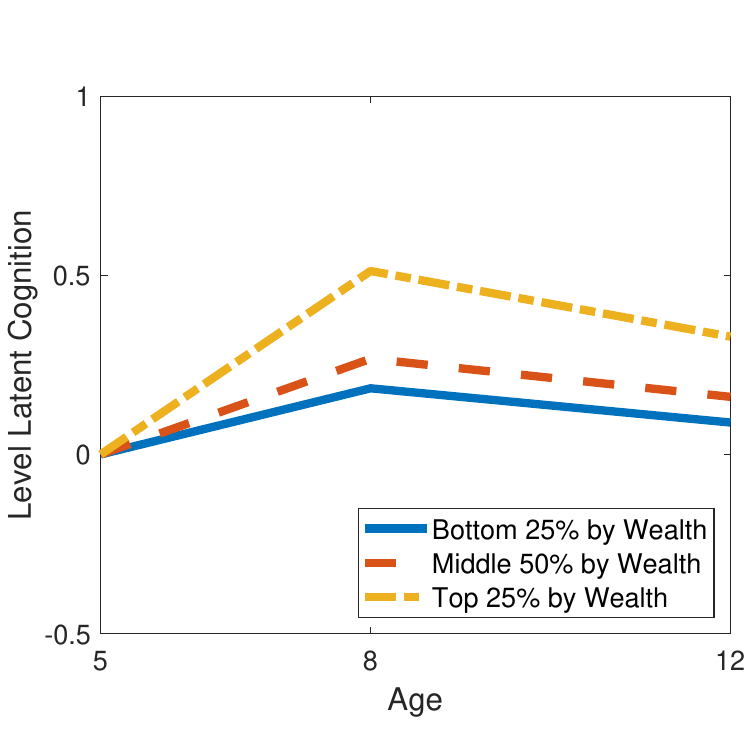}
		\end{subfigure}
		
	\end{center}

	\newgeometry{textwidth=16.6cm}
	{\footnotesize \begin{singlespace} Notes: As Figure \ref{fig:amn_5_restricted_level} but using a flexible estimator. \end{singlespace}}
\end{figure}

\clearpage
\begin{figure}[h!]
	\caption{Distributional effects of income transfers}
	\label{fig:amn_densities}

	\newgeometry{textwidth=17.2cm}
	
	\begin{center}
		
		\vspace{-6mm}
		
		\begin{subfigure}[h]{0.32\textwidth}
			\centering
			\hspace{-6mm} \caption{Age 5 - Bottom 25\%}
			\vspace{-4mm}
			\hspace{-6mm}\includegraphics[width=0.95\textwidth, trim={4cm, 8cm, 5cm,  7.5cm},clip  ]{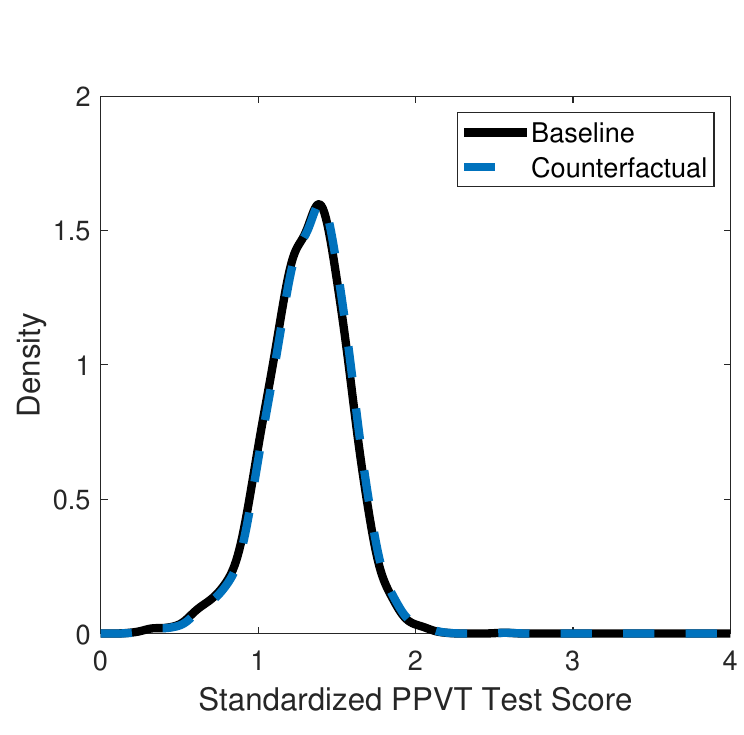}
		\end{subfigure}
		\begin{subfigure}[h]{0.32\textwidth}
			\centering
			\hspace{-6mm} \caption{Age 5 - Middle 50\%}
			\vspace{-4mm}
			\hspace{-6mm}\includegraphics[width=0.95\textwidth, trim={4cm, 8cm, 5cm,  7.5cm},clip  ]{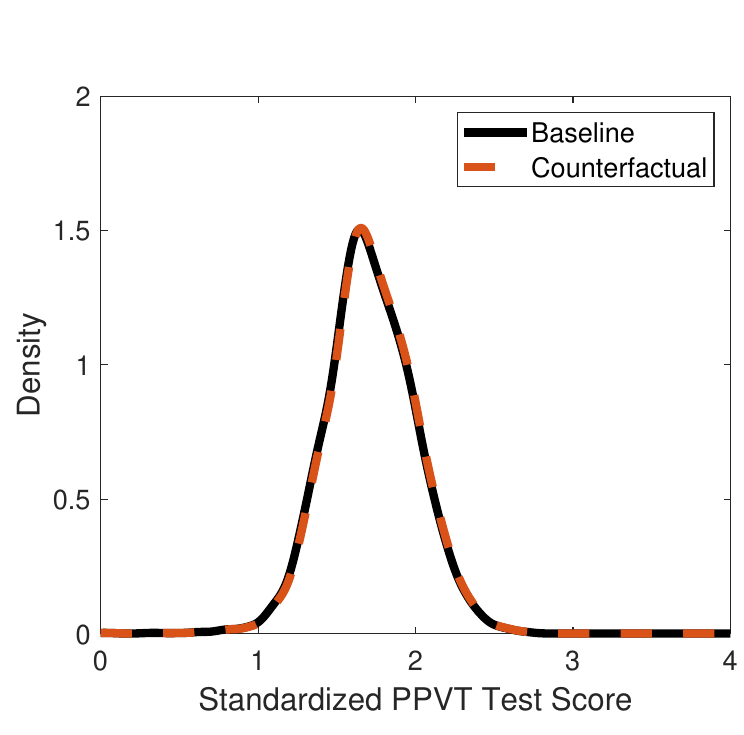}
		\end{subfigure}
		\begin{subfigure}[h]{0.32\textwidth}
			\centering
			\hspace{-6mm}\caption{Age 5 - Top 25\%}
			\vspace{-4mm}
			\hspace{-6mm}\includegraphics[width=0.95\textwidth, trim={4cm, 8cm, 5cm, 7.5cm},clip  ]{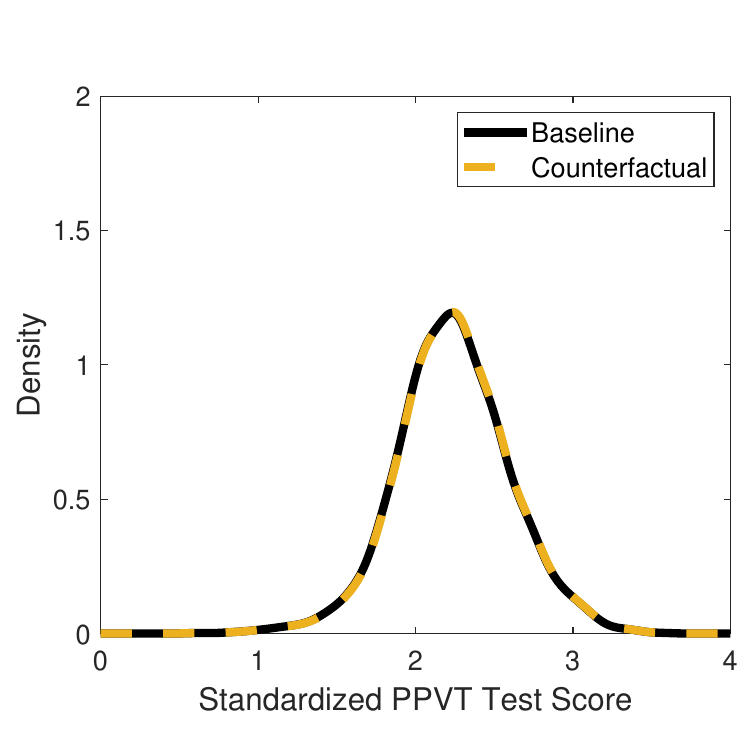}
		\end{subfigure}
		
		\begin{subfigure}[h]{0.32\textwidth}
			\centering
			\hspace{-6mm}\caption{Age 8  - Bottom 25\%}
			\vspace{-4mm}
			\hspace{-6mm}\includegraphics[width=0.95\textwidth, trim={4cm, 8cm, 5cm,  7.5cm},clip  ]{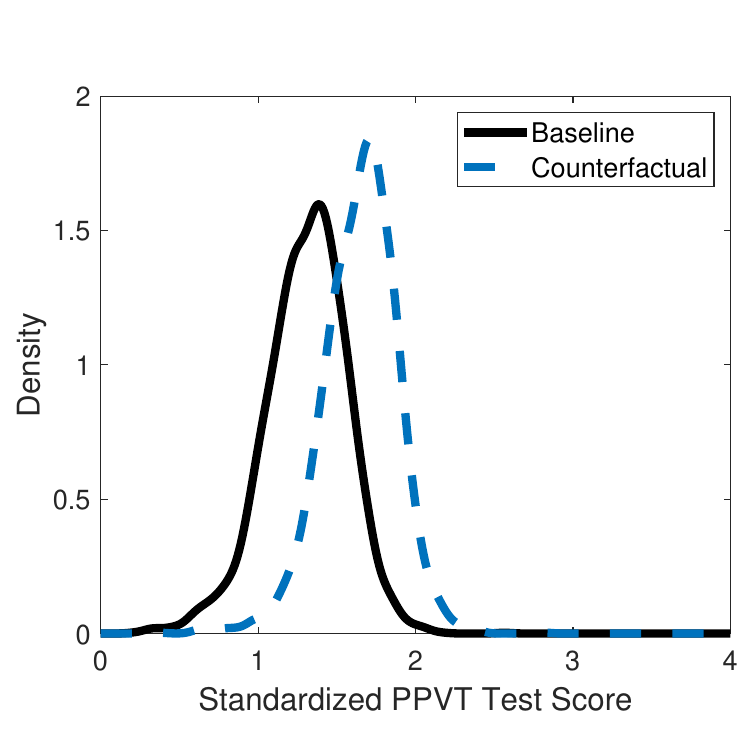}	\end{subfigure}
		\begin{subfigure}[h]{0.32\textwidth}
			\centering
			\hspace{-6mm}\caption{Age 8 - Middle 50\%}
			\vspace{-4mm}
			\hspace{-6mm}\includegraphics[width=0.95\textwidth, trim={4cm, 8cm, 5cm,  7.5cm},clip  ]{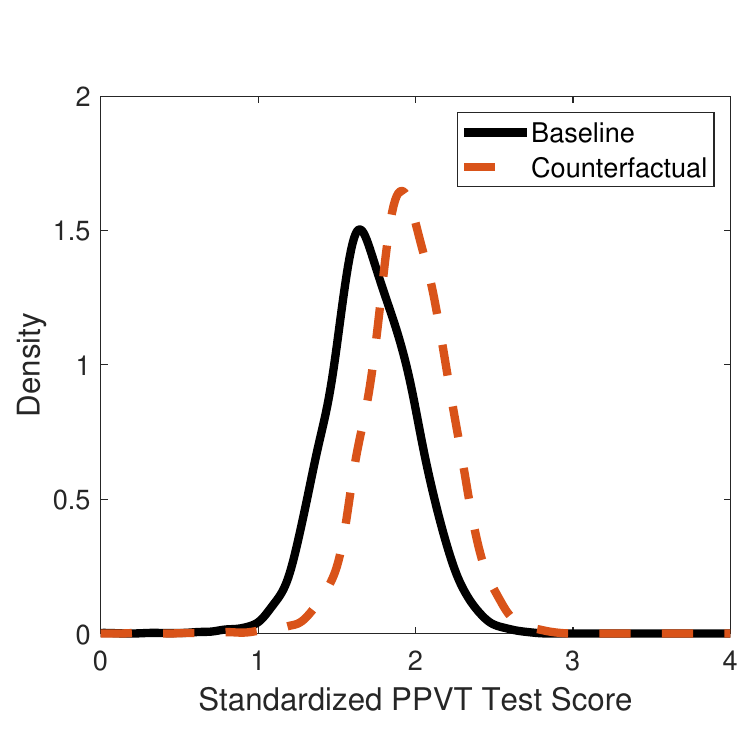}
		\end{subfigure}
		\begin{subfigure}[h]{0.32\textwidth}
			\centering
			\hspace{-6mm}\caption{Age 8  - Top 25\%}
			\vspace{-4mm}
			\hspace{-6mm}\includegraphics[width=0.95\textwidth, trim={4cm, 8cm, 5cm,  7.5cm},clip  ]{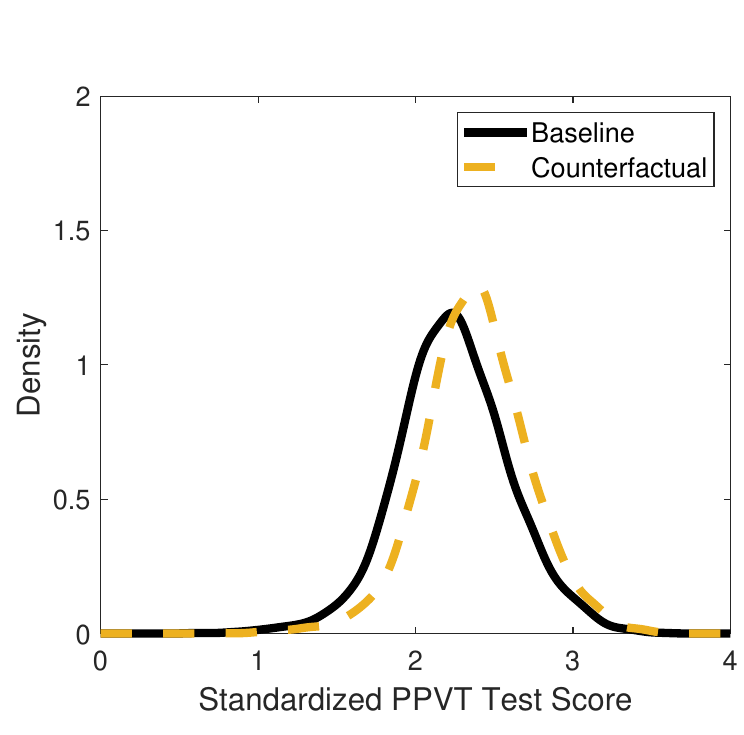}
		\end{subfigure}
		
	\end{center}

	\newgeometry{textwidth=16.6cm}
	{\footnotesize \begin{singlespace}  Notes: This figure shows counterfactual distributions of the standardized PPVT test scores at age 12. The counterfactuals are based on an income transfer equal to 25\% of mean income in the entire sample. In the top three graphs, the transfer is made at age 5. In the lower three graphs, it is made at age 8.  \end{singlespace}}
\end{figure}

\section{Proofs}

\begin{proof}[Proof of Lemma \ref{l:identjoint}]

	We have
	$$ 	\begin{pmatrix}
		cov(Z_{\theta, s,1},Z_{\theta, t,1}) \\ cov(Z_{\theta, s,1},Z_{\theta, t,2}) 	\end{pmatrix} = \begin{pmatrix}
		\lambda_{\theta,s,1} \lambda_{\theta,t,1} cov(\ln \theta_{s}  , \ln \theta_{t} ) \\ \lambda_{\theta,s,1} \lambda_{\theta,t,2} cov(\ln \theta_{s}  , \ln \theta_{t} )
	\end{pmatrix}
	$$
	and
	$$ 	\begin{pmatrix}
		cov(Z_{I, s,1},Z_{\theta, t,1}) \\ cov(Z_{I, s,1},Z_{\theta, t,2}) 	\end{pmatrix} = \begin{pmatrix}
		\lambda_{I,s,1} \lambda_{\theta,t,1} cov(\ln I_{s}  , \ln \theta_{t} ) \\ \lambda_{I,s,1} \lambda_{\theta,t,2} cov(\ln I_{s}  , \ln \theta_{t} )
	\end{pmatrix}
	$$
	Since $\lambda_{\theta,t,m} \neq 0$ and $\lambda_{I,t,m} \neq 0$ for all $m$ and $t$ and $cov(\ln I_{s}  , \ln \theta_{t} ) \neq 0$ or $cov(\ln \theta_{s}  , \ln \theta_{t} ) \neq 0$ for some $s$, it follows that  $ \lambda_{\theta,t,1}/\lambda_{\theta,t,2}$ is identified for all $t$. Analogously,  $ \lambda_{I,t,1}/\lambda_{I,t,2}$ is identified for all $t$. It now follows that for $m \neq m'$
	$$ \mu_{\theta,t,m} - \frac{\lambda_{\theta,t,m}}{\lambda_{\theta,t,m'}}\mu_{\theta,t,m'} = E\left[Z_{\theta,t,m} - \frac{\lambda_{\theta,t,m}}{\lambda_{\theta,t,m'}} Z_{\theta,t,m'} \right]$$
	is identified as well. Now write
	\begin{align*}
		Z_{\theta,t,m} &= \mu_{\theta,t,m} + \lambda_{\theta,t,m} \ln \theta_{t} + \eps_{\theta,t,m} \\
		\frac{\lambda_{\theta,t,m}}{\lambda_{\theta,t,m'}} Z_{\theta,t,m'} + \mu_{\theta,t,m} - \frac{\lambda_{\theta,t,m}}{\lambda_{\theta,t,m'}}\mu_{\theta,t,m'} &= \mu_{\theta,t,m}   + \lambda_{\theta,t,m} \ln \theta_{t} + \frac{\lambda_{\theta,t,m}}{\lambda_{\theta,t,m'}}\eps_{\theta,t,m'} 
	\end{align*}
	Lemma 1 of \citeN{EW:12} implies that the distributions of $\mu_{\theta,t,m} + \lambda_{\theta,t,m} \ln \theta_{t} $ and $\eps_{\theta,t,m}$ are identified conditional on $Y_1, \ldots Y_{T-1}$. Since $Q$ is another measure in period $T$, the distribution of $\eta_Q$ is identified as well. Similarly, the distributions of $\mu_{I,t,m} + \lambda_{I,t,m} \ln I_{t} $ and $\eps_{I,t,m}$ are identified conditional on $Y_1, \ldots Y_{T-1}$.

	Next, condition on $Y_1, \ldots Y_{T-1}$ and for any random variable $X$, let $\varphi_X$ be its conditional characteristic function. Then
	\begin{align*}
		&  \varphi_{\{Z_{\theta,t,m} \}_{t=0,\dots, T,m =1,2},  \{ Z_{I,t,m}   \}_{t=0,\dots, T-1,m =1,2},Q}(s) \\
		&   \qquad =  \varphi_{\{\mu_{\theta,t,m} + \lambda_{\theta,t,m} \ln \theta_{t}\}_{t=0,\dots, T,m =1,2},  \{ \mu_{I,t,m} + \lambda_{I,t,m} \ln I_{t} \}_{t=0,\dots, T-1,m =1,2}, \rho_0 + \rho_1 \ln \theta_{T}}(s) \\
		& \qquad \qquad  \times \varphi_{\{\eps_{\theta,t,m} \}_{t=0,\dots, T,m =1,2},  \{\eps_{I,t,m}   \}_{t=0,\dots, T-1,m =1,2},\eta}(s). 
	\end{align*}
	Since the measurement errors are independent and the marginal distributions are identified,  the joint distribution is identified. That is, 
	$ \varphi_{\{\eps_{\theta,t,m} \}_{t=0,\dots, T,m =1,2},  \{\eps _{I,t,m}   \}_{t=0,\dots, T-1,m =1,2},\eta}(s) $
	is identified. It follows that $\varphi_{\{\mu_{\theta,t,m} + \lambda_{\theta,t,m} \ln \theta_{t}\}_{t=0,\dots, T,m =1,2},  \{ \mu_{I,t,m} + \lambda_{I,t,m} \ln I_{t} \}_{t=0,\dots, T-1,m =1,2}, \rho_0 + \rho_1 \ln \theta_{T}}(s)$ is identified for all nonzeros of  $ \varphi_{\{\eps_{\theta,t,m} \}_{t=0,\dots, T,m =1,2},  \{\eps _{I,t,m}   \}_{t=0,\dots, T-1,m =1,2},\eta}(s) $.
	Since the zeros of this function are isolated and since characteristics functions are continuous, the characteristic function of $\{\mu_{\theta,t,m} + \lambda_{\theta,t,m} \ln \theta_{t}\}_{t=0,\dots, T,m =1,2},  \{ \mu_{I,t,m} + \lambda_{I,t,m} \ln I_{t} \}_{t=0,\dots, T-1,m =1,2}, \rho_0 + \rho_1 \ln \theta_{T}$ and therefore its joint distribution is identified.
\end{proof}	

\begin{proof}[Proof of Theorem \ref{th:identparams_tl}]
	
	I first show that $\{\tilde{\mu}_{\theta,t,m},\tilde{\lambda}_{\theta,t,m}\}_{t=0,\ldots,T,m=1,2}$, $\{\tilde{\mu}_{I,t,m},\tilde{\lambda}_{I,t,m}\}_{t=0,\ldots,T-1,m=1,2}$, $\{\tilde{a}_t,\tilde{\gamma}_{1t},\tilde{\gamma}_{2t}, \tilde{\gamma}_{3t}\}^{T-1}_{t=0}$, $\{ \tilde{\beta}_{0t}, \tilde{\beta}_{1t}, \tilde{\beta}_{2t}\}^{T-1}_{t=0}$, and $(\tilde{\rho}_0,\tilde{\rho}_1)$ are identified. By Lemma  \ref{l:identjoint} the joint distribution of $  ( \{  \ln \tilde{\theta}_{t}   \}_{t=0,\dots, T},\{   \ln \tilde{I}_{t} \}_{t=0,\dots, T-1}, Q)$ is point identified conditional on $(Y_0, \ldots, Y_{T-1})$. The proof of that lemma also implies that we can identify the distributions of ${\eps}_{\theta,t,m}$ and ${\eps}_{I,t,m}$ and thus, also $\tilde{\mu}_{\theta,t,m}$, $\tilde{\lambda}_{\theta,t,m}$, $\tilde{\mu}_{I,t,m}$, and $\tilde{\lambda}_{I,t,m}$ for all $m$ and $t$. From the identified joint distribution, we also know 
	$E[\ln \tilde{I}_t \mid   \tilde{\theta}_{t} , Y_{t}] =  \tilde{\beta}_{0t} + \tilde{\beta}_{1t} \ln \tilde{\theta}_{t}  + \tilde{\beta}_{2t} \ln Y_{t}  $
	and thus, we can identify $\tilde{\beta}_{0t}$, $\tilde{\beta}_{1t}$, and $\tilde{\beta}_{2t}$. Similarly, $E[Q \mid \ln \tilde{\theta}_{T} ]$ and therefore $\tilde{\rho}_0$ and $\tilde{\rho}_1$ are identified. 
	
	We can also identify $E[\ln \tilde{\theta}_{t+1} \mid \tilde{\theta}_{t},\tilde{I}_{t},Y_t]$ since
	\begin{align*}
		E[\ln \tilde{\theta}_{t+1} \mid \tilde{\theta}_{t},\tilde{I}_{t},Y_t] &= \tilde{a}_t + \tilde{\gamma}_{1t}\ln \tilde{\theta}_{t} + \tilde{\gamma}_{2t} \ln \tilde{I}_{t} + \tilde{\gamma}_{3t}\ln \tilde{\theta}_{t} \ln \tilde{I}_{t}  + E[\tilde{\eta}_{\theta,t} \mid \tilde{\theta}_{t},\tilde{I}_{t},Y_t] \\
		&= \tilde{a}_t + \tilde{\gamma}_{1t}\ln \tilde{\theta}_{t} + \tilde{\gamma}_{2t} \ln \tilde{I}_{t} + \tilde{\gamma}_{3t}\ln \tilde{\theta}_{t} \ln \tilde{I}_{t} + \lambda_{\theta,t+1,1}  E[\eta_{\theta,t} \mid \tilde{\theta}_{t},\tilde{ \eta}_{I,t},Y_t] \\
		&= \tilde{a}_t + \tilde{\gamma}_{1t}\ln \tilde{\theta}_{t} + \tilde{\gamma}_{2t} \ln \tilde{I}_{t} + \tilde{\gamma}_{3t}\ln \tilde{\theta}_{t} \ln \tilde{I}_{t}  + \lambda_{\theta,t+1,1}  E[\eta_{\theta,t} \mid  \tilde{ \eta}_{I,t} ]\\
		&= \tilde{a}_t + \tilde{\gamma}_{1t}\ln \tilde{\theta}_{t} + \tilde{\gamma}_{2t} \ln \tilde{I}_{t} + \tilde{\gamma}_{3t}\ln \tilde{\theta}_{t} \ln \tilde{I}_{t}   + (\lambda_{\theta,t+1,1}/\lambda_{I,t,1} ) \kappa_t  \tilde{ \eta}_{I,t}
	\end{align*} 
	The second line follows because conditioning on $\tilde{\theta}_{t},\tilde{I}_{t},Y_t$ or $\tilde{\theta}_{t},\tilde{\eta}_{I,t},Y_t$ is identical due to the linear relation in equation (\ref{eq:investment_rw}). The third and fourth line follow  from Assumption \ref{a:baseline}(h) and the relationship  $ \tilde{\eta}_{I,t} =  \lambda_{I,t,1} {\eta}_{I,t}$. Hence, we can identify $ \tilde{a}_t $, $\tilde{\gamma}_{1t}$, $\tilde{\gamma}_{2t}$, $\tilde{\gamma}_{3t}$, and $(\lambda_{\theta,t+1,1}/\lambda_{I,t,1} ) \kappa_t $ by regressing  $\ln \tilde{\theta}_{t+1}$ on $\ln \tilde{\theta}_{t}$, $ \ln \tilde{I}_{t}$, $\ln \tilde{I}_{t}\ln \tilde{\theta}_{t}$, and $\tilde{ \eta}_{I,t} = \ln \tilde{I}_t -   \tilde{\beta}_{0t} -  \tilde{\beta}_{1t} \ln \tilde{\theta}_{t}  -  \tilde{\beta}_{2t} \ln Y_{t} $. 
	
	Now suppose we have an alternative set of parameters, denoted by $\{\bar{\mu}_{\theta,t,m},\bar{\lambda}_{\theta,t,m}\}_{t=0,\ldots,T,m=1,2}$, $\{\bar{\mu}_{I,t,m},\bar{\lambda}_{I,t,m}\}_{t=0,\ldots,T-1,m=1,2}$, $\{\bar{a}_t,\bar{\gamma}_{1t},\bar{\gamma}_{2t}, \bar{\gamma}_{3t}\}^{T-1}_{t=0}$, $\{ \bar{\beta}_{0t}, \bar{\beta}_{1t}, \bar{\beta}_{2t}\}^{T-1}_{t=0}$, and $(\bar{\rho}_0,\bar{\rho}_1)$  that yields the same values of $\{\tilde{\mu}_{\theta,t,m},\tilde{\lambda}_{\theta,t,m}\}_{t=0,\ldots,T,m=1,2}$, $\{\tilde{\mu}_{I,t,m},\tilde{\lambda}_{I,t,m}\}_{t=0,\ldots,T-1,m=1,2}$, $\{ \tilde{\beta}_{0t}, \tilde{\beta}_{1t}, \tilde{\beta}_{2t}\}^{T-1}_{t=0}$,  $(\tilde{\rho}_0,\tilde{\rho}_1)$, and  $\{\tilde{a}_t,\tilde{\gamma}_{1t},\tilde{\gamma}_{2t}, \tilde{\gamma}_{3t}\}^{T-1}_{t=0}$. Define $\bar{\theta}_t$ and $\bar{I}_t$ such that 
	$$\bar{\mu}_{\theta,t,1} + \bar{\lambda}_{\theta,t,1} \ln \bar{\theta}_t = {\mu}_{\theta,t,1} + {\lambda}_{\theta,t,1} \ln {\theta}_t = \ln \tilde{\theta}_t$$
	$$\bar{\mu}_{I,t,1} + \bar{\lambda}_{I,t,1} \ln \bar{I}_t = {\mu}_{I,t,1} + {\lambda}_{I,t,1} \ln {I}_t = \ln \tilde{I}_t.$$
	For $m \neq 1$, we then have
	\begin{align*}
		\bar{\mu}_{\theta,t,m} + \bar{\lambda}_{\theta,t,m} \ln \bar{\theta}_t &= \bar{\mu}_{\theta,t,m} + \bar{\lambda}_{\theta,t,m}\left( \ln \tilde{\theta}_t - \bar{\mu}_{\theta,t,1}\right)/\bar{\lambda}_{\theta,t,1} \\
		&=\bar{\mu}_{\theta,t,m} - \frac{\bar{\lambda}_{\theta,t,m}}{\bar{\lambda}_{\theta,t,1}} + \frac{\bar{\lambda}_{\theta,t,m}}{\bar{\lambda}_{\theta,t,1}}\ln \tilde{\theta}_t \\
		&=\tilde{\mu}_{\theta,t,m}   + \tilde{\lambda}_{\theta,t,m} \ln \tilde{\theta}_t  
	\end{align*}
	and similarly
	$
	\bar{\mu}_{I,t,m} + \bar{\lambda}_{I,t,m} \ln \bar{I}_t  =\tilde{\mu}_{I,t,m}   + \tilde{\lambda}_{I,t,m} \ln \tilde{I}_t. 
	$
	Then the two models generate the same distribution of the measures. In addition, the parameters and $\bar{\theta}_t$ are consistent with the production technology
	because
	\begin{align*}
		\bar{\mu}_{\theta,t+1,1}  + \bar{\lambda}_{\theta,t+1,1} \ln \bar{\theta}_{t+1} &= \ln \tilde{\theta}_{t+1} \\
		&=  \tilde{a}_t + \tilde{\gamma}_{1t} \ln \tilde{\theta}_{t} +  \tilde{\gamma}_{2t} \ln \tilde{I}_t  +  \tilde{\gamma}_{3t} \ln \tilde{\theta}_{t} \ln \tilde{I}_t + \tilde{\eta}_{\theta,t}  \\
		&=  \tilde{a}_t + \tilde{\gamma}_{1t} (\bar{\mu}_{\theta,t,1} + \bar{\lambda}_{\theta,t,1} \ln \bar{\theta}_t) +  \tilde{\gamma}_{2t} \ln (\bar{\mu}_{I,t,1} + \bar{\lambda}_{I,t,1} \ln \bar{I}_t) \\
		& \quad +  \tilde{\gamma}_{3t} \ln (\bar{\mu}_{\theta,t,1} + \bar{\lambda}_{\theta,t,1} \ln \bar{\theta}_t) (\bar{\mu}_{I,t,1} + \bar{\lambda}_{I,t,1} \ln \bar{I}_t) + \tilde{\eta}_{\theta,t}. 
	\end{align*}
	Since
	$$\tilde{a}_t =  \bar{\lambda}_{\theta,t+1,1} \bar{a}_t + \bar{\mu}_{\theta,t+1,1} - \frac{\bar{\lambda}_{\theta,t+1,1}}{\bar{\lambda}_{\theta,t,1}} \bar{\mu}_{\theta,t,1} \bar{\gamma}_{1t} - \frac{\bar{\lambda}_{\theta,t+1,1}}{\bar{\lambda}_{I,t,1}} \mu_{I,t,1} \bar{\gamma}_{2t} +  \frac{\bar{\lambda}_{\theta,t+1,1}}{\bar{\lambda}_{\theta,t,1} \bar{\lambda}_{I,t,1}} \bar{\mu}_{I,t,1}  \bar{\mu}_{\theta,t,1}\bar{\gamma}_{3t} $$
	$$ \tilde{\gamma}_{1t}  = \frac{\bar{\lambda}_{\theta,t+1,1}}{\bar{\lambda}_{\theta,t,1}} \left( \bar{\gamma}_{1t} - \frac{\bar{\mu}_{I,t,1}}{\bar{\lambda}_{I,t,1}}\bar{\gamma}_{3t}  \right), \;\;   \tilde{\gamma}_{2t}  =  \frac{ \bar{\lambda}_{\theta,t+1,1}}{\bar{\lambda}_{I,t,1}} \left( \bar{\gamma}_{2t} - \frac{\bar{\mu}_{\theta,t,1}}{\bar{\lambda}_{\theta,t,1}}\bar{\gamma}_{3t} \right), \;\;  \tilde{\gamma}_{3t}  = \frac{\bar{\lambda}_{\theta,t+1,1}}{\bar{\lambda}_{\theta,t,1} \bar{\lambda}_{I,t,1}}\bar{\gamma}_{3t}$$
	and $ \tilde{\eta}_{\theta,t} =   \lambda_{\theta,t+1,1}{\eta}_{\theta,t}$, it is easy to show that with $\bar{\eta}_{\theta,t} = \frac{\lambda_{\theta,t+1,1}}{\bar{\lambda}_{\theta,t+1,1}} {\eta}_{\theta,t}$,
	\begin{align*}
		\ln \bar{\theta}_{t+1}  
		&=  \bar{a}_t + \bar{\gamma}_{1t}  \ln \bar{\theta}_t  +  \bar{\gamma}_{2t} \ln \bar{I}_t  +  \bar{\gamma}_{3t} \ln  \bar{\theta}_t \ln \bar{I}_t + \bar{\eta}_{\theta,t}.
	\end{align*}

	Analogously, one can show that 
	\begin{align*}
		\ln \bar{I}_t &=  \bar{\beta}_{0t} + \bar{\beta}_{1t} \ln \bar{\theta}_{t}  + \bar{\beta}_{2t} \ln Y_{t}   + \bar{\eta}_{I,t}  \\
		Q &=  \bar{\rho}_{0} + \bar{\rho}_{1} \ln \bar{\theta}_{T} + \bar{\eta}_Q  
	\end{align*}

	For the last part first notice that once  $\{\lambda_{\theta,t,1},\lambda_{I,t,1},\mu_{\theta,t,1},\mu_{I,t,1}\}_{t=1,\ldots, T}$ are fixed, the vector $\{\lambda_{\theta,t,m},\lambda_{I,t,m},\mu_{\theta,t,m},\mu_{I,t,m}\}_{t=1,\ldots, T}$ is uniquely determined for all $m \neq 1$ due to the restrictions the parameters in the identified set have to satisfy. From these restrictions, we can then uniquely determine $\gamma_{3t}$ and then also $\gamma_{1t}$ and $\gamma_{2t}$. Similarly, all the other parameters are uniquely determined. We then found the unique parameter vector for which $\{\lambda_{\theta,t,1},\lambda_{I,t,1},\mu_{\theta,t,1},\mu_{I,t,1}\}_{t=1,\ldots, T}$ are fixed and that yields the same values of  $\{\tilde{a}_t,\tilde{\gamma}_{1t},\tilde{\gamma}_{2t}, \tilde{\gamma}_{3t}\}^{T-1}_{t=0}$, $(\tilde{\rho}_0,\tilde{\rho}_1)$, $\{\tilde{\mu}_{\theta,t,m},\tilde{\lambda}_{\theta,t,m}\}_{t=0,\ldots,T,m=1,2}$, $\{\tilde{\mu}_{I,t,m},\tilde{\lambda}_{I,t,m}\}_{t=0,\ldots,T-1,m=1,2}$, and $\{ \tilde{\beta}_{0t}, \tilde{\beta}_{1t}, \tilde{\beta}_{2t}\}^{T-1}_{t=0}$.
\end{proof}

\begin{proof}[Proof of Corollary \ref{c:pointident}]
	The last part of Theorem \ref{th:identparams_tl} immediately implies that all parameters are point identified under Assumptions \ref{a:baseline}, \ref{a:normalization}, \ref{a:ageinvariant_technology_skills}(a), and \ref{a:ageinvariant_technology_investment}(a). Now suppose Assumptions \ref{a:baseline}, \ref{a:normalization}, \ref{a:ageinvariant_technology_skills}(b), and \ref{a:ageinvariant_technology_investment}(a) hold. Then, since $a_t = \mu_{I,t,1}= 0$, $\mu_{\theta,t+1,1} -  \mu_{\theta,t,1} \frac{\lambda_{\theta,t+1,1}}{\lambda_{\theta,t,1}}  \gamma_{1t}$ is identified. Moreover, $\frac{\lambda_{\theta,t+1,1}}{\lambda_{\theta,t,1}}  \gamma_{1t}$ is identified. We also know that $\mu_{\theta,0,1}  = 0$, which identifies $\mu_{\theta,1,1} $ and then we can recursively identify $\mu_{\theta,t,1} $ for all $t$. In addition, we can identify $\lambda_{\theta,t+1,1} \gamma_{2t} $ because $\lambda_{\theta,t+1,1} \left( \gamma_{2t} - \frac{\mu_{\theta,t,1}}{\lambda_{\theta,t,1}}\gamma_{3t} \right)$, $ \frac{\lambda_{\theta,t+1,1} }{\lambda_{\theta,t,1}}\gamma_{3t}$, and $\mu_{\theta,t,1}$ are identified. We can also identify 
	$$\frac{\lambda_{\theta,t+1,1}}{\lambda_{\theta,t,1}}\gamma_{1t} = \frac{\lambda_{\theta,t+1,1}}{\lambda_{\theta,t,1}}(1 -\gamma_{2t}-\gamma_{3t} ) =    \frac{\lambda_{\theta,t+1,1}}{\lambda_{\theta,t,1}} - \frac{\lambda_{\theta,t+1,1}}{\lambda_{\theta,t,1}}\gamma_{2t}- \frac{\lambda_{\theta,t+1,1}}{\lambda_{\theta,t,1}}\gamma_{3t}.   $$  
	Since $\lambda_{\theta,0,1}=1$, we can identify $\lambda_{\theta,1,1}$ and then, using the previous equation, $\lambda_{\theta,t,1}$ recursively for all $t$. Once we know $\mu_{\theta,t,1}$ and $\lambda_{\theta,t,1}$ for all $t$, identification of the production function parameters follows immediately. In addition, since $\frac{\lambda_{\theta,t,m}}{\lambda_{\theta,t,1}}$ and $\mu_{\theta,t,m} + \left(\lambda_{\theta,t,m}/\lambda_{\theta,t,1}\right) \mu_{\theta,t,1}$ are identified, we can identify $\lambda_{\theta,t,m}$ and $\mu_{\theta,t,m}$ for all $m$. It then also follows that $\beta_{0t}$, $\beta_{1t}$, $\beta_{2t}$, $\rho_{0}$, and $\rho_{1}$ are identified.   
	
	Now suppose Assumptions \ref{a:baseline}, \ref{a:normalization}, \ref{a:ageinvariant_technology_skills}(a), and \ref{a:ageinvariant_technology_investment}(b) hold. Then, using the argument as in the proof of Theorem \ref{th:identparams_tl}, $\{\lambda_{\theta,t,m}, \mu_{\theta,t,m},\}_{t=1,\ldots, T}$ is uniquely determined. Next, using the expressions for $\tilde{\beta}_{1t}$ and $\tilde{\beta}_{2t}$ and $\beta_{1t} + \beta_{2t} = 1$, we can identify $\beta_{1t}$, $\beta_{2t}$, and $\lambda_{I,t,1}$. With $\tilde{\beta}_{0t}$ and $\beta_{0t} = 0$, we can uniquely determine $\mu_{I,t,1}$. Once the coefficients in equation (\ref{eq:measurement_eq_ces_invest}) are identified for $m=1$, the second part of Theorem \ref{th:identparams_tl} implies identification of all other parameters.
	
	Finally suppose Assumptions \ref{a:baseline}, \ref{a:normalization}, \ref{a:ageinvariant_technology_skills}(b), and \ref{a:ageinvariant_technology_investment}(b) hold. Then, using $\beta_{0t} = 0$ and $\beta_{1t} + \beta_{2t}= 1$ we can write $\tilde{\beta}_{0t} = \mu_{I,t,1} - \tilde{\beta}_{1t} \mu_{\theta,t,1}$ and $\lambda_{\theta,t,1} \tilde{\beta}_{1t} + \tilde{\beta}_{2t} = \lambda_{I,t,1}  $. Now assume that $\mu_{\theta,t,1}$ and $\lambda_{\theta,t,1}$, are identified. Then $\mu_{I,t,1}$ and $\lambda_{I,t,1}$ are also identified. Notice that since  $\gamma_{1t} + \gamma_{2t} + \gamma_{2t} = 1$, we have 
	$$\frac{\tilde{\gamma}_{1t}}{\lambda_{I,t,1}}  + \frac{\tilde{\gamma}_{2t}}{ \lambda_{\theta,t,1}}   +\tilde{\gamma}_{3t}  = \frac{ \lambda_{\theta,t+1,1}}{{\lambda_{\theta,t,1}\lambda_{I,t,1}}} - \mu_{I,t,1}\frac{\lambda_{\theta,t,1}}{\lambda_{I,t,1}}\tilde{\gamma}_{3t} - \mu_{\theta,t,1}\frac{\lambda_{I,t,1}}{\lambda_{\theta,t,1}}\tilde{\gamma}_{3t}   $$
	and thus, $\lambda_{\theta,t+1,1}$ and $\lambda_{I,t+1,1}  = \lambda_{\theta,t+1,1} \tilde{\beta}_{1,t+1} + \tilde{\beta}_{2,t+1}$  are identified. Then $\gamma_{1t}$, $\gamma_{2t}$, and $\gamma_{3t}$ are also identified. Using the expression for $\tilde{a}_t$ together with $a_t = 0$ then identifies $\mu_{\theta,t+1,1}$ and $  \mu_{I,t+1,1}  = \tilde{\beta}_{0,t+1} + \tilde{\beta}_{1,t+1} \mu_{\theta,t+1,1}  $. Since $\mu_{\theta,0,1} = 0$ and $\lambda_{\theta,0,1} =1$, these arguments imply that $\mu_{\theta,t,1}$, $\mu_{I,t,1}$, $\lambda_{\theta,t,1}$, and  $\lambda_{I,t,1}$, are identified for all $t$.  Once the coefficients in the measurement error equation are identified for one of the measures, the second part of Theorem \ref{th:identparams_tl} implies that all other parameters are identified as well.
\end{proof}

\begin{proof}[Proof of Theorem \ref{th:identfunctions}]
	For the first part, notice that 
	\begin{eqnarray*}
		&& \hspace{-6mm} \tilde{s}_{1t}(\alpha_1,\alpha_2,\alpha_3) \equiv    \tilde{a}_t + \tilde{\gamma}_{1t} \ln Q_{\alpha_1}(\tilde{\theta}_{t}) +  \tilde{\gamma}_{2t} \ln Q_{\alpha_2}(\tilde{I}_{t})  +  \tilde{\gamma}_{3t} \ln Q_{\alpha_1}(\tilde{\theta}_{t})\ln Q_{\alpha_2}(\tilde{I}_{t})  + Q_{\alpha_3}(\tilde{\eta}_{\theta,t})     \\   
		&& =  \tilde{a}_t + \tilde{\gamma}_{1t} Q_{\alpha_1}(  \mu_{\theta,t,1} + \lambda_{\theta,t,1} \ln \theta_{t})  + \tilde{\gamma}_{2t}Q_{\alpha_2}(  \mu_{I,t,1} + \lambda_{I,t,1} \ln I_{t})  \\
		& & \quad \; +  \, \tilde{\gamma}_{3t} Q_{\alpha_1}(  \mu_{\theta,t,1} + \lambda_{\theta,t,1} \ln \theta_{t}) Q_{\alpha_2}(  \mu_{I,t,1} + \lambda_{I,t,1} \ln I_{t}) + Q_{\alpha_3}( \lambda_{\theta,t+1,1}  \eta_{\theta,t})  \\
		&&= \mu_{\theta,t+1,1} + \lambda_{\theta,t+1,1} \left(   {a}_t +  {\gamma}_{1t}  Q_{\alpha_1} \ln( \theta_{t})   +  {\gamma}_{2t}  Q_{\alpha_2} (\ln I_{t})   +  {\gamma}_{3t}  Q_{\alpha_1} \ln( \theta_{t})  Q_{\alpha_2} (\ln I_{t}) +   Q_{\alpha_3}(   \eta_{\theta,t})  \right)    \\
		&&= \mu_{\theta,t+1,1}+ \lambda_{\theta,t+1,1}	s_{1t}(\alpha_1,\alpha_2,\alpha_3).
	\end{eqnarray*}
	It follows that $ F_{\ln \tilde{\theta}_{t+1}}(\tilde{s}_{1t}(\alpha_1,\alpha_2,\alpha_3) ) = F_{\ln {\theta}_{t+1}}(s_{1t}(\alpha_1,\alpha_2,\alpha_3))$.
	Since the joint distribution of $ \left( \{  \ln \tilde{\theta}_{t}\}_{t=0,\dots, T},  \{ \ln \tilde{I}_{t} \}_{t=0,\dots, T-1} \right)$   and $\tilde{a}_t$, $\tilde{\gamma}_{1t}$, $\tilde{\gamma}_{2t}$, and $\tilde{\gamma}_{3t}$ are identified and since $\tilde{\eta}_{\theta,t}  = \ln \tilde{\theta}_{t+1} -\left( \tilde{a}_t + \tilde{\gamma}_{1t} \ln \tilde{\theta}_{t} +  \tilde{\gamma}_{2t} \ln \tilde{I}_t  +  \tilde{\gamma}_{3t} \ln \tilde{\theta}_{t} \ln \tilde{I}_t \right) $, it follows that 	 the left hand side is identified. Thus, we get point identification of $F_{\ln {\theta}_{t+1}}(s_{1t}(\alpha_1,\alpha_2,\alpha_3) )$.
	Moreover, it is easy to see that $ \tilde{\mu}_{\theta,t+1,m}+  \tilde{\lambda}_{\theta,t+1,m} \tilde{s}_{1t}(\alpha_1,\alpha_2,\alpha_3) +    Q_{\alpha_4}(\tilde{\eps}_{\theta,t+1,m}) = \mu_{\theta,t+1,1}+ \lambda_{\theta,t+1,1}	s_{1t}(\alpha_1,\alpha_2,\alpha_3)  +  Q_{\alpha_4}(\eps_{\theta,t+1,m}) $, $\frac{\partial \tilde{\lambda}_{\theta,t+1 ,m} \ln \tilde{\theta}_{t+1}}{\partial \tilde{\lambda}_{\theta,t ,m'} \ln \tilde{\theta}_t } \mid_{ \tilde{I}_t = Q_{\alpha_2}(\tilde{I}_t)  } = \frac{\partial \lambda_{\theta,t+1 ,m} \ln \theta_{t+1}}{\partial \lambda_{\theta,t ,m'} \ln \theta_t } \mid_{ I_t = Q_{\alpha_2}(I_t)  }$ and $\frac{\partial \tilde{\lambda}_{\theta,t+1 ,m} \ln \tilde{\theta}_{t+1}}{\partial \tilde{\lambda}_{I,t ,m'} \ln \tilde{I}_t } \mid_{ \tilde{\theta}_t = Q_{\alpha_1}(\tilde{\theta}_t)  } = \frac{\partial {\lambda}_{\theta,t+1 ,m} \ln {\theta}_{t+1}}{\partial {\lambda}_{I,t ,m'} \ln {I}_t } \mid_{ {\theta}_t = Q_{\alpha_1}({\theta}_t)  }$, which are also identified.

	For the second part notice that
	\begin{align*}
		\ln  \tilde{I}_t(y) & = \tilde{\beta}_{0t} + \tilde{\beta}_{1t} \ln Q_{\alpha_1}(\tilde{\theta}_t)  + \tilde{\beta}_{2t} \ln y    +  Q_{\alpha_3}(\tilde{\eta}_{I,t}) \\
		& = \tilde{\beta}_{0t} + \tilde{\beta}_{1t}  (  \mu_{\theta,t,1} + \lambda_{\theta,t,1} Q_{\alpha_1} \ln( \theta_{t}) )  + \tilde{\beta}_{2t} \ln y    +  Q_{\alpha_3}( \lambda_{I,t,1}{\eta}_{I,t}) \\
		& = \mu_{I,t,1} + \lambda_{I,t,1} \left( \beta_{0t} +  \beta_{1t}   Q_{\alpha_1} (\ln( \theta_{t}) )  +  {\beta}_{2t} \ln y    +  Q_{\alpha_3}({\eta}_{I,t}) \right) \\
		& = \mu_{I,t,1} + \lambda_{I,t,1} 	\ln  I_t(y) 
	\end{align*}
	is identified. Then using the same arguments as above, it follows that $$F_{\ln \tilde{\theta}_{t+1}}(\tilde{s}_{2t}(\alpha_1,\alpha_2,\alpha_3,y) ) =F_{\ln {\theta}_{t+1}}({s}_{2t}(\alpha_1,\alpha_2,\alpha_3,y) ) $$
	with $\tilde{s}_{2t}(\alpha_1,\alpha_2,\alpha_3,y) = \tilde{a}_t + \tilde{\gamma}_{1t} \ln Q_{\alpha_1}(\tilde{\theta}_{t}) +  \tilde{\gamma}_{2t} \ln    \tilde{I}_t(y)   +  \tilde{\gamma}_{3t} \ln Q_{\alpha_1}(\tilde{\theta}_{t})\ln  \tilde{I}_t(y)  + Q_{\alpha_2}(\tilde{\eta}_{\theta,t})  $ and $\tilde{\mu}_{\theta,t+1,m} +  \tilde{\lambda}_{\theta,t+1,m}\tilde{s}_{2t}(\alpha_1,\alpha_2,\alpha_3,y) +  Q_{\alpha_4}(\tilde{\eps}_{\theta,t+1,m}) = \mu_{\theta,t+1,m} +  \lambda_{\theta,t+1,m}s_{2t}(\alpha_1,\alpha_2,\alpha_3,y) +  Q_{\alpha_4}(\eps_{\theta,t+1,m})$ 
	are point identified for all $\alpha_1,\alpha_2,\alpha_3,\alpha_4 \in (0,1)$.

	For the third part we have
	\begin{align*}
		&  P\left(Q \leq q   \mid \theta_t= Q_{\alpha_1}(\theta_t), \{I_s = Q_{\alpha_{2s}}(I_s) \}^{T-1}_{s=t}, \{\eta_{\theta,s} = Q_{\alpha_{3s}}(\eta_{\theta,s})   \}^{T-1}_{s=t}  \right) \\
		& =  P\left(Q \leq q   \mid \tilde{\theta}_t= Q_{\alpha_1}(\tilde{\theta}_t), \{\tilde{I}_s = Q_{\alpha_{2s}}(\tilde{I}_s) \}^{T-1}_{s=t}, \{\tilde{\eta}_{\theta,s} = Q_{\alpha_{3s}}(\tilde{\eta}_{\theta,s})   \}^{T-1}_{s=t}  \right)
	\end{align*}
	and hence, the left hand side is identified.  
	
	The fourth part follows from identification of the joint distribution of $(Q,\tilde{\theta}_t, \{Y_s\}^{T-1}_{s=t})$.
	
	Finally, notice that 
	$\tilde{\gamma}_{1t} + \tilde{\gamma}_{3t} \ln \tilde{I}_t = \frac{ \lambda_{\theta,t+1,1}}{ \lambda_{\theta,t,1}}\left(  \gamma_{1t}  + \gamma_{3t} \ln I_t \right) $
	whose distribution is identified under Assumption \ref{a:ageinvariant_technology_skills}(a). 
	
\end{proof}

\begin{proof}[Proof of Theorem \ref{th:identparams_ces} ]
	
	The same arguments as in the proof of Theorem \ref{th:identparams_tl} implies that  $\tilde{\beta}_{0t}$, $\tilde{\beta}_{1t}$, $\tilde{\beta}_{2t}$, $\tilde{\rho}_0$ and $\tilde{\rho}_1$ are identified. Moreover, we can also identify $E[\ln \tilde{\theta}_{t+1} \mid \tilde{\theta}_{t},\tilde{I}_{t},Y_t]$ which we can write as 
	\begin{align*}
		E[\ln \tilde{\theta}_{t+1} \mid \tilde{\theta}_{t},\tilde{I}_{t},Y_t] &=    \frac{\lambda_{\theta,t+1,1} \psi_t}{ {\sigma}_{t}}   \ln \left (\tilde{\gamma}_{1t} \tilde{\theta}_{t}^{\frac{ {\sigma}_t}{\lambda_{\theta,t,1}}} +  \tilde{\gamma}_{2t} \tilde{I}_t^{\frac{ {\sigma}_t}{\lambda_{I,t,1}}}\right)  + E[\tilde{\eta}_{\theta,t} \mid \tilde{\theta}_{t},\tilde{I}_{t},Y_t] \\
		&=  \frac{\lambda_{\theta,t+1,1} \psi_t}{ {\sigma}_{t}}   \ln \left (\tilde{\gamma}_{1t} \tilde{\theta}_{t}^{\frac{ {\sigma}_t}{\lambda_{\theta,t,1}}} +  \tilde{\gamma}_{2t} \tilde{I}_t^{\frac{ {\sigma}_t}{\lambda_{I,t,1}}}\right)   + (\lambda_{\theta,t+1,1}/\lambda_{I,t,1} ) \kappa_t  \tilde{ \eta}_{I,t}
	\end{align*} 
	Hence
	\begin{align*}
		\frac{\frac{\partial E[\ln \tilde{\theta}_{t+1} \mid \tilde{\theta}_{t},\tilde{I}_{t},Y_t]}{\partial\tilde{\theta}_{t} } }{\frac{\partial E[\ln \tilde{\theta}_{t+1} \mid \tilde{\theta}_{t},\tilde{I}_{t},Y_t]}{\partial\tilde{I}_{t} } }&=    \frac{\lambda_{I,t,1}}{\lambda_{\theta,t,1}}  \frac{\tilde{\gamma}_{1t}}{\tilde{\gamma}_{2t}}  \tilde{\theta}_{t}^{\frac{ {\sigma}_t}{\lambda_{\theta,t,1}} -1} \tilde{I}_{t}^{ 1 -\frac{ {\sigma}_t}{\lambda_{I,t,1}}} 
	\end{align*} 
	which identifies $\frac{\lambda_{I,t,1}}{\lambda_{\theta,t,1}}  \frac{\tilde{\gamma}_{1t}}{\tilde{\gamma}_{2t}} $, $\frac{ {\sigma}_t}{\lambda_{\theta,t,1}}$ and   $\frac{ {\sigma}_t}{\lambda_{I,t,1}}$. Hence, $\frac{\tilde{\gamma}_{1t}}{\tilde{\gamma}_{2t}} $ is identified. In addition $(\lambda_{\theta,t+1,1}/\lambda_{I,t,1} ) \kappa_t $ is identified. Now write
	\begin{align*}
		E[\ln \tilde{\theta}_{t+1} \mid \tilde{\theta}_{t},\tilde{I}_{t},Y_t] &=    \frac{\lambda_{\theta,t+1,1} \psi_t}{ {\sigma}_{t}}   \ln \left(\frac{\tilde{\gamma}_{1t}}{ \tilde{\gamma}_{2t} } \tilde{\theta}_{t}^{\frac{ {\sigma}_t}{\lambda_{\theta,t,1}} } +    \tilde{I}_t^{\frac{ {\sigma}_t}{\lambda_{I,t,1}}}\right) +     \frac{\lambda_{\theta,t+1,1} \psi_t}{ {\sigma}_{t}}\ln\tilde{\gamma}_{2t}  + (\lambda_{\theta,t+1,1}/\lambda_{I,t,1} ) \kappa_t  \tilde{ \eta}_{I,t}
	\end{align*} 
	which is linear in $\ln \left(\frac{\tilde{\gamma}_{1t}}{ \tilde{\gamma}_{2t} } \tilde{\theta}_{t}^{\frac{ {\sigma}_t}{\lambda_{\theta,t,1}}} +    \tilde{I}_t^{\frac{ {\sigma}_t}{\lambda_{I,t,1}}}\right) $ and therefore $  \frac{\lambda_{\theta,t+1,1} \psi_t}{ {\sigma}_{t}}  $  and $ \frac{\lambda_{\theta,t+1,1}\psi_t}{ {\sigma}_{t}}\ln\tilde{\gamma}_{2t}$ are identified. Hence $\tilde{\gamma}_{1t}$ and $\tilde{\gamma}_{2t}$ are identified.

	Now suppose we have an alternative set of parameters, denoted by $\{\bar{\mu}_{\theta,t,m},\bar{\lambda}_{\theta,t,m}\}_{t=0,\ldots,T,m=1,2}$, $\{\bar{\mu}_{I,t,m},\bar{\lambda}_{I,t,m}\}_{t=0,\ldots,T-1,m=1,2}$, $\{\bar{a}_t,\bar{\gamma}_{1t},\bar{\gamma}_{2t}, \bar{\gamma}_{3t},\bar{\psi}_t\}^{T-1}_{t=0}$, $\{ \bar{\beta}_{0t}, \bar{\beta}_{1t}, \bar{\beta}_{2t}\}^{T-1}_{t=0}$, and $(\bar{\rho}_0,\bar{\rho}_1)$  that yields the same values of $\{\tilde{\mu}_{\theta,t,m},\tilde{\lambda}_{\theta,t,m}\}_{t=0,\ldots,T,m=1,2}$, $\{\tilde{\mu}_{I,t,m},\tilde{\lambda}_{I,t,m}\}_{t=0,\ldots,T-1,m=1,2}$, $\{ \tilde{\beta}_{0t}, \tilde{\beta}_{1t}, \tilde{\beta}_{2t}\}^{T-1}_{t=0}$, $\{\tilde{\gamma}_{1t},\tilde{\gamma}_{2t}, \frac{ {\sigma}_t}{\lambda_{\theta,t,1}}, \frac{ {\sigma}_t}{\lambda_{I,t,1}},\frac{\lambda_{\theta,t+1,1}\psi_t}{ {\sigma}_t}\}^{T-1}_{t=0}$, and $(\tilde{\rho}_0,\tilde{\rho}_1)$. Define $\bar{\theta}_t$ and $\bar{I}_t$ such that 
	\begin{align*}
		\bar{\mu}_{\theta,t,1} + \bar{\lambda}_{\theta,t,1} \ln \bar{\theta}_t &= {\mu}_{\theta,t,1} + {\lambda}_{\theta,t,1} \ln {\theta}_t = \ln \tilde{\theta}_t\\
		\bar{\mu}_{I,t,1} + \bar{\lambda}_{I,t,1} \ln \bar{I}_t &= {\mu}_{I,t,1} + {\lambda}_{I,t,1} \ln {I}_t = \ln \tilde{I}_t
	\end{align*} 
	For $m \neq 1$, the arguments of the proof of Theorem \ref{th:identparams_tl} imply that 
	\begin{align*}
		\bar{\mu}_{\theta,t,m} + \bar{\lambda}_{\theta,t,m} \ln \bar{\theta}_t  	&=\tilde{\mu}_{\theta,t,m}   + \tilde{\lambda}_{\theta,t,m} \ln \tilde{\theta}_t  \\
		\bar{\mu}_{I,t,m} + \bar{\lambda}_{I,t,m} \ln \bar{I}_t  &=\tilde{\mu}_{I,t,m}   + \tilde{\lambda}_{I,t,m} \ln \tilde{I}_t  
	\end{align*}
	Then the two models generate the same distribution of the measures. In addition, the parameters and $\bar{\theta}_t$ are consistent with the production technology 	because
	\begin{align*}
		\exp(\bar{\mu}_{\theta,t+1,1})     \bar{\theta}_{t+1}^{\bar{\lambda}_{\theta,t+1,1}}  	& = \tilde{\theta}_{t+1} \\
		&=      \left (\tilde{\gamma}_{1t} \tilde{\theta}_{t}^{\frac{ {\sigma}_t}{\lambda_{\theta,t,1}}} +  \tilde{\gamma}_{2t} \tilde{I}_t^{\frac{ {\sigma}_t}{\lambda_{I,t,1}}}\right)^{  \frac{\lambda_{\theta,t+1,1}\psi_t}{ {\sigma}_{t}}} \exp (\tilde{\eta}_{\theta,t} ) \\
		&=   \left (\tilde{\gamma}_{1t}\exp\left(\sigma_t\frac{\bar{\mu}_{\theta,t,1}}{\lambda_{\theta,t,1}} \right)    \bar{\theta}_{t}^{\frac{ \bar{\lambda}_{\theta,t,1} }{\lambda_{\theta,t,1}} {\sigma}_t} +  \tilde{\gamma}_{2t} \exp\left(\sigma_t\frac{\bar{\mu}_{I,t,1}}{\lambda_{I,t,1}} \right)    \bar{I}_t^{\frac{ \bar{\lambda}_{I,t,1}}{\lambda_{I,t,1}} {\sigma}_t}\right)^{  \frac{\lambda_{\theta,t+1,1}\psi_t}{ {\sigma}_{t}}} \exp (\tilde{\eta}_{\theta,t} )  
	\end{align*}
	and thus
	\begin{small}
		\begin{align*}
			\bar{\theta}_{t+1}   
			&= \exp\left(-\frac{\bar{\mu}_{\theta,t+1,1}}{\bar{\lambda}_{\theta,t+1,1}}\right)  \left (\tilde{\gamma}_{1t}\exp\left(\sigma_t\frac{\bar{\mu}_{\theta,t,1}}{\lambda_{\theta,t,1}} \right)    \bar{\theta}_{t}^{\frac{ \bar{\lambda}_{\theta,t,1} }{\lambda_{\theta,t,1}} {\sigma}_t} +  \tilde{\gamma}_{2t} \exp\left(\sigma_t\frac{\bar{\mu}_{I,t,1}}{\lambda_{I,t,1}} \right)    \bar{I}_t^{\frac{ \bar{\lambda}_{I,t,1}}{\lambda_{I,t,1}} {\sigma}_t}\right)^{  \frac{\lambda_{\theta,t+1,1}\psi_t}{  {\bar{\lambda}_{\theta,t+1,1}} {\sigma}_{t}}} \exp (\tilde{\eta}_{\theta,t}/{\bar{\lambda}_{\theta,t+1,1}} )  
		\end{align*}
	\end{small}
	Since $ \bar{\lambda}_{\theta,t,1}/\bar{\sigma}_t = \lambda_{\theta,t,1}/\sigma_t $, $ \bar{\lambda}_{I,t,1}/\bar{\sigma}_t = \lambda_{I,t,1}/\sigma_t $  and $ \bar{\lambda}_{\theta,t+1,1}\bar{\psi}_t/\bar{\sigma}_t = \lambda_{\theta,t+1,1}\psi_t/\sigma_t $ we get
	\begin{small}
		\begin{align*}
			\bar{\theta}_{t+1}   
			&=  \left (  \bar{\gamma}_{1t}           \bar{\theta}_{t}^{ \bar{\sigma}_t} +   \bar{\gamma}_{2t}    \bar{I}_t^{ \bar{\sigma}_{t}  }\right)^{  \frac{\bar{\psi}_t}{ \bar{\sigma}_{t} } } \exp (\bar{\eta}_{\theta,t})  
		\end{align*}
	\end{small}
	with $\bar{\eta}_{\theta,t} = \frac{\lambda_{\theta,t+1,1}}{\bar{\lambda}_{\theta,t+1,1}} {\eta}_{\theta,t}$. 
	
	Analogously, one can show that 
	\begin{align*}
		\ln \bar{I}_t &=  \bar{\beta}_{0t} + \bar{\beta}_{1t} \ln \bar{\theta}_{t}  + \bar{\beta}_{2t} \ln Y_{t}   + \bar{\eta}_{I,t}  \\
		Q &=  \bar{\rho}_{0} + \bar{\rho}_{1} \ln \bar{\theta}_{T} + \bar{\eta}_Q  .
	\end{align*}

	For the second part, notice that identification of $ \frac{ {\sigma}_t}{\lambda_{\theta,t,1}}$, $\frac{ {\sigma}_t}{\lambda_{I,t,1} }$, and $\frac{ {\sigma}_t}{\lambda_{\theta,t+1,1} \psi_t}$ implies identification of $ \frac{\lambda_{I,t,1}}{\lambda_{\theta,t,1}}$. Hence, once $\lambda_{\theta,t,1}$ is fixed, we can identify $\lambda_{I,t,1}$ and $\sigma_t$ for all $t= 1, \ldots, T-1$ and then also $\psi_t$. Using the expression for $\tilde{\gamma}_{1t}$ and $\tilde{\gamma}_{2t}$, it is easy to see that  ${\gamma}_{1t}$ and ${\gamma}_{2t}$ are identified, once $\mu_{\theta,t,1}$ and $\mu_{\theta,t,2}$ are fixed for all $t$. Identification of the remaining parameters follows from arguments as those in the proof of Theorem \ref{th:identparams_tl}.
\end{proof}

\begin{proof}[Proof of Corollary \ref{c:ces_ident} ]
	
	The first part of Theorem \ref{th:identparams_ces} implies that can identify  $ \sigma_t$, $\lambda_{\theta,t,1}$, $ \lambda_{I,t,1} $, $\psi_t$ for all $t$ under Assumptions \ref{a:baseline} and \ref{a:normalization_ces} and either Assumption \ref{a:add_restrictions_ces}(a) or  Assumption \ref{a:add_restrictions_ces}(b).
	
	Now suppose that in addition Assumptions \ref{a:ageinvariant_technology_skills_ces}(a)  and \ref{a:ageinvariant_technology_investment_ces}(a) hold.  The last part of Theorem \ref{th:identparams_ces} immediately implies that all parameters are then point identified.

	Next suppose that in addition to Assumptions \ref{a:baseline}, \ref{a:normalization_ces}, and either \ref{a:add_restrictions_ces}(a) or   \ref{a:add_restrictions_ces}(b)  Assumptions  \ref{a:ageinvariant_technology_skills}(a) and \ref{a:ageinvariant_technology_investment}(b) hold. Then  $\{  \mu_{\theta,t,m},\}_{t=1,\ldots, T}$ is uniquely determined. Moreover, we have  $\tilde{\beta}_{0t} = \mu_{I,t,1} - \tilde{\beta}_{1t}  \mu_{\theta,t,1}$, implying that  $\{  \mu_{I,t,m},\}_{t=1,\ldots, T-1}$ is identified. Once the coefficients in the measurement error equation are identified for one of the measures, the second part of Theorem \ref{th:identparams_ces} implies that all other parameters are identified as well.

	Next suppose that in addition to Assumptions \ref{a:baseline}, \ref{a:normalization_ces}, and either \ref{a:add_restrictions_ces}(a) or   \ref{a:add_restrictions_ces}(b)  Assumptions \ref{a:ageinvariant_technology_skills_ces}(b)  and \ref{a:ageinvariant_technology_investment_ces}(a) hold. Notice that 
	$$  \tilde{\gamma}_{1t}  =  \gamma_{1t}   \exp\left( \sigma_{t} \left( \frac{ \mu_{\theta,t+1,1} }{ \lambda_{\theta,t+1,1} \psi_t } - \frac{\mu_{\theta,t,1}}{\lambda_{\theta,t,1}}  \right) \right) \quad  \text{and} \quad \tilde{\gamma}_{2t}  =  \gamma_{2t} \exp\left( \sigma_{t} \left( \frac{ \mu_{\theta,t+1,1} }{ \lambda_{\theta,t+1,1} \psi_t}   \right) \right)  $$
	and
	\begin{align*}
		1 &= \gamma_{1t} + \gamma_{2t}   =  \exp\left(  - \sigma_{t} \frac{ \mu_{\theta,t+1,1} }{ \lambda_{\theta,t+1,1}\psi_t }\right)  \left( \exp\left( \sigma_{t} \frac{\mu_{\theta,t,1}}{\lambda_{\theta,t,1} }  \right)  \tilde{\gamma}_{1t} + \tilde{\gamma}_{2t}\right)
	\end{align*}
	Since  $\mu_{\theta,0,1} = 0$ and the right hand side is strictly monotone in $\mu_{\theta,1,1}$, we can identify $\mu_{\theta,1,1}$ and then recursively $\mu_{\theta,t,1}$ for all $t$.  Once the parameters in equation (\ref{eq:measurement_eq_ces}) are identified for $m=1$, the second part of Theorem \ref{th:identparams_ces} implies that all other parameters are identified as well.
	
	Finally, suppose that in addition to Assumptions \ref{a:baseline}, \ref{a:normalization_ces}, and either \ref{a:add_restrictions_ces}(a) or   \ref{a:add_restrictions_ces}(b)  Assumptions \ref{a:ageinvariant_technology_skills_ces}(b) and \ref{a:ageinvariant_technology_investment_ces}(b) hold. Notice that $\tilde{\beta}_{0t} = \mu_{I,t,1} - \tilde{\beta}_{1t}  \mu_{\theta,t,1}$
	and 
	\begin{align*}
		1   =  \exp\left(  - \sigma_{t} \frac{ \mu_{\theta,t+1,1} }{ \lambda_{\theta,t+1,1}\psi_t }\right)  \left( \exp\left( \sigma_{t} \frac{\mu_{\theta,t,1}}{\lambda_{\theta,t,1} }  \right)  \tilde{\gamma}_{1t} +    \exp\left( \sigma_{t} \frac{\mu_{I,t,1}}{\lambda_{I,t,1}}  \right) \tilde{\gamma}_{2t}  \right)
	\end{align*}
	Using $\mu_{\theta,0,1} = 0$ and the first equation, we can determine $\mu_{0,I,1}$. Given $\mu_{\theta,0,1} = 0$ and $\mu_{0,I,1}$ and the second equation, we can determine $\mu_{\theta,1,1}$. Then using recursion, we can identify $\mu_{\theta,t,1}$ and  $\mu_{I,t,1}$ for all $t$.  Once the coefficients in equation (\ref{eq:measurement_eq_ces}) are identified for $m=1$, the second part of Theorem \ref{th:identparams_ces} implies that all other parameters are identified as well.
\end{proof}

\begin{proof}[Proof of Theorem \ref{th:identfunctionsces}]
	The first four parts proof is analogous to the proof of Theorem \ref{th:identfunctions}. For the last part notice that  can identify   $\lambda_{\theta,t+1,1}/\lambda_{\theta,t,1} $ and $\lambda_{\theta,t,1}/\lambda_{I,t,1} $ for all $t$ under Assumption \ref{a:baseline} and either Assumption \ref{a:add_restrictions_ces}(a) or  Assumption \ref{a:add_restrictions_ces}(b).	Moreover,
	\begin{align*}
		\frac{\partial \ln \theta_{t+1} }{\partial \ln \theta_{t} } &= \frac{{\lambda}_{\theta,t,1}}{{\lambda}_{\theta,t+1,1}}\frac{\partial \ln \tilde{\theta}_{t+1} }{\partial \ln  \tilde{\theta}_{t} } =\frac{{\lambda}_{\theta,t,1}}{{\lambda}_{\theta,t+1,1}}\frac{\partial }{\partial \ln  \tilde{\theta}_{t} } \ln  \left( \tilde{\gamma}_{1t}    \tilde{\theta}_t^{\frac{\sigma_{t} }{\lambda_{\theta,t,1}}} + \tilde{\gamma}_{2t}  \tilde{I}_t^{\frac{\sigma_{t} }{\lambda_{I,t,1}}}  \right)^{ \frac{\lambda_{\theta,t+1,1}\psi_t}{\sigma_{t}} }  \\
		\frac{\partial \ln \theta_{t+1} }{\partial \ln I_{t} } &= \frac{{\lambda}_{I,t,1}}{{\lambda}_{\theta,t+1,1}}\frac{\partial \ln \tilde{\theta}_{t+1} }{\partial \ln  \tilde{I}_{t} } = \frac{{\lambda}_{I,t,1}}{{\lambda}_{\theta,t ,1}}  \frac{{\lambda}_{\theta,t,1}}{{\lambda}_{\theta,t+1,1}}\frac{\partial }{\partial \ln  \tilde{I}_{t} }  \left( \tilde{\gamma}_{1t}    \tilde{\theta}_t^{\frac{\sigma_{t} }{\lambda_{\theta,t,1}}} + \tilde{\gamma}_{2t}  \tilde{I}_t^{\frac{\sigma_{t} }{\lambda_{I,t,1}}}  \right)^{ \frac{\lambda_{\theta,t+1,1}\psi_t}{\sigma_{t}} }  	
	\end{align*}
	whose distributions are then identified.
\end{proof}

\begin{proof}[Proof of Theorem \ref{th:obseq}]
	
	Let $\bar{\mu}_{\theta,0,1} = 0$, $\bar{\lambda}_{\theta,0,1}= 1$ and $\ln\bar{\theta}_{0}  =  \mu_{\theta,0,1} + \lambda_{\theta,0,1} \ln \theta_{0}$. Let $\bar{\mu}_{\theta,t,1}$ and $\bar{\lambda}_{\theta,t,1}$ and  $\ln \bar{\theta}_{t}$ be such that
	$\bar{\mu}_{\theta,t,1} + \bar{\lambda}_{\theta,t,1} \ln\bar{\theta}_{t} = \mu_{\theta,t,1} + \lambda_{\theta,t,1} \ln \theta_{t}$.
	These values are not unique for $t>0$ and will be determined by the other assumptions. Now let   $\bar{\lambda}_{\theta,t,m} = \bar{\lambda}_{\theta,t,1} \frac{{\lambda}_{\theta,t,m}}{{\lambda}_{\theta,t,1}} $ and  $\bar{\mu}_{\theta,t,m} = \mu_{\theta,t,m} - \frac{\lambda_{\theta,t,m}}{\lambda_{\theta,t,1}}\mu_{\theta,t,1} + \frac{ {\lambda}_{\theta,t,m}}{ {\lambda}_{\theta,t,1}}\bar{\mu}_{\theta,t,1}      $ 
	in which case
	$\bar{\mu}_{\theta,t,m} + \bar{\lambda}_{\theta,t,m} \ln\bar{\theta}_{t} = \mu_{\theta,t,m} + \lambda_{\theta,t,m} \ln \theta_{t}$
	for all $m$.
	
	Next suppose that $\bar{\mu}_{I,t,1} = 0$ and $\bar{\lambda}_{I,t,1} = 1$ for all $t$ and define
	$$   \ln\bar{I}_{t} = \frac{1}{\bar{\lambda}_{I,t,1} }\left(\mu_{I,t,1} - \bar{\mu}_{I,t,1} + \lambda_{I,t,1} \ln I_{t}\right).$$
	and $\bar{\lambda}_{I,t,m} = \bar{\lambda}_{I,t,1} \frac{{\lambda}_{I,t,m}}{{\lambda}_{I,t,1}} $ and  $\bar{\mu}_{I,t,m} = \mu_{I,t,m} - \frac{\lambda_{I,t,m}}{\lambda_{I,t,1}}\mu_{I,t,1} + \frac{ {\lambda}_{I,t,m}}{ {\lambda}_{I,t,1}}\bar{\mu}_{I,t,1}      $ 
	in which case
	$$\bar{\mu}_{I,t,m} + \bar{\lambda}_{I,t,m} \ln\bar{I}_{t} = \mu_{I,t,m} + \lambda_{I,t,m} \ln I_{t}  $$
	for all $m$. Moreover, it is easy to see that there are $(\bar{\beta}_{0t}, \bar{\beta}_{1t},\bar{\beta}_{2t})$ and $\bar{\eta}_{I,t}$ such that
	$$\ln \bar{I}_t  =  \bar{\beta}_{0t} + \bar{\beta}_{1t} \ln\bar{ \theta}_{t}  + \bar{\beta}_{2t} \ln Y_{t}   + \bar{\eta}_{I,t} $$
	
	Instead of assuming $\bar{\mu}_{I,t,1} = 0$ and $\bar{\lambda}_{I,t,1} = 1$, first write
	\begin{align*}
		\ln I_t  &=  \beta_{0t} + \beta_{1t} \left(\frac{\bar{\mu}_{\theta,t,1} - \mu_{\theta,t,1}}{\lambda_{\theta,t,1}} + \frac{\bar{\lambda}_{\theta,t,1}}{\lambda_{\theta,t,1}} \ln\bar{\theta}_{t}\right)  + \beta_{2t} \ln Y_{t}   + \eta_{I,t} \\
		&=  \beta_{0t} + \beta_{1t} \left(\frac{\bar{\mu}_{\theta,t,1} - \mu_{\theta,t,1}}{\lambda_{\theta,t,1}}\right)   +\beta_{1t}  \frac{\bar{\lambda}_{\theta,t,1}}{\lambda_{\theta,t,1}} \ln\bar{\theta}_{t}   + \beta_{2t} \ln Y_{t}   + \eta_{I,t} 
	\end{align*}
	which implies that
	$
	\ln \bar{I}_t
	=  \bar{\beta}_1   \ln\bar{\theta}_{t}   +  \bar{\beta}_2   \ln Y_{t}   + \bar{\eta}_{I,t}  $
	with 
	$$\ln \bar{I}_t = \frac{\ln I_t -  \left(\beta_{0t} + \beta_{1t} \left(\frac{\bar{\mu}_{\theta,t,1} - \mu_{\theta,t,1}}{\lambda_{\theta,t,1}}\right) \right)  }{  \beta_{1t}  \frac{\bar{\lambda}_{\theta,t,1}}{\lambda_{\theta,t,1}} +  \beta_{2t}},$$
	$$\bar{\beta}_{1t}	= \frac{ \beta_{1t}  \frac{\bar{\lambda}_{\theta,t,1}}{\lambda_{\theta,t,1}}}{  \beta_{1t}  \frac{\bar{\lambda}_{\theta,t,1}}{\lambda_{\theta,t,1}} +  \beta_{2t}} \qquad \text{ and } \qquad \bar{\beta}_{2t}	=  	\frac{\beta_{2t}}{  \beta_{1t}  \frac{\bar{\lambda}_{\theta,t,1}}{\lambda_{\theta,t,1}} +  \beta_{2t}} $$
	It is now easy to see that there exist $\bar{\mu}_{I,t,1}$ and $\bar{\lambda}_{I,t,1}$ such that
	$$\bar{\mu}_{I,t,1} + \bar{\lambda}_{I,t,1} \ln\bar{I}_{t} = \mu_{I,t,1} + \lambda_{I,t,1} \ln I_{t}.$$
	Then define and $\bar{\lambda}_{I,t,m} = \bar{\lambda}_{I,t,1} \frac{{\lambda}_{I,t,m}}{{\lambda}_{I,t,1}} $ and  $\bar{\mu}_{I,t,m} = \mu_{I,t,m} - \frac{\lambda_{I,t,m}}{\lambda_{I,t,1}}\mu_{I,t,1} + \frac{ {\lambda}_{I,t,m}}{ {\lambda}_{I,t,1}}\bar{\mu}_{I,t,1}      $ 
	in which case for all $m$
	$$\bar{\mu}_{I,t,m} + \bar{\lambda}_{I,t,m} \ln\bar{I}_{t} = \mu_{I,t,m} + \lambda_{I,t,m} \ln I_{t} .$$

	In both cases, there are parameters that are consistent with the second, third, and forth equation of the model. In addition, we now have known constants $\bar{\mu}_{I,t,1}$ and $\bar{\lambda}_{I,t,1}$ and a random variable $\ln \bar{I}_{t}$ such that 
	$\bar{\mu}_{I,t,1} + \bar{\lambda}_{I,t,1} \ln\bar{\theta}_{t} = \mu_{I,t,1} + \lambda_{I,t,1} \ln I_{t}$.
	
	We next show that there are production function parameters consistent with Assumption \ref{a:ageinvariant_technology_skills}(b). To do so, write 
	\begin{align*}
		\ln \theta_{t+1}  &= a_t + \gamma_{1t} \ln \theta_{t} +  \gamma_{2t} \ln I_t +  \gamma_{3t} \ln \theta_{t}  \ln I_t + \eta_{\theta,t} \\
		& = a_t + \gamma_{1t} \left(\frac{\bar{\mu}_{\theta,t,1} - \mu_{\theta,t,1}}{\lambda_{\theta,t,1}} + \frac{\bar{\lambda}_{\theta,t,1}}{\lambda_{\theta,t,1}} \ln\bar{\theta}_{t}\right)  +  \gamma_{2t} \left(\frac{\bar{\mu}_{I,t,1} - \mu_{I,t,1}}{\lambda_{I,t,1}} + \frac{\bar{\lambda}_{I,t,1}}{\lambda_{I,t,1}} \ln\bar{I}_{t}\right) \\
		& \qquad  +  \gamma_{3t} \left(\frac{\bar{\mu}_{\theta,t,1} - \mu_{\theta,t,1}}{\lambda_{\theta,t,1}} + \frac{\bar{\lambda}_{\theta,t,1}}{\lambda_{\theta,t,1}} \ln\bar{\theta}_{t}\right)  \left(\frac{\bar{\mu}_{I,t,1} - \mu_{I,t,1}}{\lambda_{I,t,1}} + \frac{\bar{\lambda}_{I,t,1}}{\lambda_{I,t,1}} \ln\bar{I}_{t}\right) + \eta_{\theta,t}		\\
		& = a_t + \gamma_{1t} \frac{\bar{\mu}_{\theta,t,1} - \mu_{\theta,t,1}}{\lambda_{\theta,t,1}}  +  \gamma_{2t}  \frac{\bar{\mu}_{I,t,1} -  \mu_{I,t,1}}{\lambda_{I,t,1}} + \gamma_{3t}\frac{\bar{\mu}_{\theta,t,1} - \mu_{\theta,t,1}}{\lambda_{\theta,t,1}}\frac{\bar{\mu}_{I,t,1} - \mu_{I,t,1}}{\lambda_{I,t,1}}  \\
		& \qquad  + \left(\gamma_{1t}   \frac{\bar{\lambda}_{\theta,t,1}}{\lambda_{\theta,t,1}}    + \gamma_{3t}  \frac{\bar{\mu}_{I,t,1} - \mu_{I,t,1}}{\lambda_{I,t,1}}  \frac{\bar{\lambda}_{\theta,t,1}}{\lambda_{\theta,t,1}} \right)\ln\bar{\theta}_{t}    + \left(\gamma_{2t}  \frac{\bar{\lambda}_{I,t,1}}{\lambda_{I,t,1}}  + \gamma_{3t} \frac{\bar{\mu}_{\theta,t,1} - \mu_{\theta,t,1}}{\lambda_{\theta,t,1}}\frac{\bar{\lambda}_{I,t,1}}{\lambda_{I,t,1}}\right) \ln\bar{I}_{t} \\
		& \qquad  +  \gamma_{3t}  \frac{\bar{\lambda}_{\theta,t,1}}{\lambda_{\theta,t,1}}     \frac{\bar{\lambda}_{I,t,1}}{\lambda_{I,t,1}} \ln\bar{\theta}_{t}  \ln\bar{I}_{t}  + \eta_{\theta,t}	 
	\end{align*}
	which we can write as
	\begin{align*}
		\ln\bar{\theta}_{t+1}  
		&   = \bar{\gamma}_{1t}  \ln\bar{\theta}_{t}    + \bar{\gamma}_{2t}   \ln\bar{I}_{t}   + \bar{\gamma}_{3t}  \ln\bar{\theta}_{t}  \ln\bar{I}_{t}  + \bar{\eta}_{\theta,t}	 
	\end{align*}
	where
	$$ \ln\bar{\theta}_{t+1}  = \frac{ \ln \theta_{t+1}  
		- \left(a_t + \gamma_{1t} \frac{\bar{\mu}_{\theta,t,1} - \mu_{\theta,t,1}}{\lambda_{\theta,t,1}} + \gamma_{2t}  \frac{\bar{\mu}_{I,t,1} -  \mu_{I,t,1}}{\lambda_{I,t,1}} + \gamma_{3t}\frac{\bar{\mu}_{\theta,t,1} - \mu_{\theta,t,1}}{\lambda_{\theta,t,1}}\frac{\bar{\mu}_{I,t,1} - \mu_{I,t,1}}{\lambda_{I,t,1}} \right)}{ \gamma_{1t}   \frac{\bar{\lambda}_{\theta,t,1}}{\lambda_{\theta,t,1}}    + \gamma_{3t}  \frac{\bar{\mu}_{I,t,1} - \mu_{I,t,1}}{\lambda_{I,t,1}}  \frac{\bar{\lambda}_{\theta,t,1}}{\lambda_{\theta,t,1}}   +  \gamma_{2t}  \frac{\bar{\lambda}_{I,t,1}}{\lambda_{I,t,1}}  + \gamma_{3t} \frac{\bar{\mu}_{\theta,t,1} - \mu_{\theta,t,1}}{\lambda_{\theta,t,1}}\frac{\bar{\lambda}_{I,t,1}}{\lambda_{I,t,1}}  + + \gamma_{3t}  \frac{\bar{\lambda}_{\theta,t,1}}{\lambda_{\theta,t,1}}     \frac{\bar{\lambda}_{I,t,1}}{\lambda_{I,t,1}} }$$
	$$\bar{\gamma}_{1t} = \frac{ \gamma_{1t}   \frac{\bar{\lambda}_{\theta,t,1}}{\lambda_{\theta,t,1}}    + \gamma_{3t}  \frac{\bar{\mu}_{I,t,1} - \mu_{I,t,1}}{\lambda_{I,t,1}}  \frac{\bar{\lambda}_{\theta,t,1}}{\lambda_{\theta,t,1}}  }{ \gamma_{1t}   \frac{\bar{\lambda}_{\theta,t,1}}{\lambda_{\theta,t,1}}    + \gamma_{3t}  \frac{\bar{\mu}_{I,t,1} - \mu_{I,t,1}}{\lambda_{I,t,1}}  \frac{\bar{\lambda}_{\theta,t,1}}{\lambda_{\theta,t,1}}   +  \gamma_{2t}  \frac{\bar{\lambda}_{I,t,1}}{\lambda_{I,t,1}}  + \gamma_{3t} \frac{\bar{\mu}_{\theta,t,1} - \mu_{\theta,t,1}}{\lambda_{\theta,t,1}}\frac{\bar{\lambda}_{I,t,1}}{\lambda_{I,t,1}}  + + \gamma_{3t}  \frac{\bar{\lambda}_{\theta,t,1}}{\lambda_{\theta,t,1}}     \frac{\bar{\lambda}_{I,t,1}}{\lambda_{I,t,1}} }$$
	$$\bar{\gamma}_{2t} = \frac{ \gamma_{2t}  \frac{\bar{\lambda}_{I,t,1}}{\lambda_{I,t,1}}  + \gamma_{3t} \frac{\bar{\mu}_{\theta,t,1} - \mu_{\theta,t,1}}{\lambda_{\theta,t,1}}\frac{\bar{\lambda}_{I,t,1}}{\lambda_{I,t,1}} }{ \gamma_{1t}   \frac{\bar{\lambda}_{\theta,t,1}}{\lambda_{\theta,t,1}}    + \gamma_{3t}  \frac{\bar{\mu}_{I,t,1} - \mu_{I,t,1}}{\lambda_{I,t,1}}  \frac{\bar{\lambda}_{\theta,t,1}}{\lambda_{\theta,t,1}}   +  \gamma_{2t}  \frac{\bar{\lambda}_{I,t,1}}{\lambda_{I,t,1}}  + \gamma_{3t} \frac{\bar{\mu}_{\theta,t,1} - \mu_{\theta,t,1}}{\lambda_{\theta,t,1}}\frac{\bar{\lambda}_{I,t,1}}{\lambda_{I,t,1}}  + + \gamma_{3t}  \frac{\bar{\lambda}_{\theta,t,1}}{\lambda_{\theta,t,1}}     \frac{\bar{\lambda}_{I,t,1}}{\lambda_{I,t,1}} }$$
	$$\bar{\gamma}_{3t} = \frac{ \gamma_{3t}  \frac{\bar{\lambda}_{\theta,t,1}}{\lambda_{\theta,t,1}}     \frac{\bar{\lambda}_{I,t,1}}{\lambda_{I,t,1}}}{ \gamma_{1t}   \frac{\bar{\lambda}_{\theta,t,1}}{\lambda_{\theta,t,1}}    + \gamma_{3t}  \frac{\bar{\mu}_{I,t,1} - \mu_{I,t,1}}{\lambda_{I,t,1}}  \frac{\bar{\lambda}_{\theta,t,1}}{\lambda_{\theta,t,1}}   +  \gamma_{2t}  \frac{\bar{\lambda}_{I,t,1}}{\lambda_{I,t,1}}  + \gamma_{3t} \frac{\bar{\mu}_{\theta,t,1} - \mu_{\theta,t,1}}{\lambda_{\theta,t,1}}\frac{\bar{\lambda}_{I,t,1}}{\lambda_{I,t,1}}  + + \gamma_{3t}  \frac{\bar{\lambda}_{\theta,t,1}}{\lambda_{\theta,t,1}}     \frac{\bar{\lambda}_{I,t,1}}{\lambda_{I,t,1}} }$$
	It is now easy to see that there exist $\bar{\mu}_{\theta,t+1,1}$ and $\bar{\lambda}_{\theta,t+1,1}$ such that
	$$\bar{\mu}_{\theta,t+1,1} + \bar{\lambda}_{\theta,t+1,1} \ln\bar{\theta}_{t+1} = \mu_{\theta,t+1,1} + \lambda_{\theta,t+1,1} \ln \theta_{t+1}.$$
	
	We can now use the arguments recursively and show always exist sets of parameters for the first four equations that are consistent with the data and satisfy Assumptions \ref{a:baseline}, \ref{a:normalization}, \ref{a:ageinvariant_technology_skills}(b), and either \ref{a:ageinvariant_technology_investment}(a) or \ref{a:ageinvariant_technology_investment}(b). It also immediately follows that the exist $\bar{\rho}_0$ and $\bar{\rho}_1$ such that 
	$$ \bar{\rho}_{0} + \bar{\rho}_{1} \ln \theta_{T} = \rho_{0} + \rho_{1} \ln \theta_{T}  $$
	
	The arguments in Section \ref{s:identification_tl} imply that there always exist sets of parameters  that are consistent with the data and satisfy  that Assumptions \ref{a:baseline}, \ref{a:normalization}, either \ref{a:ageinvariant_technology_skills}(a) and  \ref{a:ageinvariant_technology_investment}(a).
	
	Finally, suppose Assumptions \ref{a:baseline}, \ref{a:normalization}, either \ref{a:ageinvariant_technology_skills}(a) and  \ref{a:ageinvariant_technology_investment}(b) hold. Let 
	$  \ln\bar{\theta}_{t} = \mu_{\theta,t,1} + \lambda_{\theta,t,1} \ln \theta_{t}$. Then we can write $
	\ln \bar{I}_t
	=  \bar{\beta}_1   \ln\bar{\theta}_{t}   +  \bar{\beta}_2   \ln Y_{t}   + \bar{\eta}_{I,t}  
	$
	with 
	$$\ln \bar{I}_t = \frac{\ln I_t -  \left(\beta_{0t} + \beta_{1t} \left(\frac{  \mu_{\theta,t,1}}{\lambda_{\theta,t,1}}\right) \right)  }{  \beta_{1t}  \frac{1}{\lambda_{\theta,t,1}} +  \beta_{2t}}, \quad  \bar{\beta}_{1t}	= \frac{ \beta_{1t}  \frac{1}{\lambda_{\theta,t,1}}}{  \beta_{1t}  \frac{1}{\lambda_{\theta,t,1}} +  \beta_{2t}} \quad \text{ and } \quad \bar{\beta}_{2t}	=  	\frac{\beta_{2t}}{  \beta_{1t}  \frac{1}{\lambda_{\theta,t,1}} +  \beta_{2t}}. $$
	It is now easy to see that there exist $\bar{\mu}_{I,t,1}$ and $\bar{\lambda}_{I,t,1}$ such that
	$$\bar{\mu}_{I,t,1} + \bar{\lambda}_{I,t,1} \ln\bar{I}_{t} = \mu_{I,t,1} + \lambda_{I,t,1} \ln I_{t}.$$
	Using the previous arguments, we can then write 
	\begin{align*}
		&\ln \bar{\theta}_{t+1}  \\
		& = \mu_{\theta,t+1,1} + \lambda_{\theta,t+1,1} \left(a_t + \gamma_{1t} \frac{\bar{\mu}_{\theta,t,1} - \mu_{\theta,t,1}}{\lambda_{\theta,t,1}}  +  \gamma_{2t}  \frac{\bar{\mu}_{I,t,1} -  \mu_{I,t,1}}{\lambda_{I,t,1}} + \gamma_{3t}\frac{\bar{\mu}_{\theta,t,1} - \mu_{\theta,t,1}}{\lambda_{\theta,t,1}}\frac{\bar{\mu}_{I,t,1} - \mu_{I,t,1}}{\lambda_{I,t,1}} \right) \\
		& \quad  +  \lambda_{\theta,t+1,1}\left(\gamma_{1t}   \frac{\bar{\lambda}_{\theta,t,1}}{\lambda_{\theta,t,1}}    + \gamma_{3t}  \frac{\bar{\mu}_{I,t,1} - \mu_{I,t,1}}{\lambda_{I,t,1}}  \frac{\bar{\lambda}_{\theta,t,1}}{\lambda_{\theta,t,1}} \right)\ln\bar{\theta}_{t}    +  \lambda_{\theta,t+1,1}\left(\gamma_{2t}  \frac{\bar{\lambda}_{I,t,1}}{\lambda_{I,t,1}}  + \gamma_{3t} \frac{\bar{\mu}_{\theta,t,1} - \mu_{\theta,t,1}}{\lambda_{\theta,t,1}}\frac{\bar{\lambda}_{I,t,1}}{\lambda_{I,t,1}}\right) \ln\bar{I}_{t} \\
		& \quad  +   \lambda_{\theta,t+1,1} \gamma_{3t}  \frac{\bar{\lambda}_{\theta,t,1}}{\lambda_{\theta,t,1}}     \frac{\bar{\lambda}_{I,t,1}}{\lambda_{I,t,1}} \ln\bar{\theta}_{t}  \ln\bar{I}_{t}  +  \lambda_{\theta,t+1,1}\eta_{\theta,t}	 
	\end{align*}
	Setting  $\bar{\mu}_{\theta,t,1} = 0$ and $\bar{\lambda}_{\theta,t,1} = 1$, it is easy to that there are $\bar{a}_t$, $ \bar{\gamma}_{1t}$, $ \bar{\gamma}_{2t}$, and $ \bar{\gamma}_{3t}$ such that
	such that 
	$\ln \bar{\theta}_{t+1}   = \bar{a}_t + \bar{\gamma}_{1t} \ln \bar{\theta}_{t} +  \bar{\gamma}_{2t} \ln \bar{I}_t +  \bar{\gamma}_{3t} \ln \bar{\theta}_{t}  \ln \bar{I}_t + \bar{\eta}_{\theta,t}$.
\end{proof}

\begin{proof}[Proof of Theorem \ref{th:obseq_ces}]

	As in the proof of Theorem \ref{th:obseq}, there are constants $\bar{\mu}_{\theta,t,1}$ and $\bar{\lambda}_{\theta,t,1}$ and a random variable $\ln \bar{\theta}_{t}$ such that 
	$\bar{\mu}_{\theta,t,1} + \bar{\lambda}_{\theta,t,1} \ln\bar{\theta}_{t} = \mu_{\theta,t,1} + \lambda_{\theta,t,1} \ln \theta_{t}$
	and $\bar{\mu}_{\theta,0,1} = 0$ and $\bar{\lambda}_{\theta,0,1} = 1$. Then there exist $\bar{\mu}_{\theta,t,m}$ and $\bar{\lambda}_{\theta,t,m}$ such that  for all $m$
	$$\bar{\mu}_{\theta,t,m} + \bar{\lambda}_{\theta,t,m} \ln\bar{\theta}_{t} = \mu_{\theta,t,m} + \lambda_{\theta,t,m} \ln \theta_{t}$$

	Next consider the production function and write  
	\begin{align*}
		\theta_{t+1} &=  \left(\gamma_{1t} \theta_{t}^{\sigma_t} +   \gamma_{2t} I_t^{\sigma_t}\right)^{\frac{\psi_t}{\sigma_t}}\exp(\eta_{\theta,t})  \\
		&=  \left(\gamma_{1t}\exp\left( {\sigma_t} \frac{\bar{\mu}_{\theta,t,1} - \mu_{\theta,t,1}}{\lambda_{\theta,t,1}} \right)    \bar{\theta}_{t}^{\sigma_t \frac{\bar{\lambda}_{\theta,t,1}}{\lambda_{\theta,t,1}} } +   \gamma_{2t} I_t^{\sigma_t}\right)^{\frac{\psi_t}{\sigma_t}}\exp(\eta_{\theta,t})  \\
		&=  \left(\gamma_{1t}\exp\left( {\sigma_t} \frac{\bar{\mu}_{\theta,t,1} - \mu_{\theta,t,1}}{\lambda_{\theta,t,1}} \right)    \bar{\theta}_{t}^{\sigma_t \frac{\bar{\lambda}_{\theta,t,1}}{\lambda_{\theta,t,1}} } +   \gamma_{2t} \left(I_t^{\frac{\lambda_{\theta,t,1}}{\bar{\lambda}_{\theta,t,1}}}\right)^{\sigma_t \frac{\bar{\lambda}_{\theta,t,1}}{\lambda_{\theta,t,1}} }\right)^{\frac{\psi_t}{\sigma_t}}\exp(\eta_{\theta,t})  
	\end{align*}
	and with $\bar{\sigma}_t = \sigma_t \frac{\bar{\lambda}_{\theta,t,1}}{\lambda_{\theta,t,1}} $  we have
	\begin{align*}
		\theta_{t+1}^{\frac{\lambda_{\theta,t,1}}{\bar{\lambda}_{\theta,t,1}}}
		&=  \left(\gamma_{1t}\exp\left( {\sigma_t} \frac{\bar{\mu}_{\theta,t,1} - \mu_{\theta,t,1}}{\lambda_{\theta,t,1}} \right)    \bar{\theta}_{t}^{\bar{\sigma}_t } +   \gamma_{2t} \left(I_t^{\frac{\lambda_{\theta,t,1}}{\bar{\lambda}_{\theta,t,1}}}\right)^{\bar{\sigma}_t}\right)^{\frac{\psi_t}{\bar{\sigma}_t } }\exp\left(\frac{\lambda_{\theta,t,1}}{\bar{\lambda}_{\theta,t,1}}\eta_{\theta,t}\right)  
	\end{align*}
	We also know that we need to satisfy the relationship
	$$\exp\left(\frac{\bar{\mu}_{I,t,1} - \mu_{I,t,1}}{\lambda_{I,t,1}} + \frac{\bar{\lambda}_{I,t,1}}{\lambda_{I,t,1}} \ln\bar{I}_{t} \right) =  I_{t}$$
	Hence, 
	\begin{align*}
		\theta_{t+1}^{\frac{\lambda_{\theta,t,1}}{\bar{\lambda}_{\theta,t,1}}} &= 
		\left(\gamma_{1t}\exp\left( {\sigma_t} \frac{\bar{\mu}_{\theta,t,1} - \mu_{\theta,t,1}}{\lambda_{\theta,t,1}} \right)    \bar{\theta}_{t}^{\bar{\sigma}_t } +   \gamma_{2t} \exp\left(\bar{\sigma}_t\frac{\lambda_{\theta,t,1}}{\bar{\lambda}_{\theta,t,1}}\frac{\bar{\mu}_{I,t,1} - \mu_{I,t,1}}{\lambda_{I,t,1}}\right)  \bar{I}_{t}^{\bar{\sigma}_t\frac{\lambda_{\theta,t,1}}{\bar{\lambda}_{\theta,t,1}}\frac{\bar{\lambda}_{I,t,1}}{\lambda_{I,t,1}}}\right)^{\frac{\psi_t}{\bar{\sigma}_t } }\exp\left(\frac{\lambda_{\theta,t,1}}{\bar{\lambda}_{\theta,t,1}}\eta_{\theta,t}\right)  
	\end{align*}
	Since 
	$$\theta_{t+1}^{\frac{\lambda_{\theta,t,1}}{\bar{\lambda}_{\theta,t,1}}} = \exp\left(\frac{\lambda_{\theta,t,1}}{\bar{\lambda}_{\theta,t,1}}\frac{\bar{\mu}_{\theta,t+1,1} - \mu_{\theta,t+1,1}}{\lambda_{\theta,t+1,1}}\right) \bar{\theta}_{t+1}^{\frac{\lambda_{\theta,t,1}}{\bar{\lambda}_{\theta,t,1}}\frac{\bar{\lambda}_{\theta,t+1,1}}{\lambda_{\theta,t+1,1}}} $$
	we then write
	\begin{align*}
		\bar{\theta}_{t+1}&= 
		\left(\frac{\gamma_{1t}\exp\left( {\sigma_t} \frac{\bar{\mu}_{\theta,t,1} - \mu_{\theta,t,1}}{\lambda_{\theta,t,1}} \right)}{\exp\left(\bar{\sigma}_t\frac{\lambda_{\theta,t,1}}{\bar{\lambda}_{\theta,t,1}}\frac{\bar{\mu}_{\theta,t+1,1} - \mu_{\theta,t+1,1}}{\lambda_{\theta,t+1,1}}\right)}    \bar{\theta}_{t}^{\bar{\sigma}_t } +  \frac{ \gamma_{2t} \exp\left(\bar{\sigma}_t\frac{\lambda_{\theta,t,1}}{\bar{\lambda}_{\theta,t,1}}\frac{\bar{\mu}_{I,t,1} - \mu_{I,t,1}}{\lambda_{I,t,1}}\right) }{\exp\left(\bar{\sigma}_t\frac{\lambda_{\theta,t,1}}{\bar{\lambda}_{\theta,t,1}}\frac{\bar{\mu}_{\theta,t+1,1} - \mu_{\theta,t+1,1}}{\lambda_{\theta,t+1,1}}\right)} \bar{I}_{t}^{\bar{\sigma}_t\frac{\lambda_{\theta,t,1}}{\bar{\lambda}_{\theta,t,1}}\frac{\bar{\lambda}_{I,t,1}}{\lambda_{I,t,1}}}\right)^{\frac{\bar{\psi}_{t}}{\bar{\sigma}_t }  }\exp\left(\frac{\lambda_{\theta,t,1}}{\bar{\lambda}_{\theta,t,1}}\eta_{\theta,t}\right)  \\
		&= 
		\left(  \bar{\gamma}_{1t} \bar{\theta}_{t}^{\bar{\sigma}_t } +  \bar{\gamma}_{2t} \bar{I}_{t}^{\bar{\sigma}_t\frac{\lambda_{\theta,t,1}}{\bar{\lambda}_{\theta,t,1}}\frac{\bar{\lambda}_{I,t,1}}{\lambda_{I,t,1}}}\right)^{\frac{\bar{\psi}_{t}}{\bar{\sigma}_t }  }\exp\left(\frac{\lambda_{\theta,t,1}}{\bar{\lambda}_{\theta,t,1}}\eta_{\theta,t}\right)  
	\end{align*}
	where
	$$\bar{\gamma}_{1t} = \frac{\gamma_{1t}\exp\left( {\sigma_t} \frac{\bar{\mu}_{\theta,t,1} - \mu_{\theta,t,1}}{\lambda_{\theta,t,1}} \right)}{\exp\left(\bar{\sigma}_t\frac{\lambda_{\theta,t,1}}{\bar{\lambda}_{\theta,t,1}}\frac{\bar{\mu}_{\theta,t+1,1} - \mu_{\theta,t+1,1}}{\lambda_{\theta,t+1,1}}\right)} \quad \text{ and } \quad  \bar{\gamma}_{2t} = \frac{ \gamma_{2t} \exp\left(\bar{\sigma}_t\frac{\lambda_{\theta,t,1}}{\bar{\lambda}_{\theta,t,1}}\frac{\bar{\mu}_{I,t,1} - \mu_{I,t,1}}{\lambda_{I,t,1}}\right) }{\exp\left(\bar{\sigma}_t\frac{\lambda_{\theta,t,1}}{\bar{\lambda}_{\theta,t,1}}\frac{\bar{\mu}_{\theta,t+1,1} - \mu_{\theta,t+1,1}}{\lambda_{\theta,t+1,1}}\right)}$$
	and $\bar{\psi}_{t} = \psi_{t} {\frac{\bar{\lambda}_{\theta,t,1}}{\lambda_{\theta,t,1}}\frac{\lambda_{\theta,t+1,1}}{\bar{\lambda}_{\theta,t+1,1}}} =  \psi_{t} {\frac{\bar{\sigma}_{t}}{{\sigma}_{t}}\frac{{\sigma}_{t+1}}{ \bar{\sigma}_{t+1}  } }$.
	
	We now show that under different combinations of assumptions, there are parameters consistent with the model. First notice that it has to hold that
	$\bar{\lambda}_{I,t,1} = \frac{\bar{\lambda}_{\theta,t,1}}{\lambda_{\theta,t,1}} \lambda_{I,t,1}$. Now write
	\begin{align*}
		\left(\frac{\bar{\mu}_{I,t,1} - \mu_{I,t,1}}{\lambda_{I,t,1}} + \frac{\bar{\lambda}_{I,t,1}}{\lambda_{I,t,1}} \ln\bar{I}_{t}\right)   &=  \beta_{0t} + \beta_{1t} \left(\frac{\bar{\mu}_{\theta,t,1} - \mu_{\theta,t,1}}{\lambda_{\theta,t,1}} + \frac{\bar{\lambda}_{\theta,t,1}}{\lambda_{\theta,t,1}} \ln\bar{\theta}_{t}\right)  + \beta_{2t} \ln Y_{t}   + \eta_{I,t}  
	\end{align*}
	which implies that
	\begin{align*}
		\ln\bar{I}_{t}    &=  \beta_{0t}\frac{\lambda_{I,t,1}}{\bar{\lambda}_{I,t,1}} - \frac{\bar{\mu}_{I,t,1} - \mu_{I,t,1}}{\lambda_{I,t,1}}\frac{\lambda_{I,t,1}}{\bar{\lambda}_{I,t,1}}  + \beta_{1t}\frac{\lambda_{I,t,1}}{\bar{\lambda}_{I,t,1}} \frac{\bar{\mu}_{\theta,t,1} - \mu_{\theta,t,1}}{\lambda_{\theta,t,1}} + \beta_{1t} \ln\bar{\theta}_{t}  + \beta_{2t} \frac{\lambda_{I,t,1}}{\bar{\lambda}_{I,t,1}}\ln Y_{t}   + \eta_{I,t}  
	\end{align*}
	It follows that $\bar{\beta}_1 = \beta_{1}$ and $\bar{\beta}_2 = \beta_{2t} \frac{\lambda_{I,t,1}}{\bar{\lambda}_{I,t,1}}$ are determined and
	$$\bar{\beta}_{0t} = \beta_{0t}\frac{\lambda_{I,t,1}}{\bar{\lambda}_{I,t,1}} - \frac{\bar{\mu}_{I,t,1} - \mu_{I,t,1}}{\lambda_{I,t,1}}\frac{\lambda_{I,t,1}}{\bar{\lambda}_{I,t,1}}  + \beta_{1t}\frac{\lambda_{I,t,1}}{\bar{\lambda}_{I,t,1}} \frac{\bar{\mu}_{\theta,t,1} - \mu_{\theta,t,1}}{\lambda_{\theta,t,1}}$$
	If $\bar{\lambda}_{\theta,t,1} = \bar{\lambda}_{\theta,t+1,1} = 1$, set  $\bar{\psi}_{t} = \psi_{t} \frac{\lambda_{\theta,t+1,1}}{\lambda_{\theta,t,1}} $ and $\bar{\sigma}_t = \sigma_t \frac{1}{\lambda_{\theta,t,1}} $. If instead $\bar{\psi}_{t} = 1$, take  $\bar{\lambda}_{\theta,t+1,1} = \psi_{t} {\frac{\bar{\lambda}_{\theta,t,1}}{\lambda_{\theta,t,1}}{\lambda_{\theta,t+1,1}}}$  and  $\bar{\sigma}_t = \sigma_t \frac{\bar{\lambda}_{\theta,t,1}}{\lambda_{\theta,t,1}} $. Either way, $\bar{\sigma}_t $, $\bar{\psi}_{t}$, and $\bar{\lambda}_{\theta,t,1}$  are uniquely determined.
	
	Now suppose that $ \bar{\mu}_{\theta,t,1} = \bar{\mu}_{I,t,1}= 0$. Then   $\bar{\beta}_{0t}$ is uniquely determined and so is
	$$\bar{\gamma}_{1t} = \frac{\gamma_{1t}\exp\left( {\sigma_t} \frac{\bar{\mu}_{\theta,t,1} - \mu_{\theta,t,1}}{\lambda_{\theta,t,1}} \right)}{\exp\left(\bar{\sigma}_t\frac{\lambda_{\theta,t,1}}{\bar{\lambda}_{\theta,t,1}}\frac{\bar{\mu}_{\theta,t+1,1} - \mu_{\theta,t+1,1}}{\lambda_{\theta,t+1,1}}\right)}  \quad \text{and} \quad \bar{\gamma}_{2t} =  \frac{ \gamma_{2t} \exp\left(\bar{\sigma}_t\frac{\lambda_{\theta,t,1}}{\bar{\lambda}_{\theta,t,1}}\frac{\bar{\mu}_{I,t,1} - \mu_{I,t,1}}{\lambda_{I,t,1}}\right) }{\exp\left(\bar{\sigma}_t\frac{\lambda_{\theta,t,1}}{\bar{\lambda}_{\theta,t,1}}\frac{\bar{\mu}_{\theta,t+1,1} - \mu_{\theta,t+1,1}}{\lambda_{\theta,t+1,1}}\right)}. $$
	Hence, we found parameters that are consistent with the first four equations of the model. 
	
	Next suppose that $ \bar{\mu}_{\theta,t,1}=0$  and $ \bar{\beta}_{0t}= 0$. Then also $\bar{\mu}_{\theta,t+1,1}  = 0$. Again let
	$$\bar{\gamma}_{1t} = \frac{\gamma_{1t}\exp\left( {\sigma_t} \frac{\bar{\mu}_{\theta,t,1} - \mu_{\theta,t,1}}{\lambda_{\theta,t,1}} \right)}{\exp\left(\bar{\sigma}_t\frac{\lambda_{\theta,t,1}}{\bar{\lambda}_{\theta,t,1}}\frac{\bar{\mu}_{\theta,t+1,1} - \mu_{\theta,t+1,1}}{\lambda_{\theta,t+1,1}}\right)} $$ 
	Then we can find $ \bar{\mu}_{I,t,1}$ such that $ \bar{\beta}_{0t}= 0$. Finally let 
	$$ \bar{\gamma}_{2t} =  \frac{ \gamma_{2t} \exp\left(\bar{\sigma}_t\frac{\lambda_{\theta,t,1}}{\bar{\lambda}_{\theta,t,1}}\frac{\bar{\mu}_{I,t,1} - \mu_{I,t,1}}{\lambda_{I,t,1}}\right) }{\exp\left(\bar{\sigma}_t\frac{\lambda_{\theta,t,1}}{\bar{\lambda}_{\theta,t,1}}\frac{\bar{\mu}_{\theta,t+1,1} - \mu_{\theta,t+1,1}}{\lambda_{\theta,t+1,1}}\right)}. $$
	Hence, we found parameters that are consistent with the first four equations of the model.
	
	Next suppose that $ \bar{\mu}_{I,t,1}= 0$ and $\gamma_{1t}+\gamma_{2t}=1$. Then  $\bar{\beta}_{0t} =\beta_{0t}\frac{\lambda_{I,t,1}}{\bar{\lambda}_{I,t,1}}$ is uniquely determined. It then has to hold that
	$$\bar{\gamma}_{1t} = \frac{\gamma_{1t}\exp\left( {\sigma_t} \frac{\bar{\mu}_{\theta,t,1} - \mu_{\theta,t,1}}{\lambda_{\theta,t,1}} \right)}{\exp\left(\bar{\sigma}_t\frac{\lambda_{\theta,t,1}}{\bar{\lambda}_{\theta,t,1}}\frac{\bar{\mu}_{\theta,t+1,1} - \mu_{\theta,t+1,1}}{\lambda_{\theta,t+1,1}}\right)}  \quad \text{and} \quad \bar{\gamma}_{2t} =  \frac{ \gamma_{2t} \exp\left(\bar{\sigma}_t\frac{\lambda_{\theta,t,1}}{\bar{\lambda}_{\theta,t,1}}\frac{\bar{\mu}_{I,t,1} - \mu_{I,t,1}}{\lambda_{I,t,1}}\right) }{\exp\left(\bar{\sigma}_t\frac{\lambda_{\theta,t,1}}{\bar{\lambda}_{\theta,t,1}}\frac{\bar{\mu}_{\theta,t+1,1} - \mu_{\theta,t+1,1}}{\lambda_{\theta,t+1,1}}\right)}. $$
	and $\bar{\gamma}_{1t} +\bar{\gamma}_{2t} =1$. Since the denominator of the fractions is strictly monotone in $\bar{\mu}_{\theta,t+1,1}$ and has range $(0,\infty)$, there exists a unique value of  $\bar{\mu}_{\theta,t+1,1}$  such that $\bar{\gamma}_{1t} +\bar{\gamma}_{2t} =1$. We then found  parameters that are consistent with the first four equations of the model.

	Finally, suppose that $\bar{\beta}_{0t}= 0$ and $\gamma_{1t}+\gamma_{2t}=1$.  Then we can find a unique value $\bar{\mu}_{I,t,1}= 0$ such that $\bar{\beta}_{0t}= 0$. Given this value, it has to hold that
	$$\bar{\gamma}_{1t} = \frac{\gamma_{1t}\exp\left( {\sigma_t} \frac{\bar{\mu}_{\theta,t,1} - \mu_{\theta,t,1}}{\lambda_{\theta,t,1}} \right)}{\exp\left(\bar{\sigma}_t\frac{\lambda_{\theta,t,1}}{\bar{\lambda}_{\theta,t,1}}\frac{\bar{\mu}_{\theta,t+1,1} - \mu_{\theta,t+1,1}}{\lambda_{\theta,t+1,1}}\right)}  \quad \text{and} \quad \bar{\gamma}_{2t} =  \frac{ \gamma_{2t} \exp\left(\bar{\sigma}_t\frac{\lambda_{\theta,t,1}}{\bar{\lambda}_{\theta,t,1}}\frac{\bar{\mu}_{I,t,1} - \mu_{I,t,1}}{\lambda_{I,t,1}}\right) }{\exp\left(\bar{\sigma}_t\frac{\lambda_{\theta,t,1}}{\bar{\lambda}_{\theta,t,1}}\frac{\bar{\mu}_{\theta,t+1,1} - \mu_{\theta,t+1,1}}{\lambda_{\theta,t+1,1}}\right)}. $$
	and $\bar{\gamma}_{1t} +\bar{\gamma}_{2t} =1$. Again, we there exists a unique value of  $\bar{\mu}_{\theta,t+1,1}$  such that $\bar{\gamma}_{1t} +\bar{\gamma}_{2t} =1$. We then found  parameters that are consistent with the first four equations of the model. 
	
	In all four cases, it also immediately follows that the exist $\bar{\rho}_0$ and $\bar{\rho}_1$ such that 
	$ \bar{\rho}_{0} + \bar{\rho}_{1} \ln \theta_{T} = \rho_{0} + \rho_{1} \ln \theta_{T}  $.
\end{proof}

\begin{proof}[Proof of Theorem \ref{th:indentgen}]

	The joint distribution of 
	$$\{\{Z_{\theta,t,m}\}_{t=0,\ldots,T, m= 1,2,3}, \{Z_{I,t,m}\}_{t=0,\ldots,T -1, m= 1,2,3}, Q,  \{\tilde{g}_{\theta,t}(\theta_t)\}^{T}_{t=0}, \{\tilde{g}_{I,t}(I_t)\}^{T-1}_{t=0}  \}$$ 
	is identified by Theorem 2 of \shortciteN{CHS:10} up to unknown and strictly increasing functions $\tilde{g}_{\theta,t}$ and $\tilde{g}_{I,t}$ conditional on $\{Y_{1}, \ldots, Y_{T-1}\}$.  Moreover, let 
	$\tilde{\eta}_{I,t} = F_{\eta_{I,t}}(\eta_{I,t}) \sim U[0,1]$.	Then we can write 
	$I_t =  h_t(\theta_{t},Y_{t}, F^{-1}_{\eta_{I,t}}(\tilde{\eta}_{I,t}))$
	and thus,
	\begin{align*}
		\tilde{g}_{I,t}(I_t) &=  	\tilde{g}_{I,t}(h_t(\tilde{g}_{\theta,t}^{-1}(\tilde{g}_{\theta,t}(\theta_t)),Y_{t}, F^{-1}_{\eta_{I,t}}(\tilde{\eta}_{I,t}) ))
	\end{align*}
	or
	$
	\tilde{g}_{I,t}(I_t) = \tilde{h}_t( \tilde{g}_{\theta,t}(\theta_t) ,Y_{t},\tilde{\eta}_{I,t})
	$.		Now notice that
	$$Q_{\alpha}(\tilde{g}_{I,t}(I_t) \mid  \tilde{g}_{\theta,t}(\theta_t),Y_{t} ) = \tilde{h}_t( \tilde{g}_{\theta,t}(\theta_t) ,Y_{t},\alpha)$$
	which implies that $\tilde{h}_t$ is identified. We can now write
	$\tilde{\eta}_{I,t} =  \tilde{h}^{-1}_t(\tilde{g}_{I,t}(I_t),\tilde{g}_{\theta,t}(\theta_t) ,Y_{t} )$
	and therefore the joint distribution of $(\tilde{g}_{I,t}(I_t),\tilde{g}_{\theta,t}(\theta_t) ,Y_{t},\tilde{\eta}_{I,t})$ is point identified.

	Now write
	$$\theta_{t+1} = f_t(\theta_{t},I_{t},\eta_{I,t}, \varsigma_{\theta,t} )    \Leftrightarrow   \tilde{g}_{\theta,t+1}(\theta_{t+1}) = \tilde{g}_{\theta,t+1}(f_t( \tilde{g}_{\theta,t}^{-1}(\tilde{g}_{ \theta,t}(\theta_{t})) , \tilde{g}_{I,t}^{-1}(\tilde{g}_{I,t}(I_{t})),F^{-1}_{\eta_{I,t}}(\tilde{\eta}_{I,t}),   \varsigma_{\theta,t} )  .$$
	We can therefore identify 
	\begin{eqnarray*} 
		&& Q_{{\alpha_4}}(\tilde{g}_{\theta,t+1}(\theta_{t+1}) \mid \tilde{g}_{\theta,t}(\theta_{t}) = Q_{\alpha_1}(\tilde{g}_{\theta,t}(\theta_{t})), \tilde{g}_{I,t}(I_{t}) = Q_{\alpha_2}(\tilde{g}_{I,t}(I_{t})), \tilde{\eta}_{I,t} = \alpha_3  ) \\
		&&= \tilde{g}_{\theta,t+1}(f_t( \tilde{g}_{\theta,t}^{-1}(Q_{\alpha_1}(\tilde{g}_{\theta,t}(\theta_{t}) ) , \tilde{g}_{I,t}^{-1}( Q_{\alpha_2}(\tilde{g}_{I,t}(I_{t}) ) , F^{-1}_{\eta_{I,t}}(\alpha_3 ),    Q_{\alpha_4}(\varsigma_{\theta,t}) )) \\
		&&= \tilde{g}_{\theta,t+1}(f_t( Q_{\alpha_1}(\theta_{t})  ,   Q_{\alpha_2} (I_{t}  ) ,  Q_{\alpha_3}(\eta_{I,t}),    Q_{\alpha_4}(\varsigma_{\theta,t}) )) 
	\end{eqnarray*}
	and 
	\begin{eqnarray*}
		&&F_{ \tilde{g}_{\theta,t+1}(\theta_{t+1}) }( \tilde{g}_{\theta,t+1}(f_t( Q_{\alpha_1}(\theta_{t})  ,   Q_{\alpha_2} (I_{t}  ) ,  Q_{\alpha_3}(\eta_{I,t}),    Q_{\alpha_4}(\varsigma_{\theta,t}) ))  )  \\ 
		&&\qquad  = F_{ \theta_{t+1} }(  f_t( Q_{\alpha_1}(\theta_{t})  ,   Q_{\alpha_2} (I_{t}  ) ,  Q_{\alpha_3}(\eta_{I,t}),    Q_{\alpha_4}(\varsigma_{\theta,t}) ) )   
	\end{eqnarray*}
	Similarly, we can write $Z_{\theta,t,m} = g_{\theta,t,m}(\tilde{g}_{\theta,t}^{-1}(\tilde{g}_{\theta,t}({\theta}_{t})),F^{-1}_{\eps_{\theta,t,m}}(\tilde{\eps}_{\theta,t,m})) = \tilde{g}_{\theta,t,m}(\tilde{g}_{\theta,t}({\theta}_{t}),\tilde{\eps}_{\theta,t,m})$, where $\tilde{\eps}_{\theta,t,m} \sim U[0,1]$, and the joint distribution of $(Z_{\theta,t,m},\tilde{g}_{\theta,t}({\theta}_{t}),\tilde{\eps}_{\theta,t,m})$ and the function  $\tilde{g}_{\theta,t,m}$ is identified for all $t$. Hence, we can identify
	\begin{eqnarray*} 
		&& Q_{{\alpha_5}}(Z_{\theta,t+1,m}  \mid \tilde{g}_{\theta,t+1}(\theta_{t+1}) ) =\tilde{g}_{\theta,t+1}(f_t( Q_{\alpha_1}(\theta_{t})  ,   Q_{\alpha_2} (I_{t}  ) ,  Q_{\alpha_3}(\eta_{I,t}),    Q_{\alpha_4}(\varsigma_{\theta,t}) ))  ) \\
		&&= 
		g_{\theta,t+1,m}(f_t( Q_{\alpha_1}(\theta_{t})  ,   Q_{\alpha_2} (I_{t}  ) ,  Q_{\alpha_3}(\eta_{I,t}),    Q_{\alpha_4}(\varsigma_{\theta,t}),Q_{\alpha_5}(\eps_{\theta,t+1,m})) 
	\end{eqnarray*}

	For the second part notice that
	\begin{align*}
		\tilde{g}_{I,t}(I_t(y))  &=  	\tilde{g}_{I,t}(h_t(\tilde{g}_{\theta,t+1}^{-1}(Q_{\alpha_1}(\tilde{g}_{\theta,t}(\theta_t))),y_{t}, F^{-1}_{\eta_{I,t}}(\alpha_3) )) =  	 \tilde{h}_t(  Q_{\alpha_1}(\tilde{g}_{\theta,t}(\theta_t)) ,y_{t},  \alpha_3  ) 
	\end{align*}
	is identified by the previous arguments.  	We can therefore identify 
	\begin{eqnarray*} 
		&& Q_{\alpha_4}(\tilde{g}_{\theta,t+1}(\theta_{t+1}) \mid \tilde{g}_{\theta,t}(\theta_{t}) = Q_{\alpha_1}(\tilde{g}_{\theta,t}(\theta_{t})), \tilde{g}_{I,t}(I_{t}) = \tilde{g}_{I,t}(I_t(y)) , \tilde{\eta}_{I,t} = \alpha_3  ) \\
		&&= \tilde{g}_{\theta,t+1}(f_t( \tilde{g}_{\theta,t}^{-1}(Q_{\alpha_1}(\tilde{g}_{\theta,t}(\theta_{t}) ) , \tilde{g}_{I,t}^{-1}( \tilde{g}_{I,t}(I_t(Y))  ) , F^{-1}_{\eta_{I,t}}(\alpha_3 ),    Q_{\alpha_4}(\varsigma_{\theta,t}) )) \\
		&&= \tilde{g}_{\theta,t+1}(f_t( Q_{\alpha_1}(\theta_{t})  ,  I_t(y)   , Q_{\alpha_2}(\eta_{I,t}),    Q_{\alpha_4}(\varsigma_{\theta,t}) )) 
	\end{eqnarray*}
	and
	\begin{align*}
		& F_{\tilde{g}_{\theta,t+1}(\theta_{t+1})  }  ( \tilde{g}_{\theta,t+1}(f_t( Q_{\alpha_1}(\theta_{t})  ,  I_t(y)   , Q_{\alpha_2}(\eta_{I,t}),    Q_{\alpha_4}(\varsigma_{\theta,t}) ))  ) \\
		& \qquad = F_{ \theta_{t+1}   }  ( f_t( Q_{\alpha_1}(\theta_{t})  ,  I_t(y)   , Q_{\alpha_2}(\eta_{I,t}),    Q_{\alpha_4}(\varsigma_{\theta,t}) ))      
	\end{align*}
	Moreover, 
	\begin{eqnarray*} 
		&& Q_{{\alpha_5}}(Z_{\theta,t+1,m} \mid \tilde{g}_{\theta,t+1}(\theta_{t+1}) )=\tilde{g}_{\theta,t+1}(f_t( Q_{\alpha_1}(\theta_{t})  ,  I_t(y)   , Q_{\alpha_2}(\eta_{I,t}),    Q_{\alpha_4}(\varsigma_{\theta,t}) ))   ) \\
		&&= 
		g_{\theta,t+1,m}(f_t( Q_{\alpha_1}(\theta_{t})  ,   Q_{\alpha_2} (I_{t}  ) ,  Q_{\alpha_3}(\eta_{I,t}),     I_t(y)),Q_{\alpha_5}(\eps_{\theta,t+1,m})).
	\end{eqnarray*}
	
	For the third part notice that 
	$$\tilde{g}_{\theta,t+1}(\theta_{t+1}) = \tilde{g}_{\theta,t+1}(f_t( \tilde{g}_{\theta,t}^{-1}(\tilde{g}_{\theta,t}(\theta_{t})) , \tilde{g}_{I,t}^{-1}(\tilde{g}_{I,t}(I_{t})),F^{-1}_{\eta_{I,t}}(\tilde{\eta}_{I,t}),   F_{\varsigma_{\theta,t}}^{-1}(\tilde{\varsigma}_{\theta,t}) ))  .$$
	with $\tilde{\varsigma}_{\theta,t} = F_{\varsigma_{\theta,t}}( {\varsigma}_{\theta,t}) \sim U[0,1]$. We can therefore write $\tilde{\varsigma}_{\theta,t}$ as an identified function of the random vector $(\tilde{g}_{\theta,t+1}(\theta_{t+1}), \tilde{g}_{\theta,t}(\theta_{t}), \tilde{g}_{I,t}(I_{t}), \tilde{\eta}_{I,t} )  $. Hence, the joint distribution of 
	$$(Q, \tilde{g}_{\theta,t}(\theta_{t}), \{\tilde{g}_{I,t}(I_{t})    \}^{T-1}_{s=t}, \{\tilde{\eta}_{I,s}\}^{T-1}_{s=t} , \{\tilde{\varsigma}_{\theta,s}   \}^{T-1}_{s=t}  )$$ 
	is identified. Finally,
	\begin{small}
		\begin{align*}
			&  P\left(Q \leq q   \mid \theta_t= Q_{\alpha_1}(\theta_t), \{I_s = Q_{\alpha_{2s}}(I_s) \}^{T-1}_{s=t}, \{\eta_{I,s} = Q_{\alpha_{3s}}(\eta_{I,s}) \}^{T-1}_{s=t} , \{\varsigma_{\theta,s} = Q_{\alpha_{4s}}(\varsigma_{\theta,s})   \}^{T-1}_{s=t}  \right) \\
			& =  P\left(Q \leq q   \mid \tilde{\theta}_t= Q_{\alpha_1}(\tilde{\theta}_t), \{\tilde{g}_{I,t}(I_{t}) = Q_{\alpha_{2s}}(\tilde{g}_{I,t}(I_{t}) ) \}^{T-1}_{s=t}, \{\tilde{\eta}_{I,s} = Q_{\alpha_{3s}}(\tilde{\eta}_{I,s}) \}^{T-1}_{s=t} , \{\tilde{\varsigma}_{\theta,s} = Q_{\alpha_{4s}}(\tilde{\varsigma}_{\theta,s})   \}^{T-1}_{s=t} \right).
		\end{align*}
	\end{small}
	
	The fourth part follows from identification of the distribution of $(Q,\tilde{g}_{\theta,t}(\theta_t), \{Y_s\}^{T-1}_{s=t})$.\end{proof}

\setstretch{1}

\bibliography{references}

\begin{thebibliography}{}

\bibitem[\protect\citeauthoryear{Adhvaryu, Nyshadham, and Tamayo}{Adhvaryu
  et~al.}{2023}]{ANT:23}
Adhvaryu, A., A.~Nyshadham, and J.~Tamayo (2023).
\newblock {Managerial Quality and Productivity Dynamics}.
\newblock {\em The Review of Economic Studies\/}~{\em 90\/}(4), 1569--1607.

\bibitem[\protect\citeauthoryear{Agostinelli and Wiswall}{Agostinelli and
  Wiswall}{2016a}]{AW:16a}
Agostinelli, F. and M.~Wiswall (2016a).
\newblock Estimating the technology of children’s skill formation.
\newblock NBER Working Paper No. 22442.

\bibitem[\protect\citeauthoryear{Agostinelli and Wiswall}{Agostinelli and
  Wiswall}{2016b}]{AW:16b}
Agostinelli, F. and M.~Wiswall (2016b).
\newblock Identification of dynamic latent factor models: The implications of
  re-normalization in a model of child development.
\newblock NBER Working Paper No. 22441.

\bibitem[\protect\citeauthoryear{Agostinelli and Wiswall}{Agostinelli and
  Wiswall}{2024}]{AW:22}
Agostinelli, F. and M.~Wiswall (2024).
\newblock Estimating the technology of children’s skill formation.
\newblock Accepted for publication at the \textit{Journal of Political
  Economy}.

\bibitem[\protect\citeauthoryear{Aguirregabiria and Suzuki}{Aguirregabiria and
  Suzuki}{2014}]{AS:14}
Aguirregabiria, V. and J.~Suzuki (2014).
\newblock Identification and counterfactuals in dynamic models of market entry
  and exit.
\newblock {\em Quantitative Marketing and Economics\/}~{\em 12}, 267--304.

\bibitem[\protect\citeauthoryear{Anderson and Rubin}{Anderson and
  Rubin}{1956}]{AR:56}
Anderson, T.~W. and H.~Rubin (1956).
\newblock Statistical inference in factor analysis.
\newblock In {\em Proceedings of the Third Berkeley Symposium on Mathematical
  Statistics and Probability, Volume 5: Contributions to Econometrics,
  Industrial Research, and Psychometry}, Berkeley, Calif., pp.\  111--150.
  University of California Press.

\bibitem[\protect\citeauthoryear{Attanasio, Cattan, Fitzsimons, Meghir, and
  Rubio-Codin}{Attanasio et~al.}{2020}]{ACDMR:19}
Attanasio, O., S.~Cattan, E.~Fitzsimons, C.~Meghir, and M.~Rubio-Codin (2020).
\newblock Estimating the production function for human capital: Results from a
  randomized control trial in {C}olombia.
\newblock {\em American Economic Review\/}~{\em 110\/}(1), 48--85.

\bibitem[\protect\citeauthoryear{Attanasio, Cattan, and Meghir}{Attanasio
  et~al.}{2022}]{ACM:2022}
Attanasio, O., S.~Cattan, and C.~Meghir (2022).
\newblock Early childhood development, human capital, and poverty.
\newblock {\em Annual Review of Economics\/}~{\em 14}, 853--892.

\bibitem[\protect\citeauthoryear{Attanasio, Meghir, and Nix}{Attanasio
  et~al.}{2020}]{AMN:19}
Attanasio, O., C.~Meghir, and E.~Nix (2020).
\newblock {Human Capital Development and Parental Investment in India}.
\newblock {\em The Review of Economic Studies\/}~{\em 87\/}(6), 2511--2541.

\bibitem[\protect\citeauthoryear{Attanasio, Meghir, Nix, and Salvati}{Attanasio
  et~al.}{2017}]{AMNS:17}
Attanasio, O., C.~Meghir, E.~Nix, and F.~Salvati (2017).
\newblock Human capital growth and poverty: Evidence from {E}thiopia and
  {P}eru.
\newblock {\em Review of Economic Dynamics\/}~{\em 25}, 234--259.

\bibitem[\protect\citeauthoryear{Aucejo and James}{Aucejo and
  James}{2021}]{AJ:16}
Aucejo, E. and J.~James (2021).
\newblock The path to college education: The role of math and verbal skills.
\newblock {\em Journal of Political Economy\/}~{\em 129\/}(10), 2905--2946.

\bibitem[\protect\citeauthoryear{Bolt, French, Hentall-MacCuish, and
  O'Dea}{Bolt et~al.}{2024}]{BFHD:24}
Bolt, U., E.~French, J.~Hentall-MacCuish, and C.~O'Dea (2024).
\newblock The intergenerational elasticity of earnings: Exploring the
  mechanisms.
\newblock Working paper.

\bibitem[\protect\citeauthoryear{Caucutt and Lochner}{Caucutt and
  Lochner}{2020}]{CL:20}
Caucutt, E.~M. and L.~Lochner (2020).
\newblock Early and late human capital investments, borrowing constraints, and
  the family.
\newblock {\em Journal of Political Economy\/}~{\em 128\/}(3), 1065--1147.

\bibitem[\protect\citeauthoryear{Chiappori, Komunjer, and Kristensen}{Chiappori
  et~al.}{2015}]{CKK:15}
Chiappori, P.-A., I.~Komunjer, and D.~Kristensen (2015).
\newblock Nonparametric identification and estimation of transformation models.
\newblock {\em Journal of Econometrics\/}~{\em 188\/}(1), 22 -- 39.

\bibitem[\protect\citeauthoryear{Cunha}{Cunha}{2011}]{Cunha:11}
Cunha, F. (2011).
\newblock Recent developments in the identification and estimation of
  production functions of skills.
\newblock {\em Fiscal Studies\/}~{\em 32\/}(2), 297--316.

\bibitem[\protect\citeauthoryear{Cunha and Heckman}{Cunha and
  Heckman}{2007}]{CH:07}
Cunha, F. and J.~Heckman (2007).
\newblock The technology of skill formation.
\newblock {\em American Economic Review\/}~{\em 97\/}(2), 31--41.

\bibitem[\protect\citeauthoryear{Cunha and Heckman}{Cunha and
  Heckman}{2008}]{CH:08}
Cunha, F. and J.~Heckman (2008).
\newblock Formulating, identifying and estimating the technology of cognitive
  and noncognitive skill formation.
\newblock {\em Journal of Human Resources\/}~{\em 43\/}(4), 738--782.

\bibitem[\protect\citeauthoryear{Cunha and Heckman}{Cunha and
  Heckman}{2009}]{CH:09}
Cunha, F. and J.~Heckman (2009).
\newblock The economics and psychology of inequality and human development.
\newblock {\em Journal of the European Economic Association\/}~{\em 7\/}(2),
  320--364.

\bibitem[\protect\citeauthoryear{Cunha, Heckman, and Schennach}{Cunha
  et~al.}{2010}]{CHS:10}
Cunha, F., J.~Heckman, and S.~Schennach (2010).
\newblock Estimating the technology of cognitive and noncognitive skill
  formation.
\newblock {\em Econometrica\/}~{\em 78\/}(3), 883--931.

\bibitem[\protect\citeauthoryear{Del~Boca, Flinn, and Wiswall}{Del~Boca
  et~al.}{2013}]{DFW:13}
Del~Boca, D., C.~Flinn, and M.~Wiswall (2013, 10).
\newblock Household choices and child development.
\newblock {\em The Review of Economic Studies\/}~{\em 81\/}(1), 137--185.

\bibitem[\protect\citeauthoryear{Del~Bono, Kinsler, and Pavan}{Del~Bono
  et~al.}{2022}]{DKP:20}
Del~Bono, E., J.~Kinsler, and R.~Pavan (2022).
\newblock Identification of dynamic latent factor models of skill formation
  with translog production.
\newblock {\em Journal of Applied Econometrics\/}~{\em 37\/}(6), 1256--1265.

\bibitem[\protect\citeauthoryear{Embrey}{Embrey}{2019}]{Embrey:19}
Embrey, I. (2019).
\newblock On the benefits of normalization in production functions.
\newblock Lancaster Economics Department Working Paper No. 2019/0049.

\bibitem[\protect\citeauthoryear{Evdokimov and White}{Evdokimov and
  White}{2012}]{EW:12}
Evdokimov, K. and H.~White (2012).
\newblock Some extensions of a lemma of kotlarski.
\newblock {\em Econometric Theory\/}~{\em 28\/}(4), 925--932.

\bibitem[\protect\citeauthoryear{Fiorini and Keane}{Fiorini and
  Keane}{2014}]{FK:14}
Fiorini, M. and M.~Keane (2014).
\newblock How the allocation of children’s time affects cognitive and
  noncognitive development.
\newblock {\em Journal of Labor Economics\/}~{\em 32\/}(4), 787--836.

\bibitem[\protect\citeauthoryear{Freyberger}{Freyberger}{2018}]{Freyberger:18}
Freyberger, J. (2018).
\newblock Nonparametric panel data models with interactive fixed effects.
\newblock {\em Review of Economic Studies\/}~{\em 85\/}(3), 1824--1851.

\bibitem[\protect\citeauthoryear{Gallipoli and Gomez}{Gallipoli and
  Gomez}{2023}]{GG:23}
Gallipoli, G. and S.~Gomez (2023).
\newblock The production of financial literacy.
\newblock Working paper.

\bibitem[\protect\citeauthoryear{Gao and Li}{Gao and Li}{2021}]{GL:19}
Gao, W. and M.~Li (2021).
\newblock Robust semiparametric estimation in panel multinomial choice models.
\newblock SSRN Working Paper No. 3282293.

\bibitem[\protect\citeauthoryear{Hamilton, Waggoner, and Zha}{Hamilton
  et~al.}{2007}]{HWZ:07}
Hamilton, J., D.~Waggoner, and T.~Zha (2007).
\newblock Normalization in econometrics.
\newblock {\em Econometric Reviews\/}~{\em 26\/}(2--4), 221--252.

\bibitem[\protect\citeauthoryear{Heckman, Pinto, and Savelyev}{Heckman
  et~al.}{2013}]{HPS:13}
Heckman, J., R.~Pinto, and P.~Savelyev (2013, October).
\newblock Understanding the mechanisms through which an influential early
  childhood program boosted adult outcomes.
\newblock {\em American Economic Review\/}~{\em 103\/}(6), 2052--86.

\bibitem[\protect\citeauthoryear{Helmers and Patnam}{Helmers and
  Patnam}{2011}]{HP:11}
Helmers, C. and M.~Patnam (2011).
\newblock The formation and evolution of childhood skill acquisition: Evidence
  from {I}ndia.
\newblock {\em Journal of Development Economics\/}~{\em 95\/}(2), 252--266.

\bibitem[\protect\citeauthoryear{Hern{\'a}ndez-Alava and
  Popli}{Hern{\'a}ndez-Alava and Popli}{2017}]{HAP:17}
Hern{\'a}ndez-Alava, M. and G.~Popli (2017, Apr).
\newblock Children's development and parental input: Evidence from the {U}{K}
  millennium cohort study.
\newblock {\em Demography\/}~{\em 54\/}(2), 485--511.

\bibitem[\protect\citeauthoryear{Kalouptsidi, Scott, and
  Souza-Rodrigues}{Kalouptsidi et~al.}{2021}]{KSS:20}
Kalouptsidi, M., P.~T. Scott, and E.~Souza-Rodrigues (2021).
\newblock Identification of counterfactuals in dynamic discrete choice models.
\newblock {\em Quantitative Economics\/}~{\em 12\/}(2), 351--403.

\bibitem[\protect\citeauthoryear{Klump and Grandville}{Klump and
  Grandville}{2000}]{KG:00}
Klump, R. and O.~D.~L. Grandville (2000).
\newblock Economic growth and the elasticity of substitution: Two theorems and
  some suggestions.
\newblock {\em American Economic Review\/}~{\em 90\/}(1), 282--291.

\bibitem[\protect\citeauthoryear{Klump, McAdam, and Willman}{Klump
  et~al.}{2012}]{KMW:12}
Klump, R., P.~McAdam, and A.~Willman (2012).
\newblock The normalized ces production function: theory and empirics.
\newblock {\em Journal of Economic Surveys\/}~{\em 26\/}(5), 769--799.

\bibitem[\protect\citeauthoryear{Komarova, Sanches, Silva~Junior, and
  Srisuma}{Komarova et~al.}{2018}]{KSSS:18}
Komarova, T., F.~Sanches, D.~Silva~Junior, and S.~Srisuma (2018).
\newblock Joint analysis of the discount factor and payoff parameters in
  dynamic discrete choice models.
\newblock {\em Quantitative Economics\/}~{\em 9\/}(3), 1153--1194.

\bibitem[\protect\citeauthoryear{Lewbel}{Lewbel}{2019}]{Lewbel:19}
Lewbel, A. (2019).
\newblock The identification zoo: Meanings of identification in econometrics.
\newblock {\em Journal of Economic Literature\/}~{\em 57\/}(4).

\bibitem[\protect\citeauthoryear{Madansky}{Madansky}{1964}]{Madansky:64}
Madansky, A. (1964).
\newblock Instrumental variables in factor analysis.
\newblock {\em Psychometrika\/}~{\em 29\/}(2), 105--113.

\bibitem[\protect\citeauthoryear{Matzkin}{Matzkin}{1994}]{Matzkin:94}
Matzkin, R.~L. (1994).
\newblock Restrictions of economic theory in nonparametric methods.
\newblock Volume~4 of {\em Handbook of Econometrics}, Chapter~42, pp.\
  2523--2558. Elsevier.

\bibitem[\protect\citeauthoryear{Matzkin}{Matzkin}{2007}]{Matzkin:07}
Matzkin, R.~L. (2007).
\newblock Nonparametric identification.
\newblock Volume~6 of {\em Handbook of Econometrics}, Chapter~73, pp.\
  5307--5368. Elsevier.

\bibitem[\protect\citeauthoryear{Murphy and Topel}{Murphy and
  Topel}{2016}]{MT:16}
Murphy, K.~M. and R.~H. Topel (2016).
\newblock Human capital investment, inequality, and economic growth.
\newblock {\em Journal of Labor Economics\/}~{\em 34\/}(2), 99--127.

\bibitem[\protect\citeauthoryear{Rubio-Ramírez, Waggoner, and
  Zha}{Rubio-Ramírez et~al.}{2010}]{RWZ:10}
Rubio-Ramírez, J.~F., D.~F. Waggoner, and T.~Zha (2010).
\newblock Structural vector autoregressions: Theory of identification and
  algorithms for inference.
\newblock {\em The Review of Economic Studies\/}~{\em 77\/}(2), 665--696.

\bibitem[\protect\citeauthoryear{Temple}{Temple}{2012}]{Temple:12}
Temple, J. (2012).
\newblock The calibration of ces production functions.
\newblock {\em Journal of Macroeconomics\/}~{\em 34\/}(2), 294--303.

\bibitem[\protect\citeauthoryear{Williams}{Williams}{2020}]{Williams:20}
Williams, B. (2020).
\newblock Identification of the linear factor model.
\newblock {\em Econometric Reviews\/}~{\em 39\/}(1), 92--109.

\end{thebibliography}

\end{document}